\documentclass{vldb}

\vldbTitle{Joins on Samples: A Theoretical Guide for Practitioners}
\vldbAuthors{Dawei Huang, Dong Young Yoon, Seth Pettie, and Barzan Mozafari}
\vldbDOI{https://doi.org/10.14778/3372716.3372726}
\vldbVolume{13}
\vldbNumber{4}
\vldbYear{2019}

\usepackage{amsthm}
\usepackage{thm-restate}
\usepackage{thmtools}
\usepackage{zi4}

\usepackage{times}

\usepackage{amssymb}
\usepackage{amsmath}
\usepackage{graphicx,color}
\usepackage{amsmath}
\usepackage{mathtools}

\usepackage[inline]{enumitem}
\usepackage{epsfig}
\usepackage{minibox}
\usepackage{bm}

\usepackage{algorithm}
\usepackage[algo2e]{algorithm2e}

\usepackage{xcolor}
\usepackage{xspace}
\usepackage{balance}

\usepackage{relsize}
\usepackage{float}
\usepackage{subcaption}
\usepackage[nocompress]{cite}
\usepackage{hyperref}

\newtheorem{theorem}{Theorem}
\newtheorem{lemma}[theorem]{Lemma}
\newtheorem{corollary}[theorem]{Corollary}
\newtheorem{definition}[theorem]{Definition}
\newtheorem{claim}[theorem]{Claim}

\newtheorem*{remark}{Remark}

\newcommand{\ignore}[1]{}

\newcommand{\tofix}[1]{{\color{red} #1}}
\newcommand{\rev}[1]{{\color{black} #1}}
\newcommand{\tempcut}[1]{{\color{black} #1}}
\newcommand{\dawei}[1]{{\color{blue}\it Dawei: #1}}

\newcommand{\barzan}[1]{{\color{blue}\it BM: #1}}
\newcommand{\cancut}[1]{{\color{black}#1}}

\newenvironment{enumerate2}
{\begin{enumerate}\setlength{\itemsep}{-0.05cm}\vspace{-0.1cm}}
{\vspace{-0.05cm}\end{enumerate}}

\def\case #1{\medskip{\bf Case}\ {\it #1:}}

\def\twofigs #1{\hbox to \textwidth{#1}}

\def\yclaim #1 {\medskip{\bf Claim #1: }}
\long\def\xclaim #1 {\medskip\noindent{\bf Claim } {\it #1}\medskip}
\def\claim #1 #2 {\medskip\noindent{\bf Claim #1.} {\em #2}\medskip}
\newcommand{\inv}[0]{\vspace*{-0.2cm}}
\newcommand{\sinv}[0]{\vspace*{-0.1cm}}

\def\eg{e.g.}
\def\ie{i.e.}

\newcommand{\E}{\mathrm{E}}
\newcommand{\var}{\mathrm{Var}}
\newcommand{\cov}{\mathrm{Cov}}

\DeclareMathOperator*{\argmax}{arg\,max}
\DeclareMathOperator*{\argmin}{arg\,min}
\allowdisplaybreaks

\newcommand{\UBS}{\mbox{UBS}\xspace}
\newcommand{\SUBS}{\mbox{SUBS}\xspace}

\newcommand{\pmin}{p_{\min}}
\newcommand{\COUNT}[0]{\texttt{COUNT}\xspace}
\newcommand{\SUM}[0]{\texttt{SUM}\xspace}
\newcommand{\AVG}[0]{\texttt{AVG}\xspace}
\newcommand{\VAR}[0]{\texttt{VAR}\xspace}
\newcommand{\STDEV}[0]{\texttt{STDEV}\xspace}
\newcommand{\PERCENTILE}[0]{\texttt{PERCENTILE}\xspace}
\newcommand{\WHERE}[0]{\texttt{WHERE}\xspace}
\newcommand{\OPT}[0]{\texttt{OPT}\xspace}
\newcommand{\TLS}[0]{\texttt{2LV}\xspace}
\newcommand{\SSUF}[0]{\textit{SS\_UF}\xspace}
\newcommand{\partyone}[0]{\texttt{party1}\xspace}
\newcommand{\partytwo}[0]{\texttt{party2}\xspace}

\newcommand{\instacart}[0]{\texttt{Instacart}\xspace}
\newcommand{\movielens}[0]{\texttt{Movielens}\xspace}
\newcommand{\tpch}[0]{\texttt{TPC-H}\xspace}
\newcommand{\synthetic}[0]{\texttt{Synthetic}\xspace}
\newcommand{\syn}[2]{\textit{\texttt{S}\{#1,#2\}}\xspace}

\newcommand{\Jagg}{\hat{J}_{agg}}
\newcommand{\Jcount}{\hat{J}_{\mbox{count}}}
\newcommand{\Jsum}{\hat{J}_{\mbox{sum}}}
\newcommand{\Javg}{\hat{J}_{\mbox{avg}}}
\newcommand{\kkey}{k_{\mbox{key}}}
\newcommand{\ktup}{k_{\operatorname{tuple}}}
\newcommand{\Esum}{E_{\operatorname{sum}}}
\newcommand{\Ecount}{E_{\operatorname{count}}}
\newcommand{\Eavg}{E_{\operatorname{avg}}}

\newcommand{\ph}[1]{\vspace{2mm} \noindent \textbf{#1}---}

\newif\iftechreport
\techreportfalse % 

\begin{document}

\title{Joins on Samples: 
A Theoretical Guide for Practitioners}
%Theoretical Answers to Practical Questions
%Separating Myth From Reality
%Separating Myth From Reality
%What's Hard and What's Not

% \pagenumbering{roman}

% \author{}
% \date{\today}

\numberofauthors{1}

\author{
\alignauthor
Dawei Huang~~~~~~~~~Dong Young Yoon~~~~~~~~~Seth Pettie~~~~~~~~~Barzan Mozafari \vspace{0.3cm} \\
\affaddr{University of Michigan, Ann Arbor}
\email{\{hdawei, dyoon, pettie, mozafari\}@umich.edu}
}

\maketitle

%!TEX root = main.tex
\begin{abstract}
Despite  decades of research on AQP (approximate query processing), our
understanding of sample-based joins has remained 
    limited and, to some extent, even superficial.
% This has played a key role in sample-based joins remaining an open problem.
The common belief in the community is that joining random samples is futile.
This belief is largely based on
% The common belief in the community is that joining random samples is futile. 
% However, much of this belief is simply based on
an early result showing that the 
    join of two \emph{uniform} samples is not 
    an independent sample of the original join,
    and that it leads to quadratically fewer output tuples. 
% Unfortunately, despite the popular belief in the community, 
Unfortunately, 
    this early result has little applicability to the key questions practitioners face. For example, 
     the success metric is often
     the final approximation's accuracy, rather than  output cardinality. 
    Moreover, there are many non-uniform sampling strategies that one can employ. 
Is sampling for joins still futile 
         in all of these settings? 
If not, what is the best sampling strategy in each case?
To the best of our knowledge, there is no formal study answering these questions. 

% While we do not claim to have completely solved this problem, our goal in this paper is to improve our understanding of sampling for joins
This paper aims to improve our understanding of sample-based joins and 
        offer a guideline for 
        practitioners building and using real-world AQP systems. 
 We study limitations of   offline samples in approximating join queries: given an offline sampling budget, 
 how well can one approximate the join of two tables? 
 We answer this question for two success  metrics: output size and estimator variance. 
We show that maximizing output size is easy, while there is an information-theoretical lower bound on the
        lowest variance achievable 
            by any sampling strategy.
            % any \barzan{Dawei, is it `any' or just our own hybrid class?} \dawei{any} sampling strategy. 
We then define a hybrid sampling scheme that 
captures all combinations of stratified, universe, 
    and Bernoulli sampling,     and show that this scheme with our optimal parameters achieves         the theoretical lower bound within a constant factor.
Since computing these optimal parameters    
    requires shuffling statistics across the network, 
        we also propose a decentralized variant
    in which each node acts autonomously using minimal   statistics.
 We also empirically validate our findings on popular SQL and AQP engines.
\end{abstract}

%!TEX root = main.tex

\section{Introduction}
\label{sec:intro}

Approximate query processing (AQP) has regained significant attention in recent years due to 
    major trends in the industry \cite{mozafari_sigmod2017_invited}. 
Larger datasets, memory wall, and the separation of compute and storage have all made it harder 
    to achieve interactive-speed analytics. 
AQP presents itself as a viable alternative in scenarios where perfect decisions can be 
made with imperfect answers~\cite{mozafari_eurosys2013}.
AQP is most appealing when negligible loss of accuracy can be traded for a significant
    gain in speedup or computational resources.
Adhoc analytics~\cite{oracle-aqp,mozafari_pvldb2012,mozafari_sigmod2014_abm}, visualization~\cite{viz1-brown,viz1-aditya,mozafari_icde2016}, IoT~\cite{mozafari_cidr2017},
A/B testing~\cite{mozafari_sigmod2014_diagnosis}, 
email marketing and customer segmentation~\cite{walmart_talk}, and real-time threat detection~\cite{infobright_acquisition} are examples of such usecases.  

\ph{Sampling and Joins}
Sampling is one of the most widely-used techniques for general-purpose AQP~\cite{green-book}. 
The high level idea is to execute the query on a small sample of the original table(s) 
    to provide a fast, but approximate, answer.
While effective for simple aggregates, using samples for join queries 
    has long remained an open problem~\cite{join_synopses}. 
There are two main approaches to AQP: \emph{offline} or \emph{online}.
Offline approaches~\cite{mozafari_sigmod2018_verdict, join_synopses, mozafari_eurosys2013, aqua1, surajit-optimized-stratified, ganti2000icicles} build samples (or other synopses) prior to query arrival.
 At run time, they simply choose appropriate samples that can yield the best accuracy/performance for each incoming query.
Online approaches, on the other hand, wander-join
    perform much of their sampling at run time based on the query at hand~\cite{dynamicp-sample-selection, online-agg-mr2, online-agg, online-agg-mr1, cosmos, kandula2016quickr}. 
Naturally, offline sampling leads to significantly higher speedup, while online 
    techniques can support a much wider class of queries~\cite{kandula2016quickr}.
The same taxonomy applies to join approximation:
offline techniques perform joins on previously-prepared samples~\cite{join_synopses, random-sampling-joins, wander-join-2018, mozafari_sigmod2018_verdict, two-level-sampling}, while
    online approaches seek to produce a sample of the output of the join at run time~\cite{ripple_join, dbo, ripple_join_scalable, wander-join}.
As mentioned, the latter often means more modest speedups (e.g., 2$\times$~\cite{kandula2016quickr}) which may not be sufficient to justify approximation, 
    or additional requirements (e.g., an index for each join column~\cite{wander-join})
        which may not be acceptable to many applications.
Thus, our focus in this paper---and what is considered an open-problem---is the offline approach:
joins on samples, not sampling the join's output.

\ph{Joins on Samples} 
The simplest strategy is as follows. Given two large tables $T_1$ and
$T_2$, 
create a uniform random sample of each, say $S_1$ and $S_2$ respectively, 
    and then use $S_1 \bowtie S_2$
        to approximate aggregate statistics of $T_1 \bowtie T_2$.
This will lead to significant speedup if   samples are much smaller 
    than   original tables, i.e., $|T_i| \gg |S_i|$. 
    
One of the earliest results in this area shows that this simple strategy is futile for two reasons~\cite{aqua1}. 
First, 
    joining two uniform samples leads to quadratically fewer output tuples,
    i.e., joining two uniform samples that are each $p$  fraction ($0$ $\leq$ $p$ $<1$) of the 
        original tables will only produce $p^2$ of the output tuples of the original join (see 
        Figure~\ref{fig:join_example}).
Second,   
joining  
uniform samples of two tables does not yield 
an independent sample of their join\footnote{Prior work has stated that joining uniform samples is not a \emph{uniform} sample of the join~\cite{join_synopses}.  
We avoid this terminology since
    uniform   means 
    equal probability of inclusion, and in this case each   tuple
    does appear in the join of the uniform 
        samples with   equal probability, but not independently. In other words, 
    joining two i.i.d. samples is an identical, but not independent, sample of the join.} (see  Section~\ref{sec:bg:sampling_in_db} for details).
The dependence of the output tuples can drastically lower the approximation accuracy~\cite{aqua1, random-sampling-joins}.
\ignore{Figure~\ref{fig:join_example} illustrates the first problem.
Here, joining two uniform samples,
    each comprised of $p$=$0.5$ fraction of their respective tables,
results in only $p^2$=$0.25$ of the join tuples (1 instead of 4).}

\begin{figure}[t]
	\inv
	\centering
	\includegraphics[width=0.45\textwidth]{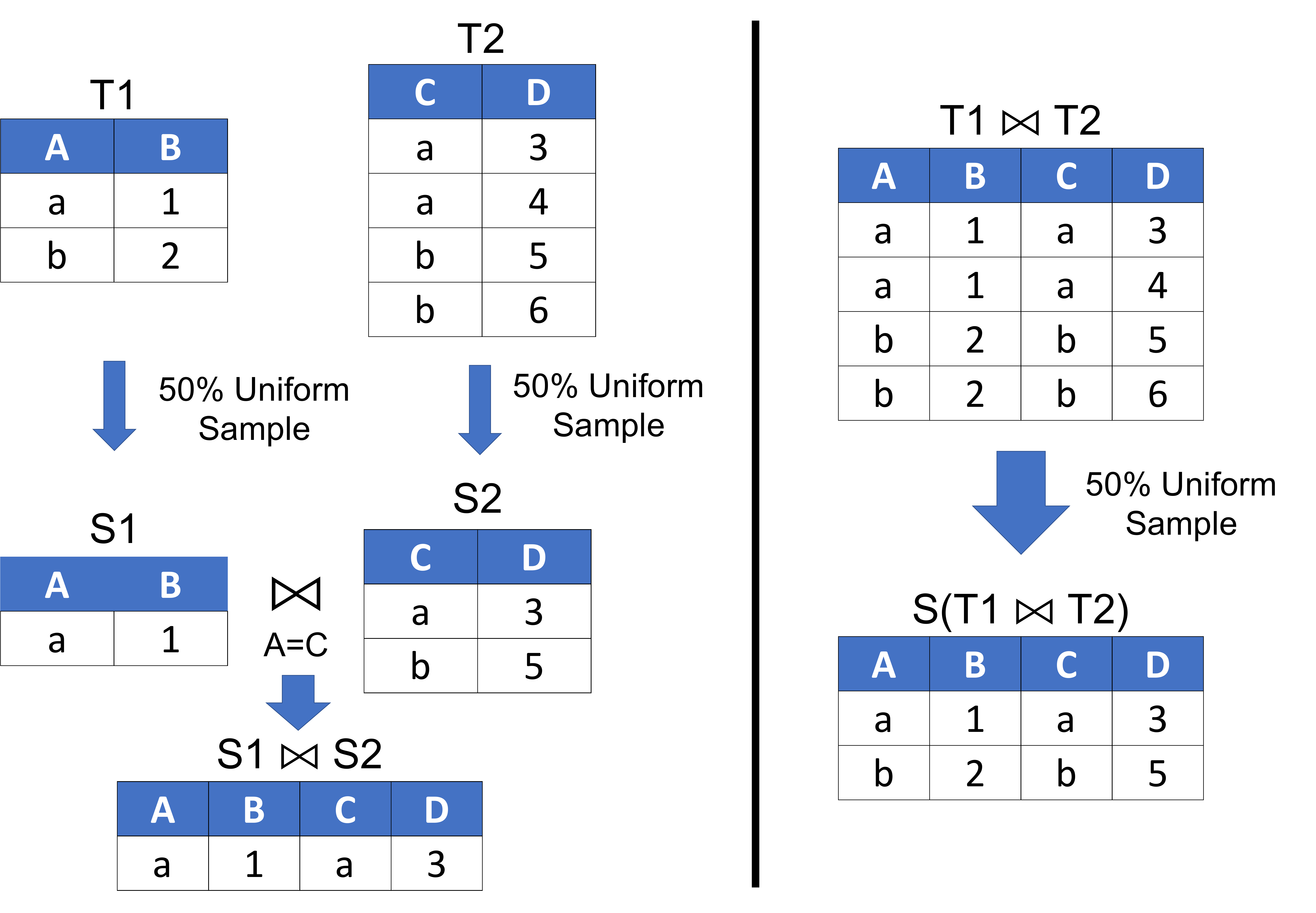}
	\inv
	\caption{\textbf{A toy example of joining two uniform samples (left) versus a uniform sample of the join (right).}}
	\inv
	\label{fig:join_example}
	\inv
\end{figure}

\ph{Prior Work}
\emph{Universe} sampling~\cite{hashed-samples, kandula2016quickr, mozafari_sigmod2018_verdict}  addresses the first drawback of uniform sampling. Although universe sampling avoids   quadratic reduction of output,
%: joining  $p$ fraction of the two tables  still produces a $p$ fraction of their join in expectation.
%    However, the tuples generated in this fashion are 
its creates even more correlation in its output, leading to much lower accuracy 
%,     especially in the presence of  popular values in the join column 
 (see Section~\ref{sec:hard:output_size}). 
    
Atserias et al. 
provide a worst case lower bound for any query involving equi-joins on multiple relations,
showing that computing exact joins with a small memory or time budget is
hard~\cite{size-bounds}.
For instance, the maximum possible join size for any cyclic join on three $n$-tuple relations is $\Theta(n^{1.5})$. 
Thus, a natural question is whether approximating joins is also hard with small memory or time. 

\ignore{They also propose an optimal join algorithm that meets this lower bound. 
This bound, however, is still polynomial in the size of the relations involved, which in modern database system can be too large to be practical, and many more application asks for a good approximation to the aggregate using a small amount of space and time.}

\ph{Our Goal}
This paper  focuses on understanding the limitation of using offline samples in approximating join queries. Given a sampling budget, how well can we approximate the join of two tables using their offline samples? 
To answer this question, we must first define what 
    constitutes a ``good'' approximation of a join. 
    We consider two metrics: 
        (1) output cardinality and (2) aggregation accuracy. 
        The former is the number of tuples of the original join that also appear in the join of the samples, whereas the latter is the  error
    of the aggregates estimated from the sample-based join with respect to their true values, if computed on the original join. Because in this paper we only consider unbiased estimators, we measure approximation error in terms of the variance of our estimators.

For the first metric,
    we provide a simple proof showing that universe sampling is optimal,
i.e. no sampling scheme with the same sampling rate can outperform universe sampling in terms of the (expected) output cardinality.
However, as we show in Section~\ref{sec:hard:output_size}, 
retaining a large number of join tuples does not imply accurate 
aggregates.
% \barzan{Dawei, which previous examples?? i don't see any}
% \dawei{i think the example is moved to a section 2.1, so I am just going to mention this here and refer it to section 2.1}
% \dy{I further fixed it as the better explanation is in Sec~\ref{sec:hard:output_size}}
It is therefore natural to also ask about the lowest variance 
    that can be achieved given a sampling rate. 
%ta inja
    To the best of our knowledge, this has remained an open problem to date. For the first time, we formally  study this problem and offer an 
    information-theoretical lower bound to this question.
We also present a hybrid sampling scheme that matches        this lower bound within a constant factor. 
This scheme involves a centralized computation, 
    which can become prohibitive for large tables
        due to large amounts of statistics that need to be shuffled across the network.
    Thus, we also propose a decentralized variant that only shuffles a minimal amount of information across the nodes---such as the table size and maximum frequency---but still achieves  the same worst case guarantees.
Finally, we generalize our sampling scheme to 
    accommodate \emph{a priori} information about 
    filters (i.e., \WHERE clause).

In this paper, we make the following contributions:
\begin{enumerate}[topsep=0.15cm,itemsep=0.1cm,leftmargin=0.3cm]
    \item We discuss two metrics---output size and estimator's variance ---for measuring the quality of join approximation, and show that universe sampling is optimal for output size and there is an information-theoretical lower bound for variance (Section~\ref{sec:hard}).
    
    \item We formalize a hybrid  scheme, called Stratified-Universe-Bernoulli Sampling (SUBS), which allows for different combinations of stratified,  universe, and Bernoulli sampling. We derive optimal sampling parameters within this 
    scheme, and show that they achieve the 
    theoretical lower bound of variance
        within a constant factor    (Section~\ref{sec:generic_sampling_scheme}--\ref{sec:avg}). 
    We also extend our analysis to accommodate
        additional information regarding the 
        \WHERE clause (Section~\ref{sec:one-sample}).
    % We show that BVS scheme can offer query estimator with much lower variance than both uniform and universe sampling 
    % \item We study two different settings, centralized (\ie, unlimited communication) vs. distributed (\ie, minimal communication), and discuss how the formulation differs between the two settings. We also further extend our analysis to joins with filters (\ie, WHERE clause).
    \item Through extensive experiments, we also empirically show that our optimal sampling parameters
        % \barzan{ask Dawei how to reword this. DVS is a scheme! and with our formulation is not meaningful}
    achieve lower error than existing sampling schemes
     in both centralized and decentralized scenarios (Section~\ref{sec:expr}).
\end{enumerate}

\section{Background}
\label{sec:bg}

In this section, we provide the necessary background on sampling-based join approximation. We also formally state our problem setting and assumptions.  

\subsection{Sampling in Databases}
\label{sec:bg:sampling_in_db}

% \subsubsection{Offline vs. Online}
% These are two approaches of how databases perform join sampling.
% The offline approach generates samples of each relation prior to query execution.
% Given a join query, the database is responsible for selecting the best set of samples that can yield the most accurate approximation.
% The online approach can provide a more accurate approximation than the offline approach as it can have a uniform sample from the join itself.
% However, the online approach has the problem of lacking generality and requiring possibly a very large storage due to having to calculate the entire join with the exact same join keys that the future query will use before the actual query arrives.
% Because of this, we only focus on the offline approach in this paper.

The following are the three main popular sampling strategies (operators) 
    used in AQP engines and database systems. 
\begin{enumerate}[topsep=0.15cm,itemsep=0.1cm,leftmargin=0.3cm]
    \item \textbf{Uniform/Bernoulli Sampling.} Any strategy that samples all tuples with the same probability is considered a uniform (random) sample. 
    Since enforcing fixed-size sampling without replacement is expensive in distributed systems, 
        Bernoulli sampling is considered a more efficient strategy~\cite{kandula2016quickr}.
        % often favored in practice as
        % a more efficient uniform strategy
In Bernoulli sampling, each tuple is included in the sample independently, with a fixed sampling probability $p$. In this paper, for simplicity, we use ``uniform'' and ``Bernoulli'' 
        interchangeably.
    As mentioned in Section~\ref{sec:intro},    
        joining two uniform samples leads to quadratically fewer output tuples.
    Further, it does not guarantee an i.i.d. sample of the original join~\cite{join_synopses}:
 the output is a uniform sample of the join but not an independent one.
 Consider an arbitrary tuple of the join, say $(t_1, t_2)$, where $t_1$ is from the first table and $t_2$ is from the second. 
  The probability of this tuple appearing in the join of the samples is always the same value, i.e., $p^2$. The output is thus a uniform sample. However, the tuples are not independent: consider another tuple of the join, say $(t_1, t_2')$ where $t_2'$ is another tuple from the second table joining with $t_1$. 
  If $(t_1, t_2)$ appears in the output, 
    the probability of $(t_1, t_2')$ also appearing  becomes $p$ instead of $p^2$, which would be the probability  if they were independent.

    \item \textbf{Universe Sampling.} Given a  column\footnote{$J$ can also be a set of multiple columns.} $J$, a (perfect) hash function $h: J \mapsto [0,1]$, and a sampling rate $p$, 
        this strategy includes a tuple $t$ in the table if 
        $h(t.J) \leq p$. 
    Universe sampling is often used for equi-joins, in which 
        the same $p$ value and hash function $h$ are applied to the join columns 
            in both tables. 
    This ensures that when a tuple $t_1$ is sampled from one table, any matching tuple $t_2$ from the other table is also sampled, simply because $t_1.J = t_2.J \Leftrightarrow h(t_1.J)=h(t_2.J)$.
    \cancut{
    This is why joining two universe samples of rate $p$ produces $p$ fraction of the original join output \emph{in expectation}.
    The output is a uniform sample of the original join, as each join tuple appears with the same probability $p$. 
    However, there is more dependence among the output     tuples. Consider two join tuples $(t_1, t_2)$ and $(t_1', t_2')$ where $t_1, t_1', t_2, t_2'$ all share the same join key.
    Then, if $(t_1,t_2)$ appears, the probability of $(t_1', t_2')$ also appearing will be $1$. Likewise, if $(t_1, t_2)$ does not appear, the probability of $(t_1', t_2')$ appearing will be $0$.
    Higher dependence means lower accuracy (see Section~\ref{sec:hard:output_size}).
    }

    \item \textbf{Stratified Sampling.} \rev{The  goal of stratified sampling is to ensure that minority groups 
    are sufficiently represented in the sample. 
    Groups are defined according to one or multiple columns, called the \emph{stratified columns}. 
    A group (a.k.a. a stratum) is a set of tuples that share 
    the same value under those stratified columns.   
    Given a set of stratified columns $C$
        and an integer parameter $\ktup$,
        a stratified sampling is a scheme that guarantees at least $\ktup$ tuples are sampled
        uniformly at random from each group.}
    \ignore{For instance, given $C$ and a target frequency $\kkey$, 
        a common stratified sampling scheme 
        is to choose $\kkey$ tuples uniformly at random from each group of tuples with the same value of $C$~\cite{mozafari_eurosys2013}.} When a group has fewer than $\ktup$ tuples, all of them are retained.
\end{enumerate}

% For instance, consider an example of $2$-way join query where each table consists of $n$ entries, all of the same join key. Supposed we are given a \texttt{count(*)} query. To estimate the query, we can perform the query on the sample and scale the answer up by $1/p$. This produce an unbiased estimator of variance $n^4/p$. While using uniform sampling, one can get an unbiased estimator of variance $n^2/p^2$, which is much smaller when $1/p << n^2$.

\subsection{Quality Metrics}
\label{sec:bg:success_metrics}

Different metrics can be used to assess the quality of a join approximation. In this paper, we focus on the following two, which are used by most AQP systems.

\ph{Output Size/Cardinality}
This metric  is the number of tuples of the original join that also appear in the join of the samples. 
It is mostly relevant  for exploratory usecases, where users visualize or examine a subset of the output.
In other cases, where an aggregate is computed from the join output, 
 retaining a large number of output tuples    does not guarantee accurate answers 
 (we show this in Section~\ref{sec:hard:output_size}).
% This metric looks at the number of tuples that the join of samples captures. 
% It may seem straightforward at first that having a more number of tuples better than having less for an accurate approximation,
% but this is not always true when we look at other metrics (e.g., variance). 

\ph{Variance} In scenarios where an aggregate function needs to be calculated from the join output, 
    the error of the aggregate approximation is more relevant than the number of intermediate tuples generated. 
For most non-extreme statistics, there are readily available unbiased estimators, e.g., Horvitz-Thompson estimator~\cite{horvitz1952generalization}.  Thus, a popular indicator of accuracy is 
    the variance of the estimator~\cite{mozafari_eurosys2013}, which  determines the size of the confidence interval given a sample size.

%The accuracy of the approximate aggregation refers to the multiplicative approximation factor of the aggregates obtained using the sample with respect to the aggregate using the original relations.
%Variance of the estimator to the original query using the offline samples is used to measure the accuracy.

% \subsubsection{Absolute/Relative Error}
% This metric measures approximation error by looking at the absolute (or relative) difference between the approximate aggregate from samples and the actual aggregate from the original relations.

\subsection{Problem Statement}
\label{sec:bg:problem_statement}

\begin{table}[t!]
  \caption{\textbf{Notations.}}
  \inv
  \centering
  \begin{tabular}{c|l}
  \hline
  Notation & Definition \\
  \hline
  $T_1$, $T_2$ & Two tables for the join \\
  $S_i$ & A sample generated from table $T_i$ \\
  $J$ & Column(s) used for the join between $T_1$ and $T_2$ \\
  $W$ & Column being aggregated (\eg, \SUM, \AVG)  \\
  $C$ & Column(s) used for filters (\ie, \WHERE clause) \\
  $\mathcal{U}$ & Set of all possible values of $J$ \\
  $a, b$ & Frequency vectors of $T_1$ and $T_2$'s 
  join columns, resp.\\
  $a_v, b_v$ & Number of tuples with join value $v$ in \\
  & $T_1$ and $T_2$, resp. \\
  $\Jagg$ & Estimator for a join query with  \\
  & aggregate function $agg$ \\
  $\epsilon$ & Sampling budget w.r.t. the original table size \\
  $n_1$, $n_2$ & Number of tuples in  $T_1$ and $T_2$, resp. \\
  $h$ & A (perfect) hash function \\
  $\ktup$ & Minimum number of tuples to be kept per group \\
  &  in stratified sampling\\
  $\kkey$ & Minimum number of join keys per group to apply \\
    &  universe sampling \rev{(universe sampling is not applied} \\
    & \rev{to groups with fewer than $\kkey$ join keys)} \\
  $p$ & Sampling rate of universe sampling \\
  $q$ & Sampling rate of uniform sampling \\
  \hline
  \end{tabular}
  \label{tab:notations}
  \inv\inv
\end{table}

In this section, we formally state the problem of 
sample-based join approximation. The notations used throughout the paper are listed in Table~\ref{tab:notations}.

\ph{Query Estimator}
Let  $S_1$ and $S_2$ be two samples generated offline from tables $T_1$ and $T_2$, respectively, and  $q_{agg}$ 
be a query that computes an aggregate function $agg$ on 
the join of $T_1$ and $T_2$. 
A query estimator $\Jagg(S_1, S_2)$ is a function 
that estimates
the value of $agg$ using two samples rather than the original tables.

\ph{Join Sampling Problem}
Given a query estimator $\Jagg$ and a sampling budget $\epsilon \in (0,1]$,
       our goal is to create a pair of samples $S_1$ and $S_2$---from tables $T_1$ and $T_2$, respectively---
        that are optimal in terms of a given success metric, 
            while respecting 
                a given storage budget  $epsilon$ on average.
Specifically, we seek $S_1$ and $S_2$ that minimize $\Jagg$'s variance or maximize its output size
    such that 
    $E[|S_1| + |S_2|] \leq \epsilon \times (|T_1| + |T_2|)$.

\tempcut{Note that we define the sampling budget in terms of an expected size (rather than a strict one), since  sampling schemes are probabilistic in nature and may slightly over- or under-use the budget.}

To formally study this problem, 
    we first need to define a class of reasonable sampling strategies. 
    In Section~\ref{sec:generic_sampling_scheme}, 
        we define a hybrid scheme that can capture different combinations of stratified, universe, and uniform sampling.

\ignore{\subsection{Known Results on Join Sampling}
\label{sec:bg:known_results}

\tofix{Chaudhuri et al.\cite{ChaudhuriMN99} first pointed out that join of uniform samples does not produce a uniform sample from the join. Moreover, it leads to quadratic reduction in the sampling rate. In the same paper, they also give a set of sampling algorithms that aims to produce an almost uniform sample of the join while requiring statistical information from both tables. Universe sampling \cite{VengerovMZC15} is proposed to address the issue in the loss in sampling rate. It can be seen as sampling join keys with a hash function as oppose to sampling tuples as Bernoulli sampling. Its performance suffers when there are high degree join keys that constitute a large portions of the join. End-biased sampling\cite{EstanN06} is a generalization of universe sampling also that favors join keys with high degrees. Chen and Yi \cite{two-level-sampling} proposed a two level sampling scheme that combines universe sampling and independent sampling and applied it to join size estimation.
}\dawei{Barzan take a look.}
}

% \dy{@Dawei: can you add known results that you studied before here? or at least give me some pointers so that I can look them up?}
% \dawei{maybe take a look at related.tex? Another good secondary resource is google ``Random Sampling over Joins Revisited'', the related work section contains some overview to the pass work.}

\subsection{Scope and Limitations}
\label{sec:bg:assumptions}

To simplify our analysis, we limit our scope in this paper.

\ph{Flat Equi-joins} 
We focus on equi (inner) joins as the most common form of joins in practice. 
We also support both $\tt WHERE$ and $\tt GROUP BY$ 
    clauses.
Because our focus is on the join itself, we ignore nested queries and only consider flat (or flattened) queries.
We primarily focus on two-way joins.
However, our results extend to multi-way joins with the same join column(s). 
\ignore{However, in its current forumulation, 
our BVS scheme does not support multi-way joins with different join columns.}

\ph{Aggregate Functions}
Most AQP systems do not support extreme statistics, such as $\tt Min$ or $\tt Max$~\cite{approx_chapter}.
Likewise, we only consider non-extreme aggregates, and primarily focus on the three basic functions, \COUNT, \SUM, and \AVG.
However, we expect our techniques to easily extend to other mean-like
statistics as well, such as \VAR, \STDEV, and \PERCENTILE.

%!TEX root = main.tex

\section{Hardness}
\label{sec:hard}

In this section, we explain why providing a large output size is insufficient for approximating joins, and formally show the hardness of approximating common aggregates based on the theory of communication complexity.

\subsection{Output Size}
\label{sec:hard:output_size}

Uniform sampling leads to  small output size.
If we sample at a rate $q$ from both table $T_1$ and table $T_2$,
the join of samples contains only $q^2$ fraction of  $T_1 \bowtie T_2$ in expectation.
Moreover, the join of two independent samples of the original tables is in general not an  independent sample of  $T_1 \bowtie T_2$,
 which hurts the sample quality.
% Universe sampling \cite{hashed-samples, kandula2016quickr} is designed to address the first issue.
In contrast, universe sampling~\cite{hashed-samples, kandula2016quickr} with sample rate $p$ can,  in expectation, sample a $p$ fraction of $T_1 \bowtie T_2$.
We prove that this is optimal 
(all omitted proofs are deferred to 
\iftechreport
Appendix~\ref{app:omitted_proof}).
\else
our report~\cite{approx-join-techreport}).
\fi

\begin{restatable}{theorem}{worstexpectation}
No sampling scheme with sample rate $\alpha$ can guarantee more than $\alpha$ fraction of $T_1 \bowtie T_2$ in expectation for all possible inputs.
% \end{theorem}
\end{restatable}

\ignore{
\begin{proof}
We can simply consider two identical tables $T_1$, $T_2$ of $n$ tuples, each having join key $1, 2, \ldots, n$. Their join has size $n$.
Since each tuple of $T_1$ joins with exactly one tuple of $T_2$, the size of the join of the samples must not be larger than the size of sample of $T_1$.
Since, by assumption, the expected size of the sample of $T_1$ is at most $\alpha n$, the expected size of the join of the samples must also be at most $\alpha n$.
\end{proof}
}

However, a large number of tuples retained in the join does not imply that the original join query can be accurately approximated.
As pointed out in \cite{two-level-sampling},
universe sampling shows poor performance in approximating queries when the frequencies of keys are concentrated on a few elements.
Consider the following extreme example with  tables
$T_1$ and $T_2$, each comprised of $n$ tuples with a single value $1$ in their join key.
In this example, universe sampling with the sampling rate $p$ produces an estimator of variance $n^4/p$, while uniform sampling with rate $q$ has a variance of $n^2/q^2$,
which is much lower when $p = q$ and $n$ is large. 
% $p$ is a constant and $n$ is large.
Thus, a larger output size does not necessarily lead to a better approximation of the query. 

\subsection{Approximating Aggregate Queries}
\label{sec:hard:approx_queries}

In this section, we focus on the core question: why is approximating  common aggregates (e.g., \COUNT, \SUM and \AVG) hard when using a small sample (or more generally,  a small summary)?
% In the following, we provide an evidence to address the above question.
% This evidence is based on the theory of communication complexity.
We address this question using the theory of communication complexity.
Specifically, to show that computing \COUNT on a join is hard, we reduce it to
    set intersection,
    a canonically hard problem in communication complexity.
\ignore{Consider two parties, Alice and Bob, holding information $x$ and $y$, respectively. 
They want to evaluate some function $f(x, y)$  while exchanging as little information as possible. }
  Assume that both Alice and Bob each hold a set of size $k$, say $A$ and $B$, respectively.
They aim to estimate the size of $t = |A \cap B|$.
Pagh et. al~\cite{min-wise-hashing} show that if Alice only sends a small summary to Bob, any unbiased estimator that Bob uses will have a large variance.

\begin{theorem}[See \cite{min-wise-hashing}] \label{thm:intersection}
Any one-way communication protocol that estimates $t$ within relative error $\delta$ with probability at least $2/3$ must send at least $\Omega(k/(t\delta^2))n$ bits.
\end{theorem}

\begin{corollary}\label{cor:lb}
Any estimator to $|A \cap B|$ produced by Bob that is based on an $s$-bits summary by Alice must have a variance of at least $\Omega(kt/s)$.
\end{corollary}

Any sample of size $s$ can be encoded using $O(\log{k \choose s})$ bits, implying that any estimator to \COUNT  that is based on a sample of size $s$ from one of the tables must have a variance of at least $\Omega(kt/s)$.

Estimating \SUM queries is at least as hard as estimating \COUNT queries, since any \COUNT  can be reduced to a \SUM  by setting all  entries in the \SUM column to $1$.

From the hard instance of set intersection,
we can also derive a hard instance for \AVG queries.
Based on Theorem~\ref{thm:intersection}, any summary of $T_1$ that can distinguish between intersection size $t(1+\delta)$ and $t(1 - \delta)$ must be at least of size $\Omega(k/(t\delta^2))$ bits.
Now we reduce this problem to estimating an \AVG query. 

Here, the two tables consist of $k + \sqrt{t}$ tuples each.
The first $k$ tuples of $T_1$ and $T_2$ are from the hard instance of set intersection,
and the values of their \AVG column are set to $2r$.
The join column of the last $\sqrt{t}$ tuples is set to some common key $v'$ that is in the first $k$ tuples, and their \AVG column is set to $0$.
Therefore, the intersection size from the first $k$ tuples is at least $t(1+\delta)$ (or at most $(t(1-\delta))$) if and only if the result of the \AVG query is at least $\frac{2rt(1+\delta)}{t(2 + \delta)} = (1+O(\delta))r$ (or at most $\frac{2rt(1+\delta)}{t(2 + \delta)} = (1-O(\delta)r)$).
By re-scaling $\delta$ by a constant factor, we can get the following theorem:

\begin{theorem}
Any summary of $T_1$ that can estimate an \AVG query with precision $\delta$ with probability at least $2/3$ must have a size of          at least $\Omega(n/(t\delta^2))$. 
\end{theorem}

%!TEX root = main.tex

\section{Generic Sampling Scheme}
\label{sec:generic_sampling_scheme}

To formally argue about the optimality of a sampling strategy, we must first define a \emph{class} of sampling schemes. As discussed in Section~\ref{sec:bg:sampling_in_db}, there are three well-known sampling operators: stratified, universe, and Bernoulli (uniform). 
However, these atomic operators can themselves be combined. For example, one can apply universe sampling of rate $0.1$ and then Bernoulli sampling of rate  $0.2$ for an overall effective sampling rate of 
$0.02$.\footnote{\rev{Statistically,
it does not matter which sampling is applied 
    first: whether a tuple passes the universe sampler and whether it passes the Bernoulli sampler are completely independent decisions, and hence, the output distribution is the same. 
Here, we apply universe sampling first only for convenience and
without loss of generality.\ignore{(conditioned on universe sampling, the distribution is a collection of Binomial random variables, whereas the distribution is more complex if we   condition on  Bernoulli sampling)}}}  
To account for such hybrid schemes, 
	we define a generic scheme that combines 
	universe and Bernoulli sampling, called UBS.\footnote{\rev{Even if we do not care about 
	    output cardinality, universe sampling 
	        can still help improve the approximation quality.
	 For example, given two tables of size $n$ with a one-to-one join relationship, the count estimator's
	 variance is $n/q^2$ under Bernoulli sampling but 
	 $n/p$ under universe sampling, which is much lower when $p$$=$$q$.}}
We also define a more generic scheme that  combines all three of stratified, universe and Bernoulli sampling, called SUBS. 
It is easy to show that the basic sampling operators are a special case of SUBS. 
% \barzan{Dong Young, go through the whole paper and replace all BVS with  UBS everywhere including in the plots}
First, we define the effective sample rate. 	 

\begin{definition}[Effective sampling rate]
We define the effective sampling rate of a sampling scheme as 
    the expected ratio of the size of the resulting sample to that of the original table.
\end{definition}

\begin{definition}[Universe-Bernoulli Sampling (UBS)]
Given a  table $T$ and a column (or set of columns) $J$ in $T$, 
 a UBS scheme is defined by a pair $(p, q)$, where
$0<$$p$$\leq 1$ is a universe sampling rate and
$0<$$q$$\leq 1$ is a Bernoulli (or uniform) sampling rate.
Let $h: \mathcal{U} \mapsto [0, 1]$ be a perfect hash function.
Then, a sample of $T$ produced by this scheme, $S = \mbox{UBS}_{p, q}(T, J)$, is produced 
 as follows:
%  \barzan{Dong Young, go through the whole paper and remove $h$ from the subtitle of UBS everywhere so that UBS only has p and q as subtitles}
\begin{algorithm}[H]
\SetAlgoLined
 $S \leftarrow \emptyset$\;
 \For{each tuple $t$} {
 \If{$h(t.J) < p$} {
 Include $t$ in $S$ independently w/ prob. $q$.
 }
 }
 \caption{$\textsc{UBS}_{p,q}(T, J)$}
\end{algorithm}
\end{definition}
%ta inja
It is easy to see that the effective sampling rate of a UBS scheme $(p,q)$ is 
$p \cdot q$.
Thus, the effective sampling rate here is independent of the actual distribution of the values in the table (and column(s) $J$). 

The goal of this sampling paradigm is to optimize the trade-off between universe sampling and Bernoulli sampling in different instances.
At one extreme, when each join value appears exactly once in both table,
universe sampling leads to lower variance than Bernoulli sampling.
This is because independent Bernoulli sampling has trouble matching tuples with the same join value, while universe sampling guarantees that when a tuple is sampled,
all matching tuples in the other table are also sampled.
At the other extreme, if all tuples have the same join value in both tables
(\ie, the join becomes a Cartesian product of the two tables),
universe sampling will either sample the entire join, or sample nothing at all, while uniform sampling will have a sample size concentrated around $qN$, thus giving an estimator of much lower variance. In section~\ref{sec:count} to ~\ref{sec:avg}, 
we give a comprehensive discussion on how to optimize $p$ and $q$ for different tables and different queries. 

The Stratified-Universe-Bernoulli Sampling Scheme applies to a table $T$ that is divided into $K$ groups (i.e., strata), denoted as $G_1, G_2, \ldots, G_k$.

\begin{definition}[Stratified-Universe-Bernoulli Sampling (SUBS)]
Given a  table $T$ of $N$ rows and a column (or set of columns) $J$ in $T$, 
 a SUBS scheme is defined by a tuple $(p_1, p_2, \ldots, p_K, q_1, q_2, \ldots, q_K)$, where
$0$$<$$p_i, q_i$$\leq$$1$ are the universe and Bernoulli sampling rates. 
Given a perfect hash function
$h$$: \mathcal{U}$$\mapsto [0, 1]$, 
a sample of $T$ produced by this scheme, $S = \mbox{UBS}_{p, q}(T, J)$, is produced 
 as follows: 
\begin{algorithm}[H]
\SetAlgoLined
 $S \leftarrow \emptyset$\;
 
 \For{each group $G_i$}{
 
 \For{each tuple $t$ in $G_i$} {
 \If{$h(t.J) < p_i$} {
 Include $t$ in $S$ independently w/ prob. $q_i$.
 }
 }
 }
 \caption{$\textsc{SUBS}_{p_1, \ldots, p_K, q_1, \ldots, q_K}(T, G, J)$}
\end{algorithm}
\end{definition}

Let $|G_i|$ denote the number of tuples in group $G_i$. Then the effective sampling rate of a SUBS scheme is  $\sum_i p_i \cdot q_i \cdot |G_i|/N$. 
We call $\epsilon_i = p \cdot q_i$ the effective sampling rate for group $G_i$. 

In both UBS and SUBS schemes, the user specifies  $\epsilon$ as  their desired sampling budget, given which our goal is to determine optimal sampling parameters $p$ and $q$ (or $p_i$ and $q_i$ values) such that the variance of our join estimator is minimized. In Section~\ref{sec:algorithm},
    we derive the optimal $p$ and $q$ for UBS. 
For SUBS, in addition to $\epsilon$,
    the user also provides two additional parameters $\kkey$ and $\ktup$ (explained below).
Next, we show how to     
    determine the effective sampling rate $\epsilon_i$ for each group $G_i$ based on these parameters in SUBS.
Given $\epsilon_i$ for each group, 
    the problem is then reduced to finding the 
        optimal parameters for UBS for that group (i.e., 
        $p_i$ and $q_i$).  Moreover, as we will show in Sections~\ref{sec:count}--\ref{sec:avg},  particularly in Lemma~\ref{lem:same-rate}, the universe sampling rate for every group must be the same, and must be the same as the universe sampling rate of the other table in two-way joins. Hence, we use a single universe sampling rate $p = p_1 = \ldots = p_k$ across all  groups.
        
As mentioned in Section~\ref{sec:bg:sampling_in_db},
\rev{$\ktup$} is a user-specified lower bound on the minimum
number of tuples\tempcut{\footnote{The lower bound holds only on average, due to the probabilistic nature of sampling.}} in each group the sample must retain. 
$\kkey$ is an additional user-specified parameter required  for the SUBS scheme. It specifies a threshold at which to activate
    the universe sampler. In particular, if a group contains too few (i.e., less than $\kkey$) \emph{join keys}, we do not perform any universe sampling as it will have a high chance of filtering out all tuples. 
    Hence, we apply universe sampling only to those groups with $\geq \kkey$ join keys. For groups with fewer than $\kkey$ join keys, we will only apply  Bernoulli sampling with rate $\epsilon_i$.

We call a group \emph{large} if it contains at least $\kkey$ join keys, otherwise, we call it a \emph{small} group. We use $N_{b}$ to denote the total number of tuples in all large groups, and $N_{s}$ to denote the total number of tuples in all small groups. Similarly, let $M_{b}$ and $M_{s}$ denote the number of large and small groups, respectively. Then, we decide the sampling budget $\epsilon_i$ for each group $G_i$ as follows:
 \begin{enumerate}[nosep,leftmargin=0.4cm,topsep=0.1cm]%[topsep=0.12cm,itemsep=0.05cm,leftmargin=0.5cm]
     \item \rev{If $M_{s} \ktup > \epsilon N_{s}$ or $M_{b}   \ktup > \epsilon N_{b} $,} we notify the user that creating a sample given their parameters is infeasible. 
     
     \item Otherwise, 
        % \barzan{Dong Young, plz fix the styling issues here:1) the bullet symbols aren't showing  2) the min and max subscripts are bigger than their main symbol (e.g., K), 3) the t is too small to be visible}
      \begin{itemize}[topsep=0.0cm]
         \item \rev{Let $\epsilon_{s}' = \frac{K_{s} \cdot \ktup}{ N_{s}}$ and let $\epsilon_{s}'' = \epsilon - \epsilon_{s}'$. Then for each small group $G_i$, the sampling budget is $\epsilon_i = \frac{\ktup}{ |G_i|} + \epsilon_{s}''$.}
         \item \rev{Let $\epsilon_{b}' = \frac{K_{b} \cdot  \ktup}{ N_{b}}$ and let $\epsilon_{b}'' = \epsilon - \epsilon_{b}'$. Then for each large group $G_i$, the sampling budget  is $\epsilon_i = \frac{\ktup}{|G_i|} + \epsilon_{b}''$.} 
         
     \end{itemize}
 \end{enumerate}
 
Once  $\epsilon_i$ is determined for each group,     the problem of deciding optimal SUBS parameters is reduced to deciding the optimal  SUBS parameters  for $K$ separate groups.
  This effective sampling rate $\epsilon_i$ guarantees that each
large
group will have at least $t$ tuples in the sample on average, and the remaining budget is divided evenly. 
 Thus, the corresponding uniform sampling rate
    for each large group is $q_i = \epsilon_i / p$. Moreover, we pose the constraint
 that the universe sampling rate $p$ should be at least $1/s$ to guarantee that, on average, there is at least one join key passing through the universe sampler. 
 
 For small groups, we simply apply uniform sampling with rate $\epsilon_i$. This is equivalent to  setting $p = 1$ for these groups.
  
 Overall, this strategy provides the following guarantees:
 \begin{enumerate}[nosep,leftmargin=0.4cm]%[topsep=0.12cm,itemsep=0.01cm,leftmargin=0.4cm]
 \item Each group will have at least $t$ tuples in the sample, on average.
 
 \item The probability of each group being missed is at most $(1$ $-$ $1/s)^s$ $<$ $0.367$. In general, if we set $p$$>$$c/s$ for some constant $c$$>$$1$,  this probability will become $0.367^c$.\ignore{The probability of a group being missed by at most $(1 - p)^s$, which is at most $(1 - 1/s)^{s}$, which is at most $1/e$, if you replace $1/s$ as $c/s$, then it is at most $1/e^c$.  }
 
 \item The approximation of the original query will  be optimal in terms of its variance (see Sections~\ref{sec:count}--\ref{sec:avg}).
 
 \end{enumerate}

%!TEX root = main.tex

\section{Optimal Sampling}
\label{sec:algorithm}

As shown in Section~\ref{sec:generic_sampling_scheme},
     finding the optimal sampling parameters within
        the SUBS scheme can be reduced
            to finding those 
                within the UBS scheme. 
Thus, in this section, we focus on deriving 
 the UBS parameters that 
        minimize error for each aggregation type
        ($\tt COUNT$, $\tt SUM$, and $\tt AVG$).
Initially, we also assume there is no
    $\tt WHERE$ clause. 
Later, in Section~\ref{sec:one-sample},
    we show how to
    handle $\tt WHERE$ conditions and how to
    create a single sample
    instead of creating one per each 
    aggregation type and $\tt WHERE$ condition.

\ph{Centralized vs. Decentralized} For each aggregation type, we analyze two scenarios: centralized and decentralized. 
Centralized setting is when 
the frequencies of the join keys in both tables are known. This represents situations where both tables are stored on the same server, or each server communicates its full frequency statistics to other parties.
Decentralized setting 
	is a
scenario where the two tables are each stored on a separate server~\cite{mozafari_sigmod2018_rdma}, and exchanging full frequency statistics across the network is costly.\footnote{Here, we focus on two servers, but the math can  easily be generalized to decentralized networks of multiple servers.} 

\rev{
\ph{Decentralized Protocols}
In a decentralized setting, 
	each party (i.e., server) only has access  to full statistics  of its own table (\eg, frequencies,  join column distribution).
The goal then is for each party to determine 
    its sampling strategy, while minimizing communications with the other party.
Depending on the amount of information exchanged,  
    one can pursue different protocols for achieving this goal.
In this paper, 
    we study a simple 
    sampling protocol, which we call \textsc{Dictatorship}.
Here, one server, say \partyone, is chosen as the dictator.
We also assume that 
the parties know each other's sampling budgets and table sizes ($\epsilon_1$, $\epsilon_2$, $|T_1|$, and $|T_2|$).
The dictator observes the distributional information of its own table, say $T_1$, 
	and decides a shared universe sampling rate $p$ between $\max\{\epsilon_1, \epsilon_2\}$ and $1$.
This $p$ is sent to the other server (\partytwo) and both servers use $p$ as their  universe sampling rate.\footnote{\rev{Using the same universe 
sampling rate  is justified by Lemma~\ref{lem:same-rate}.}}
Their uniform sampling rates will thus be $q_1 = \epsilon_1 / p$ and $q_2 = \epsilon_2 / p$, respectively. 

Since \partyone only has 
    $T_1$'s frequency information, 
    it chooses 
        an optimal value of $p$ that minimizes the \emph{worst case} variance of $\Jagg$,
\ie, the variance when the frequencies in $T_2$ are  chosen adversarially.
This can be formulated as a robust 
optimization~\cite{mozafari_sigmod2015}:
\begin{equation} 
    p^* = \argmin_{\max\{\epsilon_1, \epsilon_2\}\leq p\leq 1} \max_{b} \var[\Jagg] 
    \label{opt:dis:agg}
\end{equation}
where $b$ ranges over all possible frequency vectors of $T_2$. 
In the rest of this paper, we use \textsc{Dictatorship} in our decentralized  analysis (we defer more complex
     protocols
%that exchange 
%     additional information (e.g., voting,
%        iterative convergence), 
to}
\iftechreport
        \rev{Appendix~\ref{app:other-decentralized}).}
\else   
        \rev{\cite{approx-join-techreport}).}
\fi

\ignore{Suppose that the two tables are distributed across two parties (e.g., two servers).
We call them \partyone and \partytwo,
where \partyone holds $T_1$ and \partytwo holds $T_2$.
This setting is also applicable when both tables are on the same server, but due to performance considerations,
we desire a sampling strategy that operates on a single table at a time without requiring
detailed information about the other table.}

\ignore{
This is mainly because if we require two parties that have the same universe sampling rate $p$,
$p$ should be at least $\max\{\epsilon_1, \epsilon_2\}$ if we want the effective sampling rate for $T_1$ and $T_2$ to be \textit{exactly} $\epsilon_1$ and $\epsilon_2$. 
This assumption simplifies the analysis and presentation of our results.
}

\subsection{Join Size Estimation: Count on Joins}
\label{sec:count}

We start by considering the following simplified query:
\begin{verbatim}
    select count(*) from T1 join T2 on J
\end{verbatim}
where $T_1$ and $T_2$ are two tables joined on column(s) $J$.
Consider $S_1=\UBS_{(p_1, q_1)} (T_1,J)$ and $S_2=\UBS_{(p_2, q_2)} (T_2,J)$.
Then, we can define an unbiased estimator for the above query, $\Ecount = |T_1 \bowtie_J T_2|$,
using $S_1$ and $S_2$ as follows.
Observe that given any pair of tuples $t_1 \in T_1$ and $t_2 \in T_2$, where $t_1.J = t_2.J$,
the probability that  $(t_1, t_2)$ enters $S_1$$\bowtie$$S_2$ is $\pmin q_1 q_2$, where $\pmin$$=$$\min\{p_1,$ $p_2\}$.
Hence, the following is an unbiased estimator for $\Ecount$.
\begin{equation}
    \Jcount(p_1, q_1, p_2, q_2, S_1, S_2) = \frac{1}{\pmin q_1 q_2} |S_1 \bowtie S_2|.
\end{equation}
When the arguments $p_1, q_1, p_2, q_2, S_1, S_2$ are clear from the context, we omit them and 
simply write $\Jcount$.

\ignore{
\begin{definition}[Join Size Estimation]
Given sampling budgets $\epsilon_1, \epsilon_2$, our goal is to 
find  parameters $(p_1, q_1)$ and $(p_2, q_2)$ that minimize the variance of  $\Jcount$
subject to $p_1 q_1 = \epsilon_1$  and $p_2 q_2 = \epsilon_2$. 
%We call $(p_1, q_1)$ and $(p_2, q_2)$ the optimal $\UBS$ sampling parameters for count.
\end{definition}
}
\begin{restatable}{lemma}{jcountvar}
\label{lem:jcount:var}
Let $S_1 = \UBS_{p_1, q_1}(T_1, J)$ and \\ $S_2 = \UBS_{p_2, q_2}(T_2, J)$. The variance of   $\Jcount$ is as follows:
\begin{align*}
    \var(\Jcount) = \frac{1-p}{p} \gamma_{2,2} + \frac{1 - q_2}{pq_2} \gamma_{2,1} \\ + \frac{1 - q_1}{pq_1} \gamma_{1,2}  + \frac{(1-q_1)(1-q_2)}{pq_1q_2} \gamma_{1,1}.
\end{align*}
where \rev{$\gamma_{i,j} = \sum_{v} a_v^i b_v^j$}.
\end{restatable}

To minimize $\var(\Jcount)$ under a fixed sampling budget, %(i.e., effective  sampling rate), 
the two tables should always use the same \emph{universe} sampling rate.
If $p_1$$>$$p_2$,  the \emph{effective universe sampling rate} is only $p_2$, \ie, only $p_2$ fraction of the join keys inside $T_1$ appear in the join of the samples, and the remaining $p_1$$-$$p_2$ fraction is simply \emph{wasted}.
Then, we can change the universe sampling rate of $T_1$ to $p_2$ and increase its uniform sampling rate to obtain a lower variance.

\begin{restatable}{lemma}{samerate}
\label{lem:same-rate}
Given tables $T_1$, $T_2$ joined on column(s) $J$, a fixed sampling parameter $(p_1, q_1)$ for $T_1$, and a fixed effective sampling rate $\epsilon_2$ for $T_2$, the variance of $\Jcount$ is minimized when $T_2$ uses $p_1$ as its universe sampling rate and correspondingly $\epsilon_2 / p_1$ as its uniform sampling rate. 
\end{restatable}

Note that Lemma~\ref{lem:same-rate} applies to both centralized and decentralized settings, \ie, it applies to any feasible sampling parameter $(p_1, q_1)$ and $(p_2, q_2)$, regardless 
of how the sampling parameter is decided. 
Next, we analyze each setting. 

\subsubsection{Centralized Sampling for Count}
\label{sec:count:centralized}

We have the following result.
\begin{restatable}{theorem}{thmvarcountcent}
\label{thm:var-count-centralize} 
When $T_1$ and $T_2$ use
sampling parameters $(p, \epsilon_1 / p)$ and $(p, \epsilon_2/p)$,  $\Jcount$'s variance is given by:
\begin{align*} 
      &\var[\Jcount] = (\frac{1}{p} - 1) \rev{\gamma_{2,2}} + (\frac{1}{\epsilon_2} - \frac{1}{p})\rev{\gamma_{2,1}} \\ 
    &+ (\frac{1}{\epsilon_1} - \frac{1}{p})\rev{\gamma_{1,2}} + (\frac{p}{\epsilon_1 \epsilon_2} - \frac{1}{\epsilon_1} - \frac{1}{\epsilon_2} + \frac{1}{p})\rev{\gamma_{1,1}}.
\end{align*}
\end{restatable}
\rev{Since each term in Theorem~\ref{thm:var-count-centralize} that depends on $p$ is  proportional either to $p$ or $1/p$, to find a $p$ that minimizes the variance, one can simply set the first order derivatives (with respect to $p$) to $0$.}

\begin{restatable}{theorem}{optrate}
\label{thm:opt-rate}
Let $T_1$ and $T_2$ be two tables joined on column(s) $J$.
Let $a_v$ and $b_v$ be the frequency of value $v$ in column(s) $J$ of tables  $T_1$ and $T_2$, respectively.
 Given their  sampling rates $\epsilon_1$ and $\epsilon_2$, the optimal sampling parameters $(p_1, q_1)$ and $(p_2, q_2)$ are given by:\\
\begin{equation*}
p_1\text{=}p_2\text{=}\min\{1, \max\{\epsilon_1, \epsilon_2,
\rev{\sqrt{\frac{\epsilon_1\epsilon_2\gamma_{2,2} - \gamma_{1,2} - \gamma_{2,1} + \gamma_{1,1}}{\gamma_{1,1}}}}\}\}
\end{equation*}
and $q_1 \text{=} \epsilon_1 / p$, $q_2 \text{=} \epsilon_2 / p$.
\end{restatable}

Substituting this into Lemma~\ref{lem:jcount:var},
the resulting variance is only a constant factor of Theorem~\ref{thm:intersection}'s theoretical limit. 
\rev{For instance, consider a
primary-key-foreign-key join query where $a_v\in \{0, 1\}$ and $b_v$ is smaller than some constant, say $5$, 
and $\epsilon_1$$=$$\epsilon_2$$=$$\epsilon$ for any $\epsilon$,  Theorem~\ref{thm:opt-rate} chooses 
$p_1$$=$$p_2$$=$$\epsilon$}. Then the variance given by Theorem~\ref{thm:var-count-centralize} becomes $(1/\epsilon - 1) J$ where $J = \sum_v a_v b_v$ is the size of the
join. Since $\epsilon$ is the expected ratio of the sample
to table size, the expression $(1/\epsilon - 1)J$ matches the lower bound in Corollary~\ref{cor:lb} except for a constant factor.

% DY: I am moving the below to replace previously Lemma 13
% \subsubsection{"New" Distributed Setting}

% plan: 
% \begin{enumerate2}
% \item First argue why $p = 1$ is more desirable when we want worst case guarantee without any information about the opposing parties.
% \item 
% \end{enumerate2}

% In distributed setting, we wish to decide a good sampling parameter $p$ such while both parties maintains minimal information about the opposing parties. For example, in most of the practical situation, it can be computationally infeasible to maintain frequencies moments $\sum_v a_v^s b_v^t$ as in Theorem~\ref{thm:opt-rate}

% This is based on the following Lemma:
% \begin{lemma}
% Let $F_a = \max_{v} a_v$ and $F_b = \max_v b_v$. Let $p^* = \sqrt{\frac{\epsilon_1\epsilon_2\sum_v(a_v^2 b_v^2 - a_v^2b_v - a_vb_v^2 + a_vb_v)}{\sum_v{a_vb_v}}}$, then $p^* \leq \sqrt{\epsilon_1\epsilon_2 (F_aF_b - F_a - F_b + 1)}$
% \end{lemma}.
% \begin{proof}
% This follows from $\sum_v a_v^2 b_v^2 \leq TF_1F_2$, $\sum_v a_v^2b_v \leq TF_1$ and $\sum_v a_vb_v^2 \leq TF_2$.
% \end{proof}

\subsubsection{Decentralized Sampling for Count}
\label{sec:count:decentralized}

\rev{Motivated by Lemma~\ref{lem:same-rate}, 
the \textsc{Dictatorship} protocol 
	uses the same universe sampling rate $p$ for both parties 
	in the decentralized setting, by solving the following 
    robust optimization problem:
\begin{equation*} 
    \argmin_{\max\{\epsilon_1, \epsilon_2\}\leq p\leq 1} \max_{b} \var[\Jcount] 
%    \label{opt:dis}
\end{equation*}
}

Based on Lemma~\ref{lem:jcount:var} and \ref{thm:opt-rate}, 
given the  effective sampling rates $\epsilon_1$ and $\epsilon_2$, we can express $\var[\Jcount]$ as a function of frequencies $\{a_v\}$ and $\{b_v\}$, and universe sampling rate $p$ as follows.
\begin{equation} 
\begin{gathered}
    \var[\Jcount] = (\frac{1}{p} - 1) \gamma_{2,2}+ (\frac{1}{\epsilon_2} - \frac{1}{p})\gamma_{2,1}\\ 
    + (\frac{1}{\epsilon_1} - \frac{1}{p})\gamma_{1,2} + (\frac{p}{\epsilon_1 \epsilon_2} - \frac{1}{\epsilon_1} - \frac{1}{\epsilon_2} + \frac{1}{p})\gamma_{1,1}.
\end{gathered}
\label{eqn:count}
\end{equation}

\cancut{
\begin{restatable}{lemma}{robustcount}
\label{lem:robust-count}
    Let $a_*$ be the maximum frequency in table $T_1$, $v_*$ be any value that has that frequency, and $n_b$ be the total number of tuples in $T_2$. The optimal value for  $\max_{\bm{b} \in \mathcal{K}_{n_b}} \var[\Jcount]$ is given by $(\frac{1}{p} - 1) a_*^2 n_b^2  + (\frac{1}{\epsilon_2} - \frac{1}{p})a_*^2 n_b + (\frac{1}{\epsilon_1} - \frac{1}{p}) a_* n_b^2 + (\frac{p}{\epsilon_1 \epsilon_2} - \frac{1}{\epsilon_1} - \frac{1}{\epsilon_2} - \frac{1}{p}) a_* n_b$
\end{restatable}
}
% \begin{proof}
%     Since $\var[\Jcount]$ is strictly convex in terms of the $b_v$'s in its domain. To maximize it, it suffices to consider only vertices of the polytope $K_{n_b}$. These are the all $0$ vector $\bm{0}$, and the vector $\bm{\hat{b}}_v$ that has $n_b$ at its $v$-th entries and $0$ everywhere else. By basic math, the vector that maximize $\var[\Jcount]$ is $\hat{b}_{v_*}$.  
% \end{proof}

In equation (\ref{eqn:count}),
	given $\{a_v\}$ and a fixed $p$,
		the variance  is a convex function  of the frequency vector $\{b_v\}$.
Thus, the frequency vector $\{b_v\}$ that maximizes the variance,
\ie, the worst case $\{b_v\}$, is one where exactly one join key has a non zero frequency. This join key should be the one with the maximum frequency in $T_1$.
This is not a representative case and using it to decide a sampling rate  might drastically hinder the performance on average.
We therefore require that both servers also share a simple piece of information regarding the \emph{maximum frequency} of the join keys in each table, say 
$F_a = \max_v a_v$ and $F_b = \max_v b_v$.
With this information, the new optimal sampling rate is given by:

\begin{restatable}{theorem}{thmcountdist}
Given $\epsilon_1$ and $\epsilon_2$, the optimal UBS parameter $(p, q_1)$ and $(p, q_2)$ for \COUNT in the decentralized setting are given by  
\begin{equation*}
    p = \min\{1, \max\{\epsilon_1, \epsilon_2, \sqrt{\epsilon_1\epsilon_2
    (F_aF_b - F_a - F_b + 1)}\}\}
    % p = \min\{1. \max\{\epsilon_1, \epsilon_2, \sqrt{\frac{\epsilon_1\epsilon_2
    % (a_*n_b^2 + a_*^2n_b + a_*n_b^2 + a_* n_b)}{\sum_v{a_*n_b}}}\}\}
\end{equation*}
and $q_1 = \epsilon_1 / p$, $q_2 = \epsilon_2/p$.
\end{restatable} 
% \subsubsection{Distributed Setting with Sketches}

\subsection{Sum on Joins}
\label{sec:sum}
Let $\Esum$ be the output of the following simplified query:
\begin{verbatim}
    select sum(T1.W)
    from T1 join T2 on J
\end{verbatim}
% To define an unbiased estimator, let $F$ be the exact (scalar) value returned by the above query,
Let $F$  be the sum of column $W$ in the joined samples $S_1 \bowtie S_2$. Then,
 the following is an unbiased estimator for $\Esum$:
 \begin{equation}
 \Jsum = \frac{1}{p_{\min} q_1 q_2} F
 \end{equation}
 where $\pmin =\min\{p_1, p_2\}$.
 \begin{restatable}{lemma}{lemsumunbias}
$\E[\Jsum] = \Esum$. 
 \end{restatable}

\ignore{
\begin{proof}
Similar to $\Jcount$, each pair of tuples $(t_1, t_2)$ in the join appears in the join of the sample with probability $\pmin q_1q_2$. We have:
\begin{align*}
     &E[\Jsum] = \frac{1}{\pmin q_1q_2} E[SUM_W] \\
    =& \frac{1}{\pmin q_1q_2} \sum_{\substack{(t_1, t_2): \\ t_1 \in T_1, t_2 \in T_2 \\ t_1.J = t_2.J}} ((\pmin q_1q_2)t_1.c + (1 - \pmin q_1q_2)\cdot 0) \\
    =& \sum_{\substack{(t_1, t_2): t_1 \in T_1, t_2 \in T_2 \\ t_1.J = t_2.J}} t_1.c
\end{align*}
\end{proof}
}

Let $\mu_v$ and $\sigma_v^2$ be respectively the mean and variance of attribute $W$ of the tuples in $S_1$ that have the join value $v$. Further, recall that $a_v$ is the number of tuples in $T_1$ with join value $v$. The following lemma gives the variance of $\Jsum$.

\begin{restatable}{lemma}{varsum}%[Variance of $\Jsum$]
\label{lem:jsum:var}
The variance of $\Jsum$ is given by:
\begin{equation}
\begin{gathered}
    \var[\Jsum] = \frac{1 - q_2}{pq_2} \beta_1 + \frac{1 - q_1}{pq_1} \beta_2 \\
    + \frac{(1 - q_1)(1 - q_2)}{p q_1 q_2} \beta_3 + \frac{1 - p}{p} \beta_4
\end{gathered}
\label{eqn:var-sum}
\end{equation}
\rev{where $\beta_1 = \sum_v a_v^2 \mu_v^2 b_v$, $\beta_2 = a_v(\mu_v^2 + \sigma_v^2) b_v^2$, $\beta_3 =  a_v (\mu_v^2 + \sigma_v^2) b_v$ and $\beta_4 =  a_v^2 \mu_v^2 b_v^2$.}
\end{restatable}

Analogous to Lemma~\ref{lem:same-rate}, we have the following result.

\begin{restatable}{lemma}{sameratesum}
\label{lem:same-rate-sum}
Given tables $T_1$, $T_2$ joined on column(s) $J$, fixed sampling parameters $(p_1, q_1)$ for $T_1$, and a fixed effective sampling rate $\epsilon_2 \leq p_1$ for $T_2$,
the variance of $\Jsum$ is minimized when $T_2$ also uses $p_1$ as its universe sampling rate and correspondingly, $\epsilon_2 / p_1$ as its uniform sampling rate. 
\end{restatable}

\subsubsection{Centralized Sampling for Sum}
\label{sec:sum:centralized}

Based on Lemma~\ref{lem:same-rate-sum}, we use the same universe sampling rate $p \geq \epsilon_1, \epsilon_2$ for both tables, with their corresponding uniform sampling rates being $q_1 = \epsilon_1/p$ and $q_2 = \epsilon_2/p$.
 Then we can further simplify equation~\ref{eqn:var-sum} into:
\begin{restatable}{theorem}{thmsumcent} \label{thm:var-sum-centralize}
When $T_1$ and $T_2$ both use the universe sampling rate $p$ and respectively use the uniform sampling rate $q_1 = \epsilon_1/p$ and $q_2 = \epsilon_2/p$, the variance of $\Jsum$ is given by:
\begin{align*}
  \var[\Jsum] =& \sum_v (\frac{1}{\epsilon_2} - \frac{1}{p}) \rev{\beta_1} + (\frac{1}{\epsilon_1} - \frac{1}{p}) \rev{\beta_2} \\
  &+ (\frac{p}{\epsilon_1\epsilon_2} - \frac{1}{\epsilon_1} - \frac{1}{\epsilon_2} + \frac{1}{p}) \rev{\beta_3} + (\frac{1}{p} - 1) \rev{\beta_4}.
\end{align*}
\end{restatable}
%The proof of the following theorem is similar to Theorem~\ref{thm:opt-rate}.
\begin{restatable}{theorem}{thmoptratesum} \label{thm:opt-rate-sum}
 Given effective sampling rates $\epsilon_1, \epsilon_2$, the optimal sampling parameters for \SUM in a centralized setting are given by $p\text{=}\min\{1, \max\{\epsilon_1, \epsilon_2,\sqrt{\epsilon_1\epsilon_2 \rev{\frac{\beta_1 + \beta_3 - \beta_2 - \beta_4}{\beta_3}} }\}\}$, 
$q_1\text{=}\frac{\epsilon_1}{p}$ and 
$q_2\text{=}\frac{\epsilon_2}{p}$.
\end{restatable}

\subsubsection{Decentralized Sampling for Sum}
\label{sec:sum:decentralized}

\rev{Lemma~\ref{lem:same-rate-sum} implies that, in a decentralized setting for \SUM estimation,
    the universe sampling rate $p$
	must be decided by the party that has 
	$T_1$, i.e.,
	the table with the aggregate column.}

% With a similar argument as Lemma~\ref{thm:count_dist},
	Given a fixed $T_1$ and $p$,
		$\var[\Jsum]$ is a strictly convex  function of $T_2$'s frequency vector.
Hence, the worst case instance is a point distribution where all tuples in $T_2$ share the same join key. However, for \SUM, the worst case distributions in $T_2$ are \emph{not} the same for all possible sampling parameters $p$.
% In \SUM, a different $p$ might lead to a different worst case distribution.
Define $h_v(p)$ to be $\var[\Jsum]$ as a function of $p$ where $T_2$'s frequency vector  is all  concentrated on the join key $v$, 
and define $h^*(p) = \max_v h_v(p)$.
\iftechreport
Since all $h_v(p)$'s are convex in $p$,  $h^*(p)$ is still convex and 
its exact minimum can be computed using a sweep line 
algorithm (see \cite[\S 8]{CLRS} for details).
In a nutshell, the algorithm  sweeps  all possible values of $p$ and uses a data structure to keep track of  $\max_v h_v(p)$ at that particular value.
The data structure uses $O(|\mathcal{U}|)$ memory, which can be costly in practice. 
\else
\rev{
Since $h^*(p)$ is convex and piece-wise quadratic, its minimum can be attained using a sweepline algorithm (see \cite[\S 8]{CLRS} for details). However, the memory usage is too costly in practice. 
}
\fi

Therefore, we propose a simple sampling scheme whose worst case variance is at most twice the variance of the optimal scheme.
Instead of using $h^*(p)$ to keep track of the maximum of all $h_v(p)$, we use an approximate $h'(p) = \max\{ h_{v_1}(p), h_{v_2}(p)\}$, where $v_1 = \argmax_v a_v^2 \mu^2$ and $v_2 = \argmax_v a_v(\mu_v^2 + \sigma_v^2)$ 
 to approximate $h^*(p)$.
 \iftechreport
Since $h_v(p)$ is a function in the form
of $h(p) = Ap + B/p + C$ for some constant $A, B, C > 0$, the value of $p^* = \argmin_{\{\epsilon_1, \epsilon_2 \leq p \leq 1\}} h'(p)$ can be easily  solved using quadratic equations and basic case analysis. \else
\rev{
The function $h'$ is much simpler and its minimum can be easily found using quadratic equations and basic case analysis.}
\fi
For more details on the algorithm, refer to
\iftechreport
Appendix~\ref{app:omitted_algo}.
\else
Appendix B in~\cite{approx-join-techreport}.
\fi
% Due to space constraints, we defer the details to our technical report~\cite{approx-join-techreport}.  

Let $p' = \argmin h'(p)$ and $p^* = \argmin h^*(p)$. 
\iftechreport
We claim that choosing
 $p'$ as our sampling parameter can only increase the optimal worst case variance by a factor of $2$. This is follows from the simple fact that $h_{v_1}(p)$ uppers bounds the terms in $h_{v}(p)$ that depends on $a_v^2\mu_v^2$, and $h_{v_2}(p)$ upper bounds the terms that depends on $a_v(\mu_v^2 + \sigma_v^2)$. Hence their maximum is at least half of $h_v(p)$, for any $v$ and $p$.
\else
    We have:
\fi

\begin{restatable}{lemma}{lemapprox}
\label{lem:approx}
For any $p$$\geq$$\epsilon_1,$$\epsilon_2$, we have 
$\frac{h^*(p)}{2}$$\leq$$h'(p)$$\leq$$h^*(p)$.
\end{restatable}

\begin{corollary}
We have:
$h^*(p') \leq 2 h^*(p^*)$.
\end{corollary}

\subsection{Average on Joins}
\label{sec:avg}

Let $\Eavg$ be the output of the following simplified query:
\begin{verbatim}
    select avg(T1.W) 
    from T1 join T2 on J
\end{verbatim}
In general, producing an unbiased estimator for \AVG is hard.\footnote{The denominator, \ie, the size of the sampled join, can even be zero. Furthermore, the expectation of a random variable's reciprocal is not equal to the reciprocal of its expectation.}
Instead, we define and analyze the following estimator. 
%\tofix{The estimator approaches the true value as it approaches infinity.}
Let $S$ and $C$ be the \SUM and \COUNT of column $W$ in $S_1 \bowtie S_2$. 
We define our estimator as $\Javg = S/C$. 
There are two advantages over using separate samples to evaluate \SUM and \COUNT: (1) we can use a larger sample to estimate both queries, and (2) since \SUM and \COUNT will be positively correlated, the variance of their ratio will be lower.
Due to the lack of a close form expression for the variance of the ratio of two random variables, next we present a first order bivariate Taylor expansion to approximate the ratio. 

%The detailed proof is deferred to the appendix. The following theorem give the approximate variance of the estimator $\Jagg$.
\ignore{
For any $f(X, Y)$, the bivariate Taylor expansion around $(\theta_x, \theta_y)$ is:
\[
    f(X, Y) = f(\theta_x, \theta_y) + f_x'(\theta_x, \theta_y)(x - \theta_x) + f'(\theta_x, \theta_y)(y - \theta_y) + R.
\]
where $R$ is a remainder of lower order terms. The expectation of $f(X, Y)$ can be approximated using the expansion around the expectation of $X$ and $Y$, $\mu_X$ and $\mu_Y$:
\begin{align*}
     &E[f(, Y)] \\ 
     \approx& E[f(\mu_X, \mu_Y) + f_x'(\mu_X, \mu_Y)(x - \mu_x) + f'(\mu_x, \mu_y)(y - \mu_y)] \\
        =& E[f(\mu_X, \mu_Y)] + 0 + 0 \\
        =& f(\mu_X, \mu_Y)].
\end{align*}

Let $S$ and $C$ be the random variables denoting the sum and the size of the join of samples, and let $f(X, Y) = X/Y$.
This shows that $\E[S/C] \approx \E[S]/\E[C]$, which is exactly equal to the average over join.
Although, the estimator is not unbiased, its expectation tends to the truth value when the join size tends to infinity. 

Now we can analyze its variance, and with some more involved analysis we can show that:
% DY: are we going to have separate tech report? commenting out below for now.
% \footnote{Further derivations and details are omitted to conserve space. Interested readers can refer to our technical report~\cite{anon_report}.}:
\begin{align}
    \var[S/C] &\approx (\frac{E[S]^2}{E[C]^2})( \frac{\var[S]}{E[S]^2} - \frac{2\cov[S, C]}{E[S]E[C]} + \frac{\var[C]}{\E[C]^2}).
\end{align}
}

\iftechreport
\else

\begin{restatable}{theorem}{thmvaravg}
 \label{thm:var-avg}
Let $S$ and $C$ be  random variables denoting the sum and cardinality of the join of two samples produced by applying UBS sampling parameters $(p_1, q_1)$ to $T_1$ and $(p_2, q_2)$ to $T_2$. Let $\pmin = \min\{p_1, p_2\}$. We have:
\small
\begin{align}
    \var[S/C] &= (\frac{E[S]^2}{E[C]^2})( \frac{\var[S]}{E[S]^2} - \frac{2\cov[S, C]}{E[S]E[C]} + \frac{\var[C]}{\E[C]^2}) + R
\end{align}
\normalsize
where $R$ is a remainder of lower order terms, and
\small
\begin{align*}
    \cov[S,C] =& \pmin q_1 q_2 \lbrack (1 - q_2)q_1 \sum_v a_v^2\mu_v b_v + (1-q_1)q_2\sum_v a_v\mu_v b_v^2  \\
    +&(1-q_1)(1-q_2) \sum_v a_v\mu_vb_v +(1-\pmin) q_1 q_2 \sum_v a_v^2 \mu_v b_v^2 \rbrack
\end{align*}
and other expectation and
variance terms are given by Theorems \ref{thm:var-count-centralize} and \ref{thm:var-sum-centralize}.
\normalsize
\end{restatable}

\fi

\subsubsection{Centralized Sampling for Average}
\label{sec:avg:centralized}
 In a centralized setting  
%we first substitute \barzan{what??}. it is true that for two fixed relations $T_1$ and $T_2$ 
where  $a_v$, $b_v$, $\mu_v$ and $\sigma_v$  are given for all $v$, every term in the expression $\frac{E[S]^2}{E[C]^2}( \frac{\var[S]}{E[S]^2} - 2\frac{\cov[S, C]}{E[S]E[C]} + \frac{\var[C]}{\E[C]^2}$ that depends on $p$ is proportional to either $p$ or $1/p$.\footnote{Notice that $E[S]/E[C]$ is independent of $p$.}
\iftechreport
\else

\rev{Thus, similar to Theorems~\ref{thm:opt-rate} and~\ref{thm:opt-rate-sum}, we can again find a $p$ that minimizes the variance given by Theorem~\ref{thm:var-avg}. We defer the exact expression of the optimal parameter to~\cite{approx-join-techreport}
due to space 
constraints.}

\fi

\subsubsection{Decentralized Sampling for Average}
\label{sec:avg:decentralized}

Minimizing the \emph{worst case} variance for \AVG (for the decentralized setting) is much more involved  than the average case. 
In most cases,
the objective function (variance) is neither convex nor concave in  $T_2$'s frequencies.
However, note that every term in 
\iftechreport
    Theorem~\ref{thm:opt-rate-avg}
\else
   Theorem~\ref{thm:var-avg} 
\fi
is an inner product $\langle x, y\rangle$, where $x$ and $y$ are two vectors stored on \partyone and \partytwo, respectively.
Fortunately, inner products can be approximated by transferring a very small amount of information using the AMS sketch\cite{AlonMS99, DobraGGR02}.
With such a sketch, we can derive an approximate sampling rate without communicating the full frequency statistics.

%\input{filters.tex}

%!TEX root = main.tex

\section{Multiple Queries and Filters}
\label{sec:one-sample}

% In previous section, we derived an optimal 
%     sampling strategy given an aggregate function, aggregation column, and join column. 
% However, creating a separate sample for each combination of aggregation function, join and aggregation columns, and \WHERE clause is clearly impractical. 
Creating a separate sample for each combination of aggregation function, aggregation column, and \WHERE clause is clearly impractical. 
In this section, we show how to create a single sample
        per join pattern that supports  multiple queries at the cost of some possible loss of approximation quality. 
First, we ignore the \WHERE clauses and then show how 
    they can be handled too.

\ignore{
In essence, one can take the union  of multiple universe samples to form a single sample 
    that is simultaneously useful for various join patterns.  
However, the resulting sample size  depends heavily
    on inter-column correlations
    and, in the worse case, can 
    be as large as creating a separate universe sample for each join pattern. 
To remove the need for capturing correlations, 
    here we present a formulation that
    creates a separate sample for every pair of table and join key (a join key can itself be one or more columns).
} 

\rev{
\ph{Multiple Tables and Queries}  
We formulate our input as a graph $G$=$\langle V,E\rangle$. The vertex set $V$ is the set of all table and join key pairs, and the edge set $E$ corresponds to all join queries of interest. 
Specifically,
for every join query between tables $T_1$ and $T_2$ on $J_1 = J_2$,
we have a corresponding edge $e$ between vertices $(T_1, J_1)\in V$ and $(T_2, J_2)\in V$ 
  (henceforth, we will use a query and its corresponding edge interchangeably).
This means $G$ is a multigraph, with potentially parallel edges or self-loops. 
For each vertex $v=(T,J)\in V$, we must output a sampling budget $\epsilon_v$ as well as the corresponding universe sampling rate $p_v$, which will be used to create 
    a sample $S = \mbox{UBS}_{p_v, \epsilon_v/p_v}(T, J)$.
    This sample will be used for any query that involves
    a join with $T$ on column(s) $J$.

According to Lemmas \ref{lem:jcount:var} and \ref{lem:jsum:var}, and Theorem~\ref{thm:var-avg}, for each edge 
$e$$=$$(v_1,$$v_2)$$\in$$E$, 
we can express the estimator variance of  its corresponding query as a function of $\epsilon_{v_1}, \epsilon_{v_2}, p_{v_1}, p_{v_2}$ and $p_e$, where $p_e$ is an auxiliary variable denoting the minimum of $p_1$ and $p_2$:
\begin{equation}
f_e(p, \epsilon_{v_1}, \epsilon_{v_2}, p_{v_1}, p_{v_2})\text{=}\frac{1}{p_e}(A_e\text{+}B_e \frac{p_1}{\epsilon_{v_1}}\text{+}C_e \frac{p_2}{\epsilon_{v_2}}\text{+}D_e \frac{p_1 p_2}{\epsilon_{v_1}\epsilon_{v_2}} )
\label{eq:error:form}
\end{equation}
where $A_e, B_e, C_e, D_e$ are constants that depend on the distributional information of the tables in $v_1$ and $v_2$.
To cast this as an optimization problem, we also take in a user specified weight $\omega_e$ for each edge $e$ and express our objective as:
\begin{equation} \label{eqn:multiple}
    F = \sum_{e = (v_1, v_2) \in E} \omega_e f_e(p_e, \epsilon_{v_1}, \epsilon_{v_2}, p_{v_1}, p_{v_2})
\end{equation}

The choice of $\omega_e$ values is up to the user. For example, they can  all be set to $1$, or to the relative frequency, importance, or
probability of appearance (e.g., based on  past workloads) of the query corresponding to $e$. 
Then, to find the optimal sampling parameters we   solve the following optimization:
\ignore{$\epsilon$'s and $p$'s, we minimize the objective $F$ subject to the constraint that $\underset{v=(T,J)\in V}{\Sigma} \epsilon_v \cdot \text{size}(T) \leq B$,} 
\begin{align} 
    \min_{\epsilon_v, p_v,p_e}~~F \text{~~~~~~subject to~~} \underset{v=(T,J)\in V}{\Sigma} \epsilon_v \cdot \text{size}(T) \leq B
    \label{eq:physical:design}
\end{align}
where $size(T)$ is the storage footprint of table $T$, and $B$ is the overall storage budget  for creating samples.
\cancut{Note that 
by replacing the non-linear $p_e=\min(p_{v_1},p_{v_1})$ constraints 
with $p_e \leq p_{v_1}$ and 
$p_e \leq p_{v_2}$, 
(\ref{eq:physical:design}) is reduced to 
    a smooth optimization problem,
        which can be solved numerically with off-the-shelf solvers~\cite{BV2014}.} }

\ph{Known Filters} To incorporate \WHERE clauses, 
we   simply regard a query with a filter $c$ on  $T_1 \bowtie T_2$ as a query without a filter but on a sub-table that satisfies $c$, namely $T'= \sigma_c(T_1 \bowtie T_2)$. 

\ph{Unknown Filters with Distributional Information} When the columns appearing in the \WHERE clause can be predicted but the exact constants are unknown, a similar technique can be applied. For example, if an equality constraint $C>x$  is anticipated but $x$ may take on $100$ different values, we can \emph{conceptually} treat it as $100$ separate queries, each with a different value of $x$ in its \WHERE clause. This reduces our problem 
\ignore{of sampling for a query with unknown values in its \WHERE clause} 
to that of sampling for multiple queries without a \WHERE clause, which we know how to handle using  equation~\eqref{eqn:multiple}.\footnote{Note that, even though each query in this case is on a different table, 
    they are all sub-tables of the same original table, and hence their sampling rate $p$ is the same.}
Here, the weight $\omega_i$ can be used to exploit 
any  distributional information that might be available. In general, 
    $\omega_i$ should be set to reflect the 
    probability of each possible \WHERE clause appearing in the future. 
For example, if there are
$R$ possible \WHERE clauses and all are equally likely,
we can set $\omega_i=1/R$, but if popular values in a column are more likely to appear in the filters, we can use the column's histogram to assign $\omega_i$. 

\ph{Unknown Filters}
When there is no information regarding the columns (or their values) in future filters, we  can take a different approach. 
Since  the estimator variance   is a monotone  function in the frequencies of each join key (see Theorem~\ref{thm:var-count-centralize},  Theorems~\ref{thm:var-sum-centralize} and \ref{thm:var-avg}), 
the larger the frequencies, the larger the variance. This means the worst case variance always happens when the \WHERE clause selects all tuples from the original table. Hence, in the absence of any distributional information regarding future \WHERE clauses,  
we can simply focus on the original query  without any filters to minimize our worst case variance.

\section{Experiments}
\label{sec:expr}

\tempcut{
Our experiments aim to answer the following questions:

\begin{enumerate}[label={(\roman*)},leftmargin=0.2em, itemindent=2.3em, topsep=0.1em, itemsep=0.1em]
    \item How does our optimal  sampling  compare to other baselines in  centralized   and  decentralized settings? 
    (\S \ref{sec:expr:centralized}, \S\ref{sec:expr:decentralized})
    \item How well does  our optimal UBS  sampling handle join queries with filters?
    (\S\ref{sec:expr:filters})
    \item How does our optimal UBS sampling perform when using a single sample for multiple queries?
    (\S\ref{sec:expr:combined})
    \item How does our optimal SUBS sampling compare to existing stratified sampling strategies?
    (\S\ref{sec:expr:stratified})
    \item How much does a decentralized setting reduce the
    resource consumption and sample creation overhead?    (\S\ref{sec:expr:sample_creation})
    \iftechreport
    \item How does our optimal  sampling  compare to a more general sampling
        scheme, namely two-level sampling~\cite{two-level-sampling}?
    (\S\ref{sec:expr:2lv})
    \fi

\end{enumerate}
}

\subsection{Experiment Setup}
\label{sec:expr:setup}

\ph{Hardware and Software}
We borrowed a cluster of 18 \textit{c220g5} nodes from  CloudLab~\cite{cloudlab}. 
%	a publicly-available  cluster  for scientific research.
 Each node was equipped with an Intel Xeon Silver 4114 processor with 10  cores (2.2Ghz each) and 192GB of RAM.
We used Impala 2.12.0 as our backend database to store   data  and execute queries.
% \barzan{u need 1) at least 3 different systems not just impala. 2) the reviewers will ask us to do a head-to-head comparison with an existing AQP engine too!}

%ta inja
\ph{Datasets} We used several real-life and synthetic datasets:
\begin{enumerate}[topsep=0.12cm,itemsep=0.05cm,leftmargin=0.2cm]
\item  \textbf{\instacart~\cite{instacart}.}
This is a real-world dataset from an online grocery. 
We used their \textit{orders} and \textit{order\_products} tables (3M and 32M tuples, resp.), 
joined on  \textit{order\_id}.

\item \textbf{\movielens~\cite{harper2016movielens}.}
This is a real-world movie rating dataset.
We used their %\textit{latest-full} dataset, and joined their 
\textit{ratings}  and \textit{movies} tables (27M and 58K tuples, resp.), joined on \textit{movieid}.

\item \textbf{\tpch~\cite{tpch}}.
%We generated a standard TPC-H dataset with 
We used a scale factor of \rev{1000}, and joined
\textit{l\_orderkey} of 
the fact table (\textit{lineitem}, 6B tuples) 
with \textit{o\_orderkey} of the 
largest dimension table (\textit{orders}, 1.5B tuples).
\ignore{Specifically, we joined  \textit{l\_orderkey} of
the
\textit{lineitem} table with \textit{o\_orderkey}  of the \textit{orders} table (\rev{6B and 1.5B} tuples, resp.). }

\item \textbf{\synthetic.}
 To better control the join key distribution, 
	we also generated several synthetic datasets,
	where tables $T_1$ and $T_2$ each had
		$100$M tuples and a join column $J$.
$T_1$ had an additional column $W$  for aggregation,  
 drawn from a power law distribution with range $[$1, 1000$]$ and $\alpha$$=$$3.5$.
 We varied the distribution of the join key in each table to be one of uniform, normal, or power law, creating different datasets (listed in Table~\ref{tab:sampling_rate_cent}).
 The values of column $J$ were  integers randomly drawn from $[$1, 10M$]$ according to the chosen distribution. 
Whenever joining with \emph{power2} (see below), we used 100K join keys in both relations.
For normal distribution, we used a truncated distribution with $\sigma$$=$$1000/5$.
We used two different variants of power law distribution for $J$, 
one with $\alpha$$=$$1.5$ and 10M join keys (referred to as \emph{power1}), and one with $\alpha$$=$$2.0$ and 100K join keys (referred to as \emph{power2}). 
We denote each synthetic dataset according to its tables' distributions, 
\textit{\texttt{S}\{distribution of $T_1$,distribution of $T_2$\}}, \eg,
\syn{uniform}{uniform}.
 
\end{enumerate} 

\ignore{\ph{Queries}
% \barzan{rename this para into `Queries' and explain the actual queries used instead of just saying which aggregates you used. 
% Make sure your queries are not so trivial}
For each dataset, 
we tested a set of three queries that join two tables and calculate aggregates.
The aggregates in each of three queries for a dataset correspond to
one of the three most frequently used (non-extreme) aggregates,
\COUNT, \SUM, and \AVG, respectively.
}
% We experimented with the three most frequently used (non-extreme) aggregates: \{\COUNT, \SUM, \AVG\}.

\ph{Baselines}
We compared our optimal \UBS parameters (referred to as \OPT)
against six baselines. The \UBS parameters of these baselines,
	$B_1,\ldots,B_6$, 	
	are listed in Table \ref{tab:sample_baseline}.
$B_1$ and $B_6$ are simply universe and uniform  sampling, respectively.
 $B_2,\ldots,B_5$ represent different hybrid variants of these sampling schemes.
Sampling budgets were $\epsilon_1 = \epsilon_2 = 0.001$, except for \instacart and \movielens where, due to their small number of tuples,
we used $0.01$ and $0.1$, respectively.
\iftechreport
 In Section~\ref{sec:expr:2lv}, we also compare against a more general scheme, called \emph{two-level sampling}~\cite{two-level-sampling}, which utilizes significantly more parameters than \UBS (see Section~\ref{sec:related} for an overview of two-level sampling).
\fi

% Sampling budgets were $\epsilon_1 = \epsilon_2 = 0.001$, except for \tpch and \movielens where we used \rev{$0.001$} and $0.1$ due to their large and small number of tuples, respectively.

\begin{table}[t]
  \caption{\textbf{Six \UBS baselines, each with different   $p$  and $q$.}}
  \inv
  \centering
  \begin{tabular}{l|c|c|c|c|c|c}
  \hline
  & $B_1$ & $B_2$ & $B_3$ & $B_4$ & $B_5$ & $B_6$ \\
  \hline
  $p$ & 0.001 & 0.0015 & 0.003 & 0.333 & 0.6667 & 1.000 \\
  $q$ & 1.000 & 0.6667 & 0.333 & 0.003 & 0.0015 & 0.001 \\
  \hline
  \end{tabular}
  \label{tab:sample_baseline}
  \inv
\end{table}

\ph{Implementation} 
We implemented our optimal parameter calculations in Python application.
Our sample generation logic read required information, such as table size and join key frequencies, from the database, and then constructed SQL statements to build appropriate samples in the target database. We used Python to compute approximate answers from sample-based queries.

\ph{Variance Calculations}
We generated $\beta$$=$$500$ pairs of samples for each experiment, and 
re-ran the queries on each pair, to calculate the variance of our approximations. 

\subsection{Join Approximation: Centralized Setting}
\label{sec:expr:centralized}

\begin{figure*}[t]
\centering
\includegraphics[width=\textwidth]{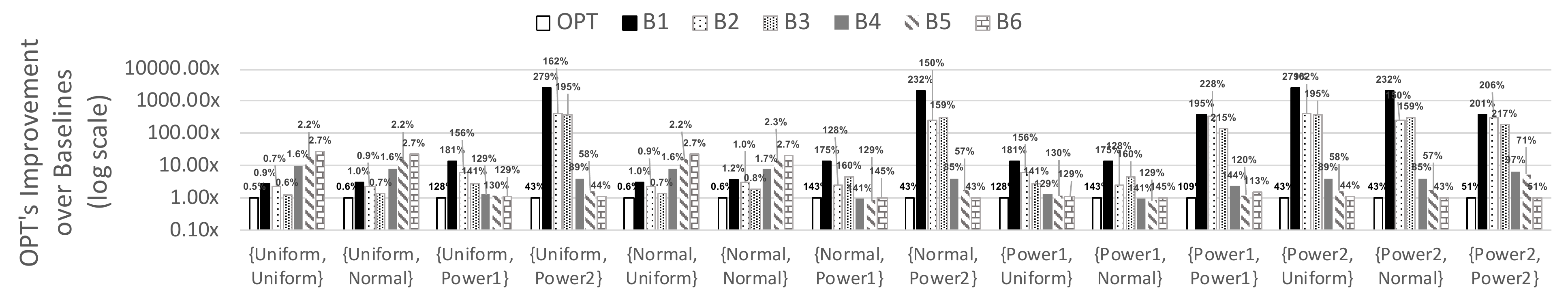}
\inv\inv\inv\inv
\caption{\textbf{\OPT's improvement in terms of variance for \COUNT over six baselines with synthetic dataset (percentages are relative error).}}
\label{fig:variance_count_synthetic_cent}
\inv\inv
\end{figure*}

\begin{figure*}[t]
\inv\inv
\centering
\includegraphics[width=\textwidth]{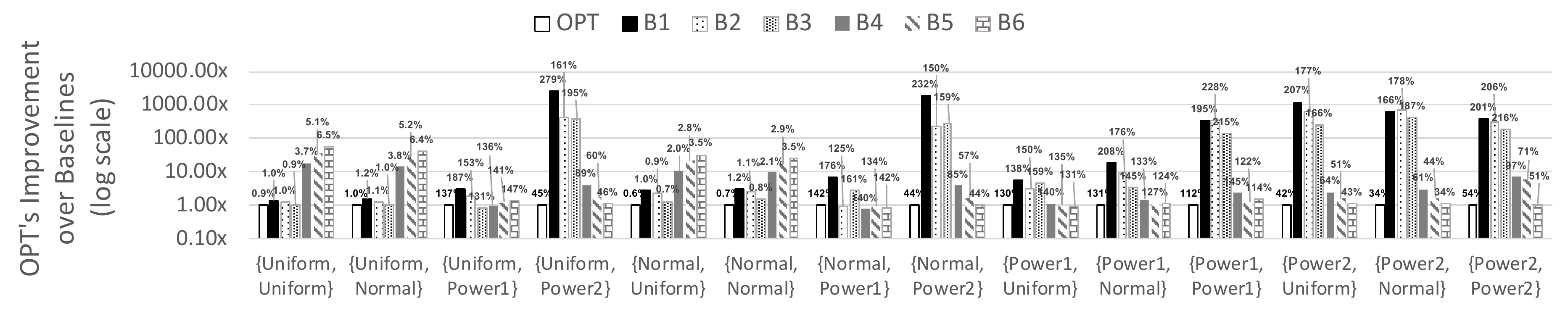}
\inv\inv\inv\inv
\caption{\textbf{\OPT's improvement in terms of variance for \SUM over six baselines with synthetic dataset (percentages are relative error).}}
\label{fig:variance_sum_synthetic_cent}
\inv
\end{figure*}

\begin{figure*}[t]
\centering
\includegraphics[width=\textwidth]{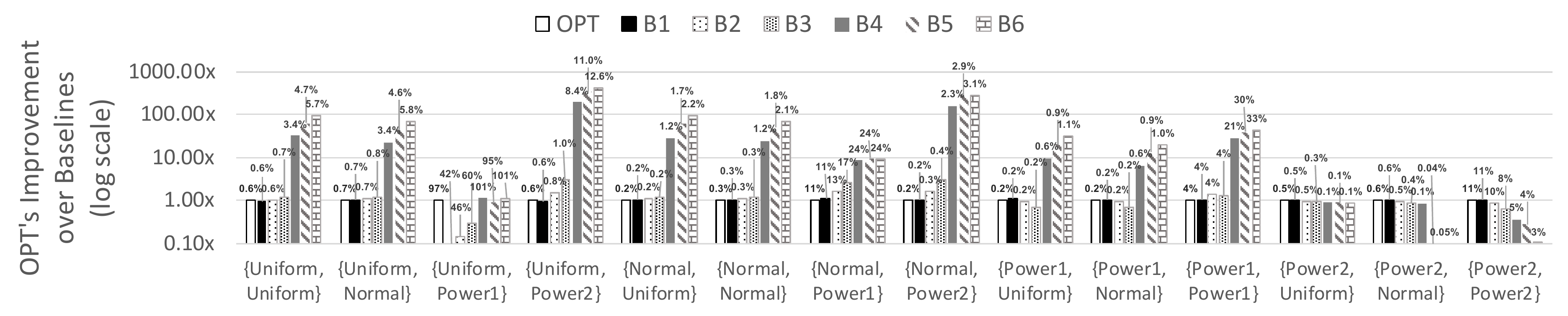}
\inv\inv\inv\inv
\caption{\textbf{\OPT's improvement in terms of variance for \AVG over six baselines with synthetic dataset (percentages are relative error).}}
\label{fig:variance_avg_synthetic_cent}
\inv
\end{figure*}

% \begin{figure*}[t]
% \centering
% \includegraphics[width=\textwidth]{./fig/others_cent.pdf}
% \inv\inv\inv\inv
% \caption{Variances of the query estimators for all three aggregates from \OPT and six baselines with \instacart, \movielens and \tpch datasets.}
% \label{fig:variance_others_cent}
% \inv
% \end{figure*}
\begin{figure*}[!ht]
  \centering
  \begin{subfigure}{0.32\textwidth}
    \includegraphics[width=\textwidth]{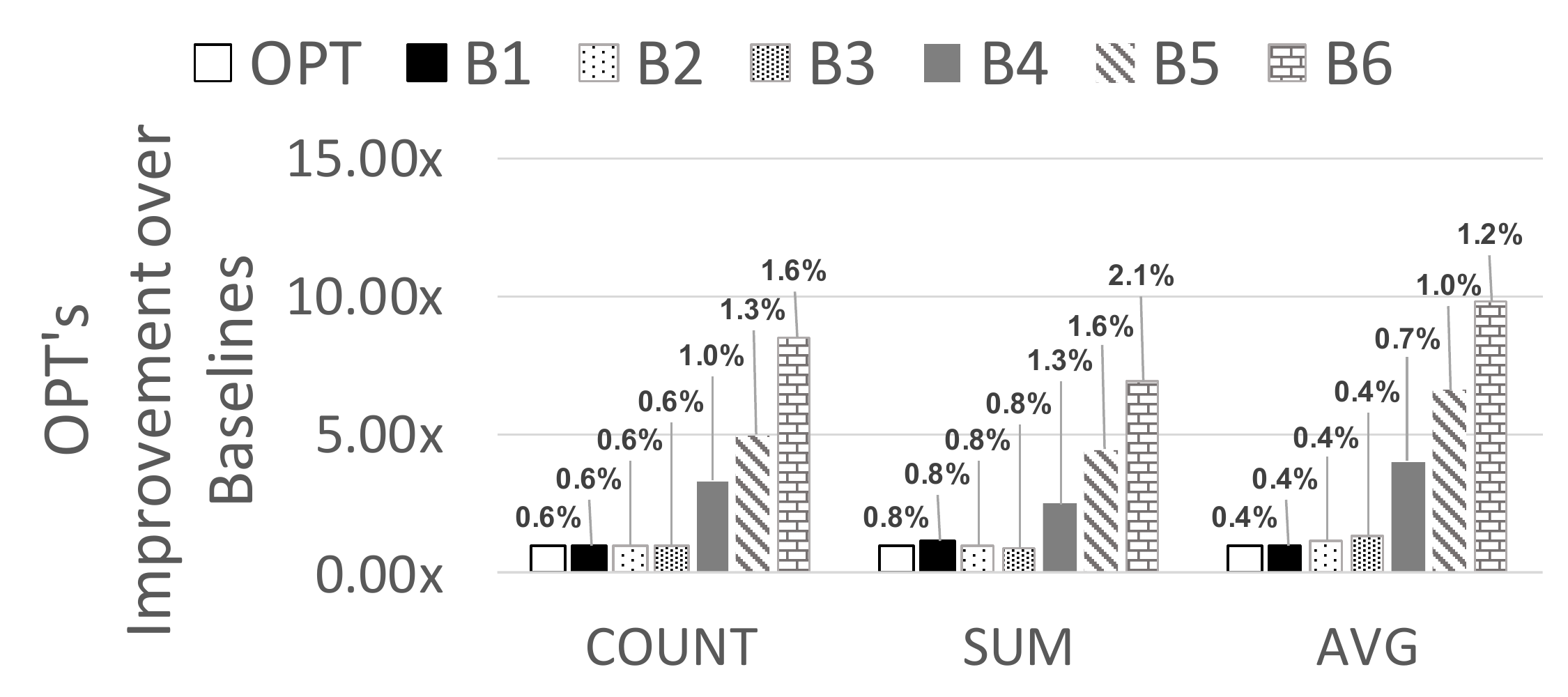}
    \inv\inv\inv
    \caption{\instacart}
    \label{fig:instacart_var}
  \end{subfigure}
  \begin{subfigure}{0.32\textwidth}
    \includegraphics[width=\textwidth]{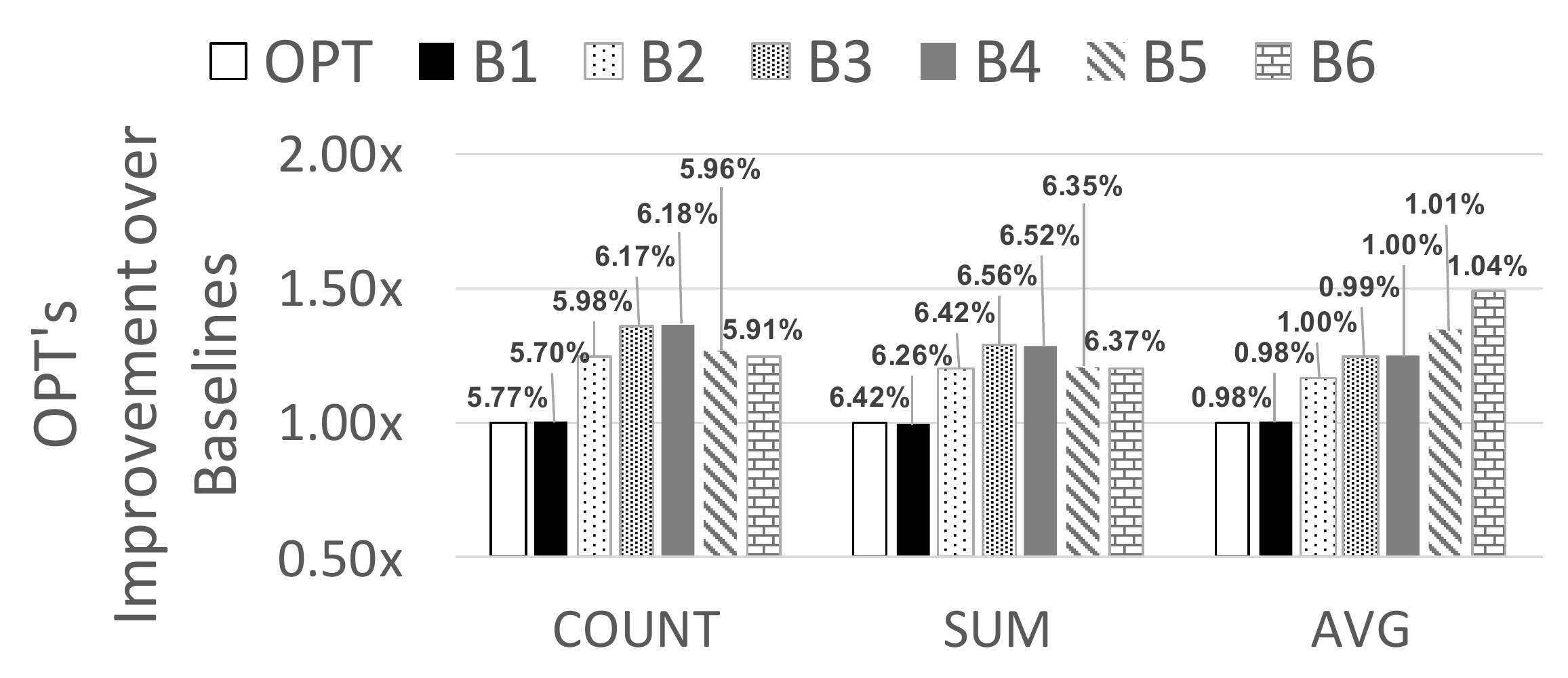}
    \inv\inv\inv
    \caption{\movielens}
    \label{fig:movielens_var}
  \end{subfigure}
  \begin{subfigure}{0.32\textwidth}
    \includegraphics[width=\textwidth]{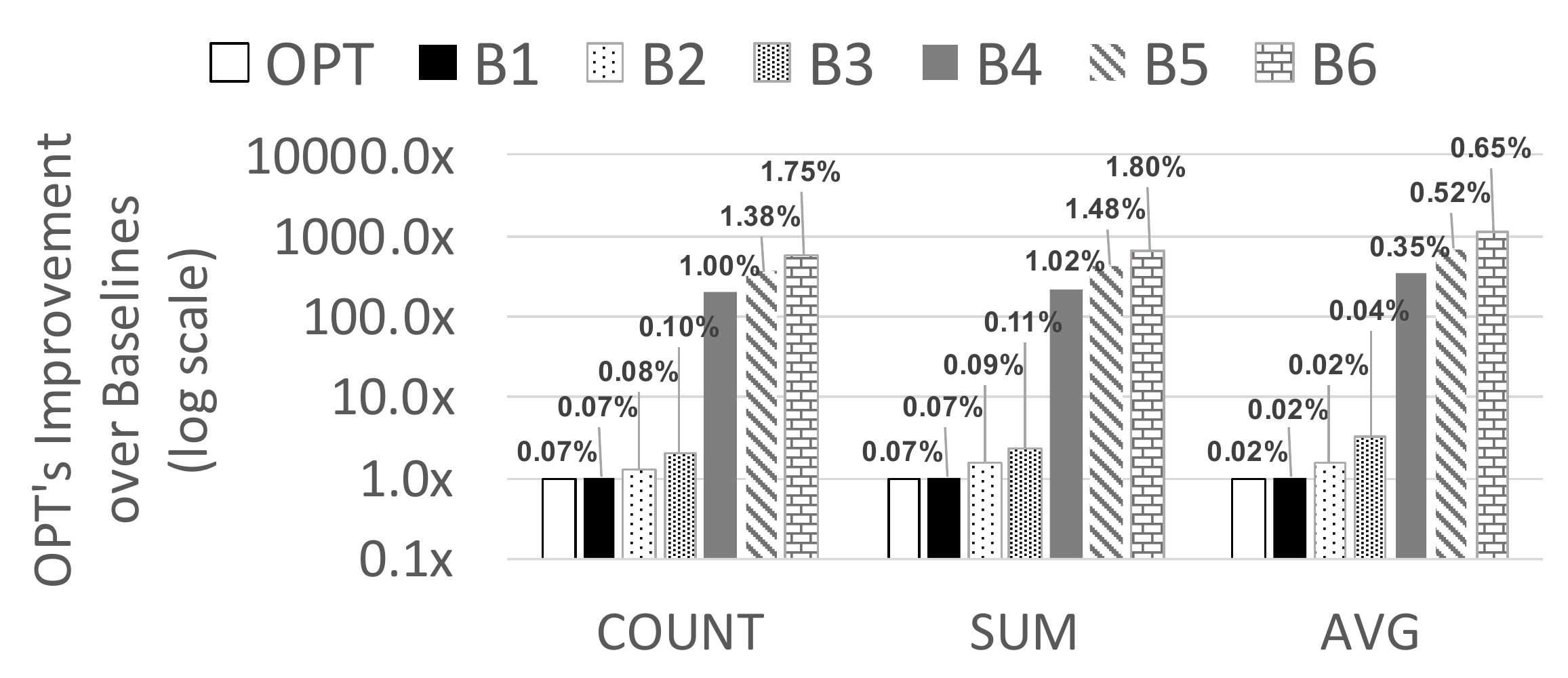}
    \inv\inv\inv
    \caption{\tpch}
    \label{fig:tpch_var}
  \end{subfigure}
   \inv\sinv
  \caption{\textbf{\OPT's improvement in terms of variance over the baselines on benchmark datasets (percentages are relative error).}}
  \label{fig:variance_others_cent}
  \inv
\end{figure*}

\begin{table}[t]
\caption{\textbf{Optimal sampling parameters (centralized setting).}}
\inv
\centering
\small
\begin{tabular}{c|p{0.45cm}p{0.55cm}|p{0.45cm}p{0.55cm}|p{0.45cm}p{0.55cm}}
\hline
Dataset &  \multicolumn{2}{|c|}{\COUNT} & \multicolumn{2}{|c|}{\SUM} & \multicolumn{2}{|c}{\AVG} \\
\hline
{}   & $p$   & $q$    & $p$   & $q$ & $p$ & $q$ \\
\syn{uniform}{uniform}   & 0.010 & 0.1   & 0.004  & 0.264 & 0.001 & 1.000 \\
\syn{uniform}{normal}   & 0.012 & 0.083   & 0.005  & 0.220 & 0.001 & 1.000 \\
\syn{uniform}{power1}   & 1.000 & 0.001   & 1.000  & 0.001 & 0.692 & 0.001 \\
\syn{uniform}{power2}   & 1.000 & 0.001  & 1.000 & 0.001 & 0.001  & 1.000 \\
\syn{normal}{uniform}   & 0.012 & 0.083   & 0.009  & 0.111 & 0.001 & 1.000 \\
\syn{normal}{normal}    & 0.014 & 0.069   & 0.011  & 0.093 & 0.001 & 1.000 \\
\syn{normal}{power1}   & 1.000 & 0.001   & 1.000  & 0.001 & 0.001 & 1.000 \\
\syn{normal}{power2}   & 1.000 & 0.001 & 1.000 & 0.001 &  0.001 & 1.000  \\
\syn{power1}{uniform}   & 1.000 & 0.001   & 1.000  & 0.001 & 0.001 & 1.000 \\
\syn{power1}{normal}    & 1.000 & 0.001   & 1.000  & 0.001 & 0.001 & 1.000 \\
\syn{power1}{power1}   &  1.000 & 0.001   & 1.000  & 0.001 & 0.001 & 1.000 \\
\syn{power2}{uniform}   & 1.000 & 0.001  & 1.000 & 0.001 & 0.001  & 1.000  \\
\syn{power2}{normal}   & 1.000 & 0.001 & 1.000 & 0.001 & 0.001 & 1.000 \\
\syn{power2}{power2}   & 1.000 & 0.001  & 1.000 & 0.001 & 0.001 & 1.000  \\
\instacart & 0.01 & 1.00 & 0.01 & 1.00 & 0.01 & 1.00 \\
\movielens & 0.1 & 1.00 & 0.1 & 1.00 & 0.1 & 1.00 \\
\tpch & 0.001 & 1.00 & 0.001 & 1.00 & 0.001 & 1.00 \\
\hline
\end{tabular}
\inv\inv
\label{tab:sampling_rate_cent}
\end{table}

% DY: backup
% \syn{uniform}{uniform}   &  0.10 & 0.10   & 0.10  & 0.10 & 0.01 & 1.00 \\
% \syn{uniform}{normal}   &  0.12 & 0.08   & 0.11  & 0.09 & 0.01 & 1.00 \\
% \syn{uniform}{power law}  &  1.00 & 0.01   & 1.00 & 0.01 & 0.01 & 1.00 \\
% \syn{normal}{uniform}   &  0.12 & 0.08   & 0.10  & 0.10 & 0.01 & 1.00 \\
% \syn{normal}{normal}   &  0.15 & 0.07   & 0.13  & 0.08 & 0.01 & 1.00 \\
% \syn{normal}{power law}   &  1.00 & 0.01   & 1.00  & 0.01 & 0.01 & 1.00 \\
% \syn{power law}{uniform}   &  1.00 & 0.01   & 1.00  & 0.01 & 0.01 & 1.00 \\
% \syn{power law}{normal}   &  1.00 & 0.01   & 1.00  & 0.01 & 0.01 & 1.00 \\
% \syn{power law}{power law}   &  1.00 & 0.01   & 1.00  & 0.01 & 0.01 & 1.00 \\
% \instacart & 0.01 & 1.00 & 0.01 & 1.00 & 0.01 & 1.00 \\
% \movielens & 0.1 & 1.00 & 0.1 & 1.00 & 0.1 & 1.00 \\
% \tpch & 0.001 & 1.00 & 0.001 & 1.00 & 0.001 & 1.00 \\

% \begin{figure*}[t]
% \centering
% \includegraphics[width=.3\textwidth]{./fig/sum_cent1.pdf}\quad
% \includegraphics[width=.3\textwidth]{./fig/sum_cent2.pdf}\quad
% \includegraphics[width=.3\textwidth]{./fig/sum_cent3.pdf}

% \medskip

% \includegraphics[width=.3\textwidth]{./fig/sum_cent4.pdf}\quad
% \includegraphics[width=.3\textwidth]{./fig/sum_cent5.pdf}
% \inv\inv
% \caption{Variances of the query estimator for \SUM from different sampling schemes.}
% \label{fig:variance_sum_cent}
% \inv\inv
% \end{figure*}

Table~\ref{tab:sampling_rate_cent} shows the sampling rates used by \OPT for each dataset and aggregate function in the centralized setting.
 For \synthetic,
the optimal parameters were some mixture of uniform and universe sampling when both tables were only moderately skewed (\ie, uniform or normal distributions) for \COUNT and \SUM, whereas it reduced to a simple uniform sampling for power law distribution. 
% \rev{
% the optimal parameters were a full uniform sampling for \COUNT and \SUM, and a full universe sampling for \AVG in most cases.
% The only exceptions were \COUNT and \SUM for \syn{power law}{power law} where the optimal parameters were some mixture of uniform and universe sampling.
This is due to the higher probability of missing extremely popular join keys with universe sampling.
To the contrary, for \AVG, \OPT reduced to a simple universe sampling in most cases. 
% \barzan{regardless of the distribution?}.
This is because maximizing the output size in this case was the best way to reduce  variance.
For the other datasets (\instacart, \movielens, and \tpch),
the optimal parameters led to universe sampling, regardless of aggregate type,
and their joins were PK-FK, 
hence making  uniform sampling less useful for the table with primary keys.

% \barzan{plz make up your mind if you want to use ~ before \ref{??} and \cite{??} and then be consistent with it.}

Figure~\ref{fig:variance_count_synthetic_cent} shows \OPT's improvement 
over the baselines 
in terms of variance for \COUNT queries.
\rev{Each bar is also annotated with the relative percentage error of the corresponding baseline.}
\OPT outperformed all   baselines in most cases, achieving over 10x lower variance than the worst baseline.
Figures~\ref{fig:variance_sum_synthetic_cent} and~\ref{fig:variance_avg_synthetic_cent} show the same experiment for \SUM and \AVG.
In both cases, \OPT achieved the minimum variance across all sampling strategies,
except for \AVG when \rev{$T_1$ or $T_2$} was a power law distribution. 
This is because \OPT for \AVG was calculated using a Taylor approximation, which is accurate only when the estimators of \SUM and \COUNT are both within the proximity of their true values.
 Moreover, sample variance  converges slowly to the theoretical variance,  particularly for skew distributions, such as power law.
This is why  estimated variances for \OPT were not  optimal for some \synthetic datasets. 
However, \OPT still achieved the lowest variance across all real-world datasets, as shown in 
Figure~\ref{fig:variance_others_cent}. Here, for the selected join key, \OPT determined that a full universe sampling was the best sampling scheme.

In summary, this experiment highlights  \OPT's ability in outperforming simple uniform or universe sampling---or choosing one of them, when optimal---for aggregates on joins.

\subsection{Join Approximation: Decentralized}
\label{sec:expr:decentralized}

% \barzan{the biggest issue with this section is that you just assume ppl know why you're measuring worst-case performance. We have 2 options which we should consider very carefully: 1) either explain very clearly why unlike the centralized setting the more relevant metric is the worst-case 
% or more safely, 2) report both average and worst-case scenarios.}

We  evaluated both \OPT and other baselines under a decentralized setting using \instacart and \synthetic datasets.
% 	More specifically, we evaluated how well  \OPT performs under a worst case scenario, when only limited statistics can be exchanged 	
% 		between the nodes hosting the two tables. \barzan{is prev sentence what u meant by `worst case' or did you meant worst case across 
% 		all possible data in the other table? }
% \tofix{This time, we considered all 9 possible permutations as they are not symmetric, unlike in the centralized setting.}
% \barzan{why are they NOT symmetric in this case? u need to explain}
% \tofix{We also added three additional combinations: (uniform,uniform\_max), (normal,normal\_max),  and (powerlaw,powerlaw\_max),
% where \textit{`\_max'} distributions were expected to result in a worse approximate aggregate variance given the distribution of $T_1$ when compared to other distributions that we had for $T_2$.} \barzan{last 2 sentences are not english}
Here, 
we constructed a possible worst case distribution for $T_2$ that was still somewhat realistic, 
given the distribution of $T_1$ and minimal information about $T_2$ (i..e, $T_2$'s cardinality). To do this, we used the following steps:
1) let $J_{MAX(T_1)}$ be the most frequent join key value in $T_1$; 
2) assign 75\% of the join key values of $T_2$ to have the value of $J_{MAX(T_1)}$ and draw the rest of the join key values from a uniform distribution.

 Figure~\ref{fig:variance_synthetic_dec} shows the results.
For \synthetic, the \OPT was the same   under both settings whenever
there was a power law distribution or the aggregate was \AVG.
This is because our assumption of the worst case distribution for $T_2$ was close to a power law distribution.
For \COUNT and \SUM with \synthetic dataset,
\OPT in the decentralized setting  had a much higher variance than \OPT in the centralized setting when there was no power law distribution.
With \instacart, 
\OPT in the decentralized setting was the same as \OPT in the centralized setting, which had the minimum variance among the baselines.
This illustrates that \OPT in the decentralized setting can perform well with real-world data where the joins are mostly PK-FK.
This also shows that if a reasonable assumption is possible on the distribution of $T_2$, \OPT can be as effective in the decentralized setting as it is in a centralized one, while requiring significantly less communication.

\begin{figure*}[!ht]
  \centering
  \begin{subfigure}{0.32\textwidth}
    \includegraphics[width=\textwidth]{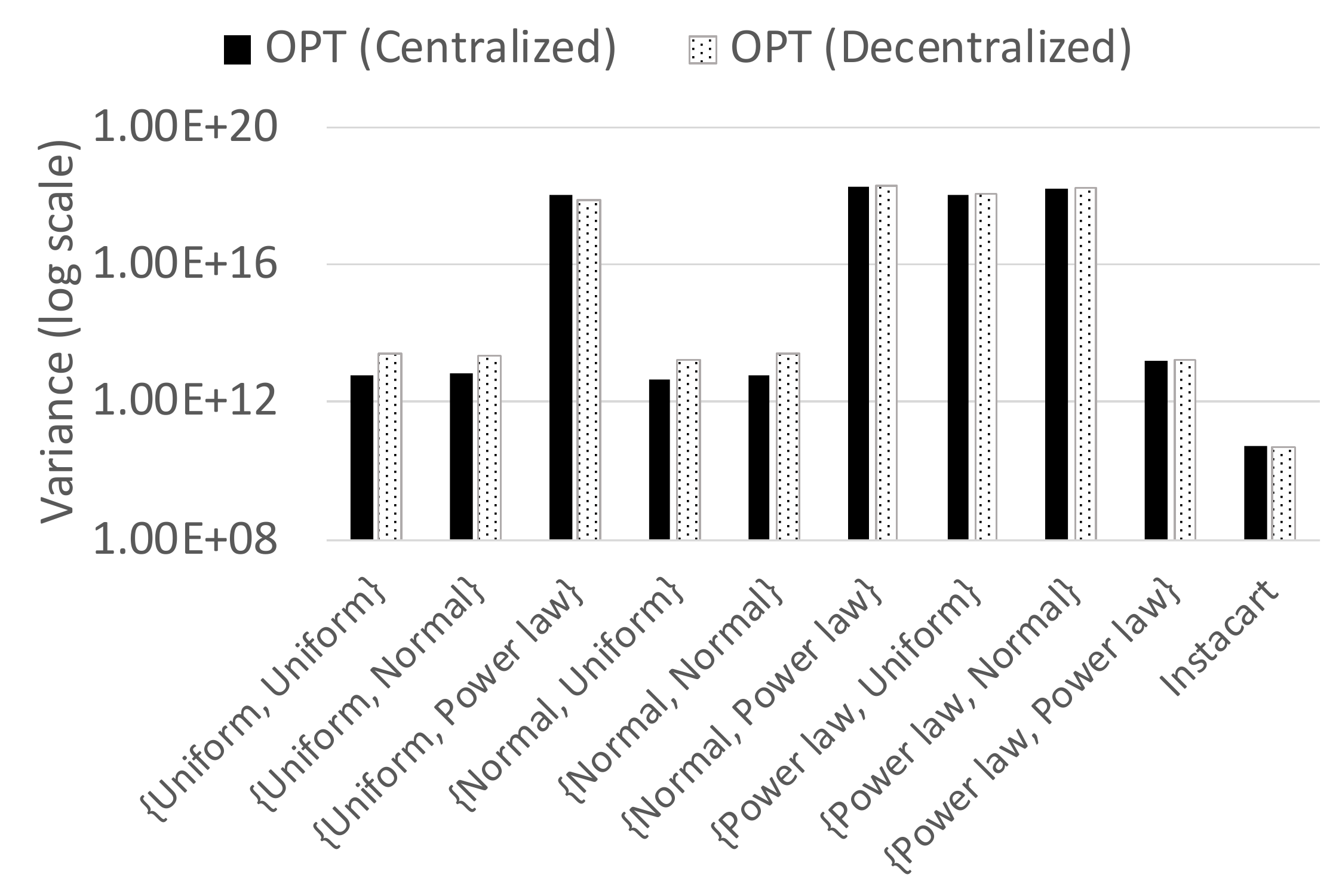}
    \inv\inv\inv
    \caption{\COUNT}
    \label{fig:variance_synthetic_count_dec}
  \end{subfigure}
  \begin{subfigure}{0.32\textwidth}
    \includegraphics[width=\textwidth]{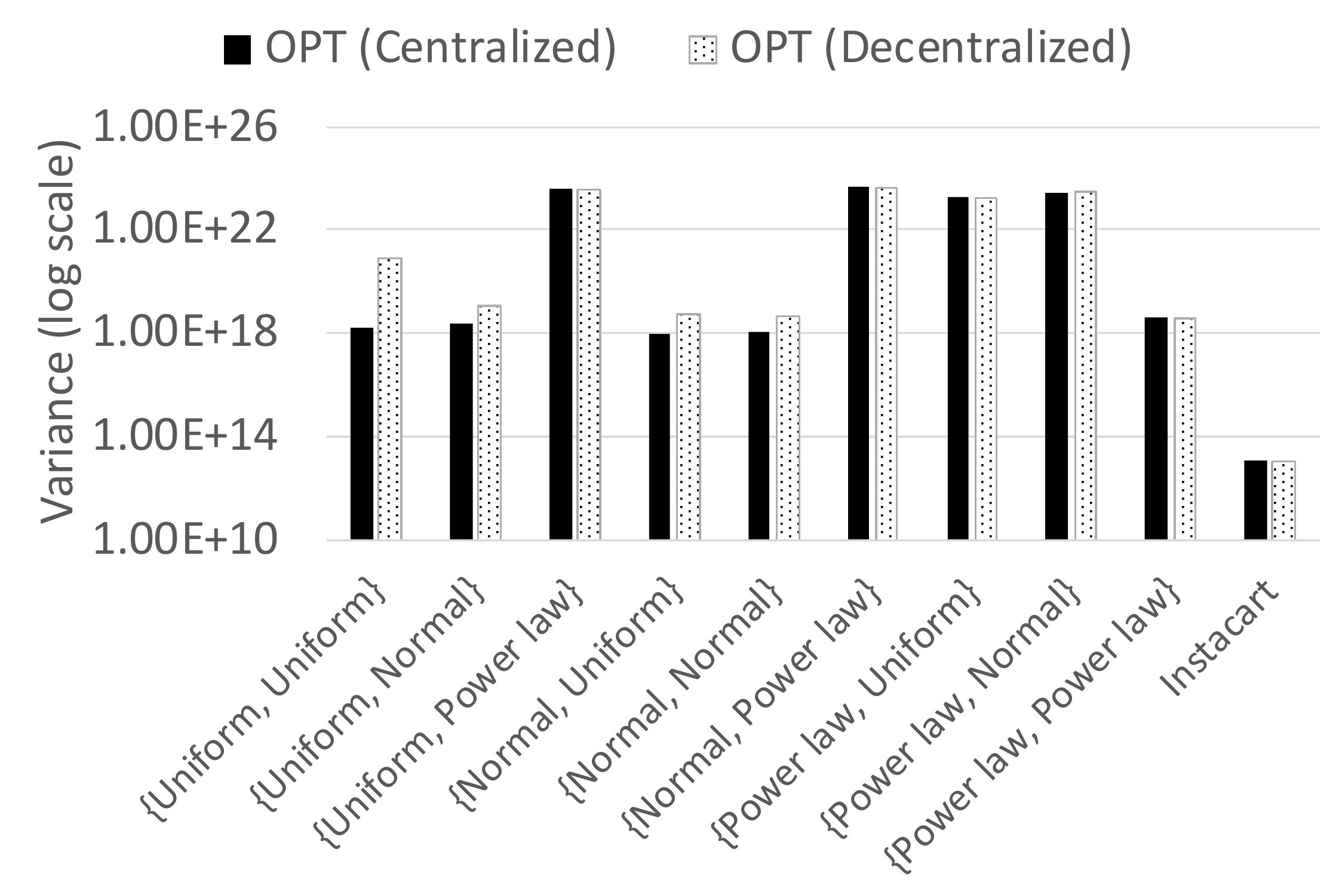}
    \inv\inv\inv
    \caption{\SUM}
    \label{fig:variance_synthetic_sum_sec}
  \end{subfigure}
  \begin{subfigure}{0.32\textwidth}
    \includegraphics[width=\textwidth]{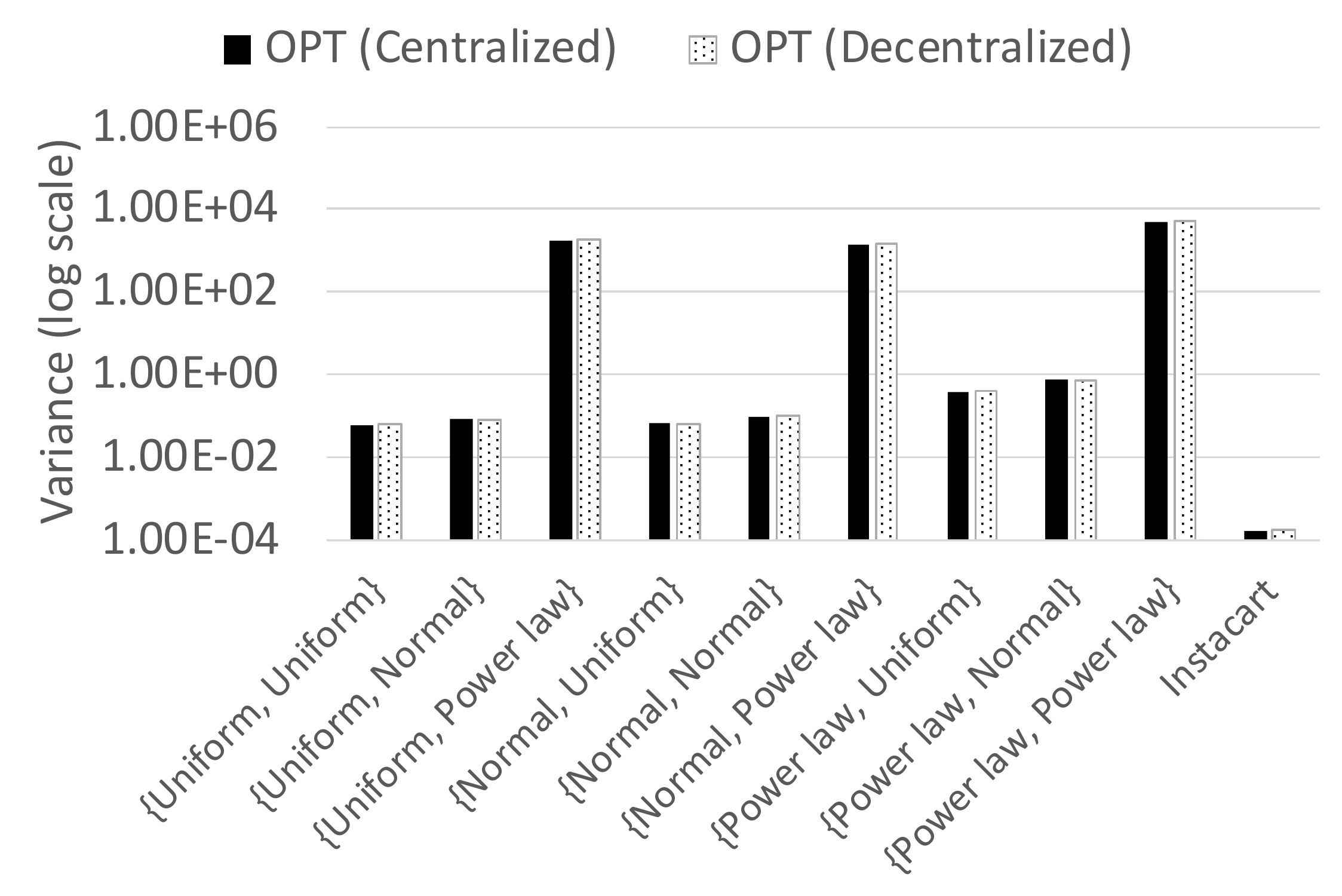}
    \inv\inv\inv
    \caption{\AVG}
    \label{fig:variance_synthetic_avg_dec}
  \end{subfigure}
   \inv\sinv
  \caption{\textbf{Variances of the query estimators for \OPT in the centralized and decentralized settings.}}
  \label{fig:variance_synthetic_dec}
\end{figure*}

% Figure \ref{fig:dist_var} shows the \barzan{average or worst?} variances of \OPT and other baselines strategies (Table~\ref{tab:sample_baseline} 
% 	 in the decentralized setting.
% This result demonstrates that \OPT does choose the best sampling rates for all aggregate functions  in terms of 
% 	the worst case \tofix{blah?},  where there is only limited communication between the distributed nodes hosting $T_1$ and $T_2$.

\subsection{Join Approximation with Filters}
\label{sec:expr:filters}

\begin{figure*}[!ht]
  \centering
  \begin{subfigure}{0.23\textwidth}
    \includegraphics[width=\textwidth]{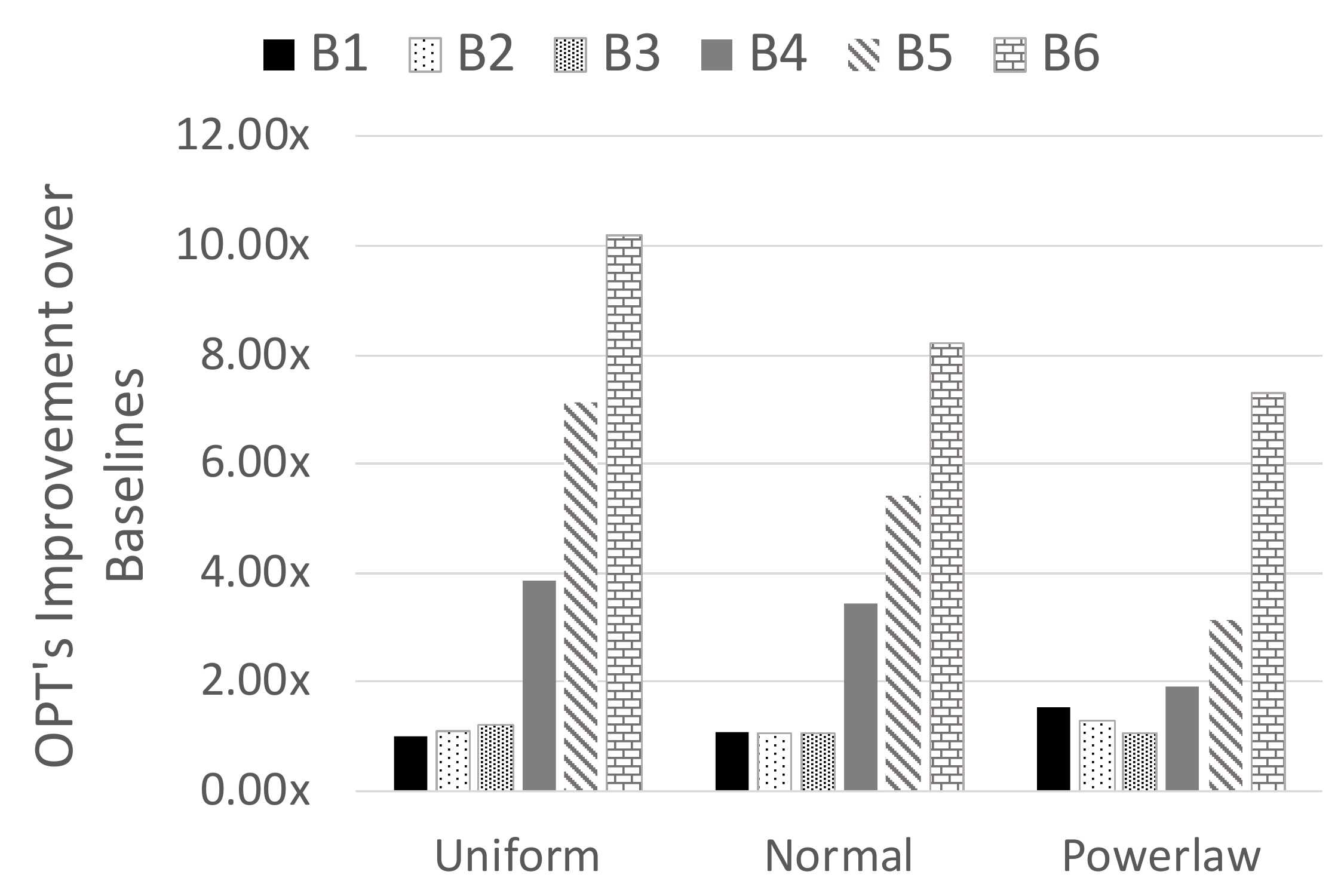}
    \inv\inv\inv
    \caption{\COUNT}
    \label{fig:filter_count}
  \end{subfigure}
  \begin{subfigure}{0.23\textwidth}
    \includegraphics[width=\textwidth]{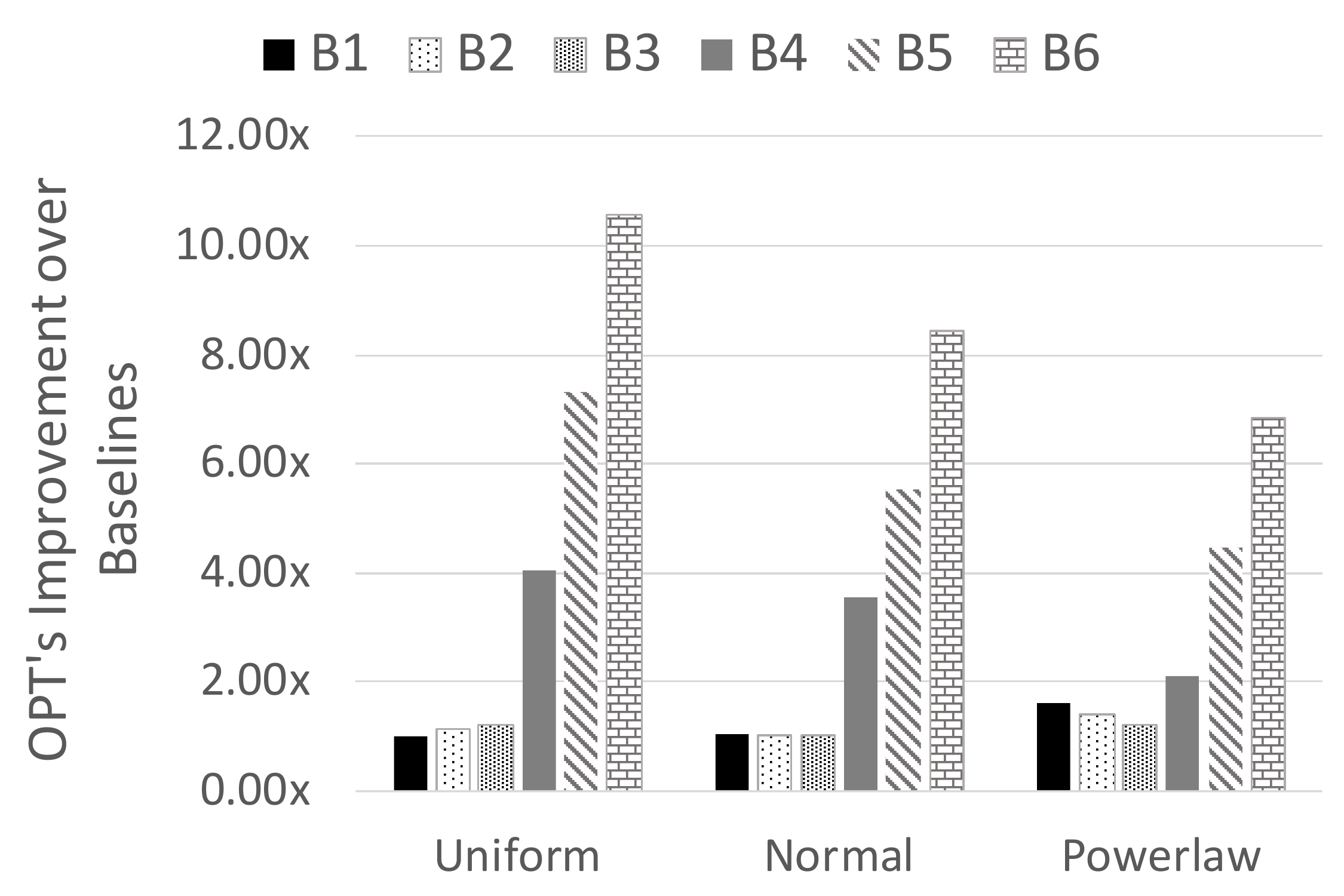}
    \inv\inv\inv
    \caption{\SUM}
    \label{fig:filter_sum}
  \end{subfigure}
  \begin{subfigure}{0.23\textwidth}
    \includegraphics[width=\textwidth]{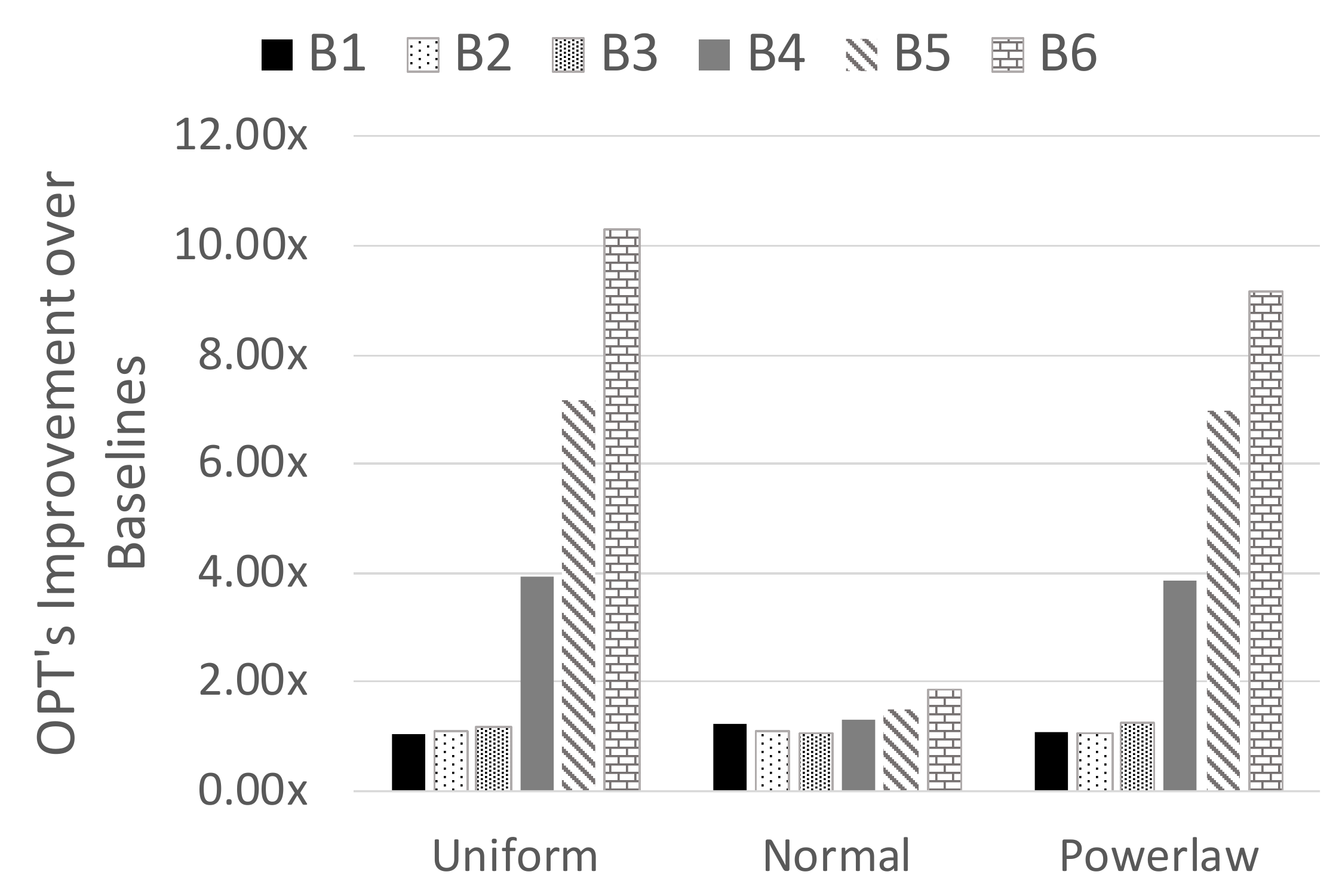}
    \inv\inv\inv
    \caption{\AVG}
    \label{fig:filter_avg}
  \end{subfigure}
  \begin{subfigure}{0.23\textwidth}
    \includegraphics[width=\textwidth]{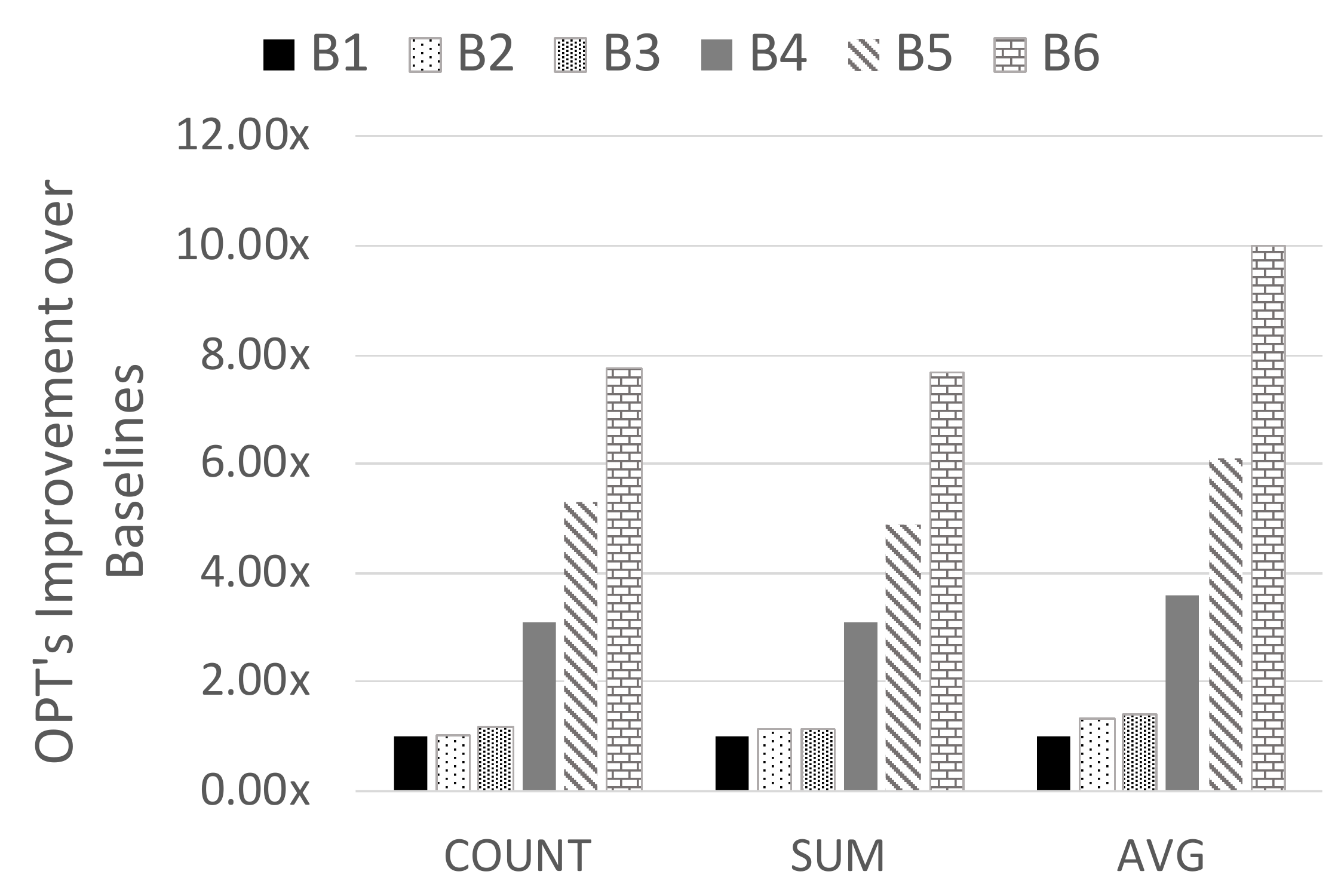}
    \inv\inv\inv
    \caption{\instacart}
    \label{fig:filter_instacart}
  \end{subfigure}
   \inv\sinv
  \caption{\textbf{\OPT's improvement in terms of the estimator's variance over six baselines in the presence of filters.}}
  \label{fig:filter_var}
  \inv
\end{figure*}

\begin{table}[t]
\caption{\textbf{Optimal sampling parameters  for \syn{uniform}{uniform} for different distributions of the filtered column $C$.}}
\inv
\centering
\begin{tabular}{c|p{0.6cm}p{0.7cm}|p{0.6cm}p{0.7cm}|p{0.6cm}p{0.7cm}}
\hline
Dist. of $C$ &  \multicolumn{2}{|c|}{\COUNT} & \multicolumn{2}{|c|}{\SUM} & \multicolumn{2}{|c}{\AVG} \\
\hline
{}   & $p$   & $q$    & $p$   & $q$ & $p$ & $q$ \\
Uniform &  0.010 & 1.000   & 0.010  & 1.000 & 0.010 & 1.000 \\
Normal   &  0.018 & 0.555   & 0.015  & 0.648 & 0.010 & 1.000 \\
Power law   &  0.051 & 0.195   & 0.050  & 0.201 & 0.010 & 1.000 \\
\hline
\end{tabular}
\inv
\label{tab:sampling_rate_filter}
\end{table}

To study \OPT's effectiveness in the presence of filters, we 
	 used \syn{uniform}{uniform} and \instacart datasets with $\epsilon$$=$$0.01$.
We added an extra column $C$ to $T_1$ in \syn{uniform}{uniform}, with  
 integers  in $[1, 100]$, and tried three distributions (uniform, normal, power law).
For \instacart,
we used the \textit{order\_hour\_of\_day} column for filtering, which had an almost normal distribution. 
We used an equality operator and  chose the comparison value $x$ uniformly at random. We calculated the average variance over all possible values of $c$.

Table~\ref{tab:sampling_rate_filter} shows the  sampling rates chosen by \OPT, while
Figure~\ref{fig:filter_var} shows \OPT's improvement over baselines in terms of  average variance. 
Again, \OPT  successfully achieved the lowest average variance among all baselines in all cases, up to 10x improvement compared to the worst baseline. 
This experiment confirms that \UBS with \OPT is highly effective for join approximation, even in the presence of filters.

\subsection{Combining Samples}
\label{sec:expr:combined}
% \dy{combined? one sample? plz let me know which one sounds better. will change once we are set on terminology.}

\begin{table}[t]
\caption{\textbf{Sampling parameters ($p$ and $q$) of \OPT using individual samples for different aggregates versus a combined sample (\syn{normal}{normal} dataset).}}
\inv
\centering
\begin{tabular}{c|p{0.45cm}p{0.65cm}|p{0.45cm}p{0.65cm}|p{0.45cm}p{0.65cm}}
\hline
Scheme &  \multicolumn{2}{|c|}{\COUNT} & \multicolumn{2}{|c|}{\SUM} & \multicolumn{2}{|c}{\AVG} \\
\hline
{}   & $p$   & $q$    & $p$   & $q$ & $p$ & $q$ \\
\OPT (individual)  &  0.145 & 0.069   & 0.125  & 0.080 & 0.010 & 1.000 \\
\OPT (combined)   &  0.133 & 0.075   & 0.133  & 0.075 & 0.133 & 0.075 \\
\hline
\end{tabular}
\label{tab:sampling_rate_combined}
\inv\inv
\end{table}

\begin{figure}[ht]
	\inv\inv
	\centering
	\includegraphics[width=0.35\textwidth]{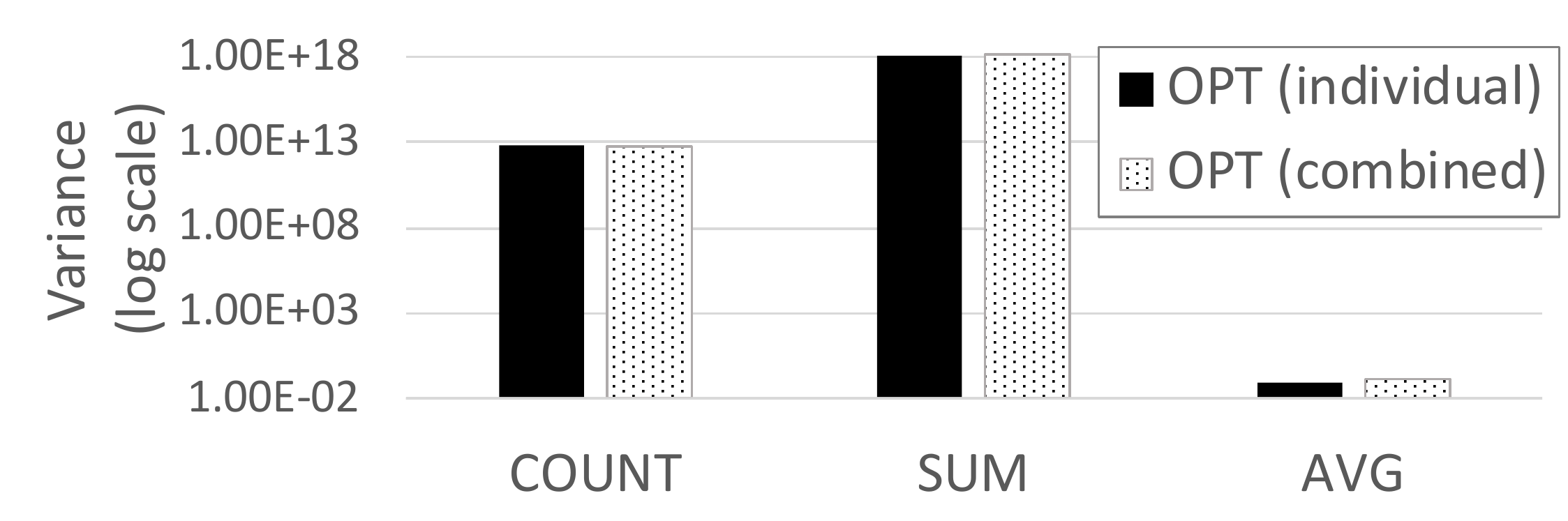}
	\inv
	\caption{\textbf{Variance of the query estimators for \OPT (individual) and \OPT (combined) for the \syn{normal}{normal} dataset.}}
	\inv
	\label{fig:variance_combined}
    \inv
\end{figure}

We evaluated the idea of using a single sample for multiple queries instead of generating individual samples for each query, as discussed in Section~\ref{sec:one-sample}.
Here, we use \OPT (individual) and \OPT (combined)  
	to denote the use of one-sample-per-query and one-sample-for-multiple-queries,
	respectively.
%Due to space constraint, we report using \syn{normal}{normal} dataset.
 For \OPT (combined), 
	we considered a scenario where each of \COUNT, \SUM, and \AVG is equally likely to appear. 
Table~\ref{tab:sampling_rate_combined} reports the sampling rates chosen in each case.
\ignore{Here, \OPT (combined) chose $p$ and $q$ that lie between the optimal sampling parameters for \COUNT and \SUM in \OPT (individual).
This is because the variance of an estimator increases much less by deviating from \AVG than the other two aggregates.}
 As shown in Figure~\ref{fig:variance_combined}, 
 %\OPT (combined). % as well as all baselines.
%\barzan{Dong Young, report both worst and best baselines!} 
 without having to generate an individual sample for each query, 
the variances of \OPT (combined)  were only slightly higher than those of \OPT (individual).
% \barzan{remove baselines from Fig~\ref{fig:variance_combined} and only leave indiv and combined OPT!}
This experiment shows that it is possible to create a single sample for multiple queries    without sacrificing too much optimality.

\subsection{Stratified Sampling}
\label{sec:expr:stratified}

\begin{figure*}[!ht]
\inv\inv
  \centering
  \begin{subfigure}{0.3\textwidth}
    \includegraphics[width=\textwidth]{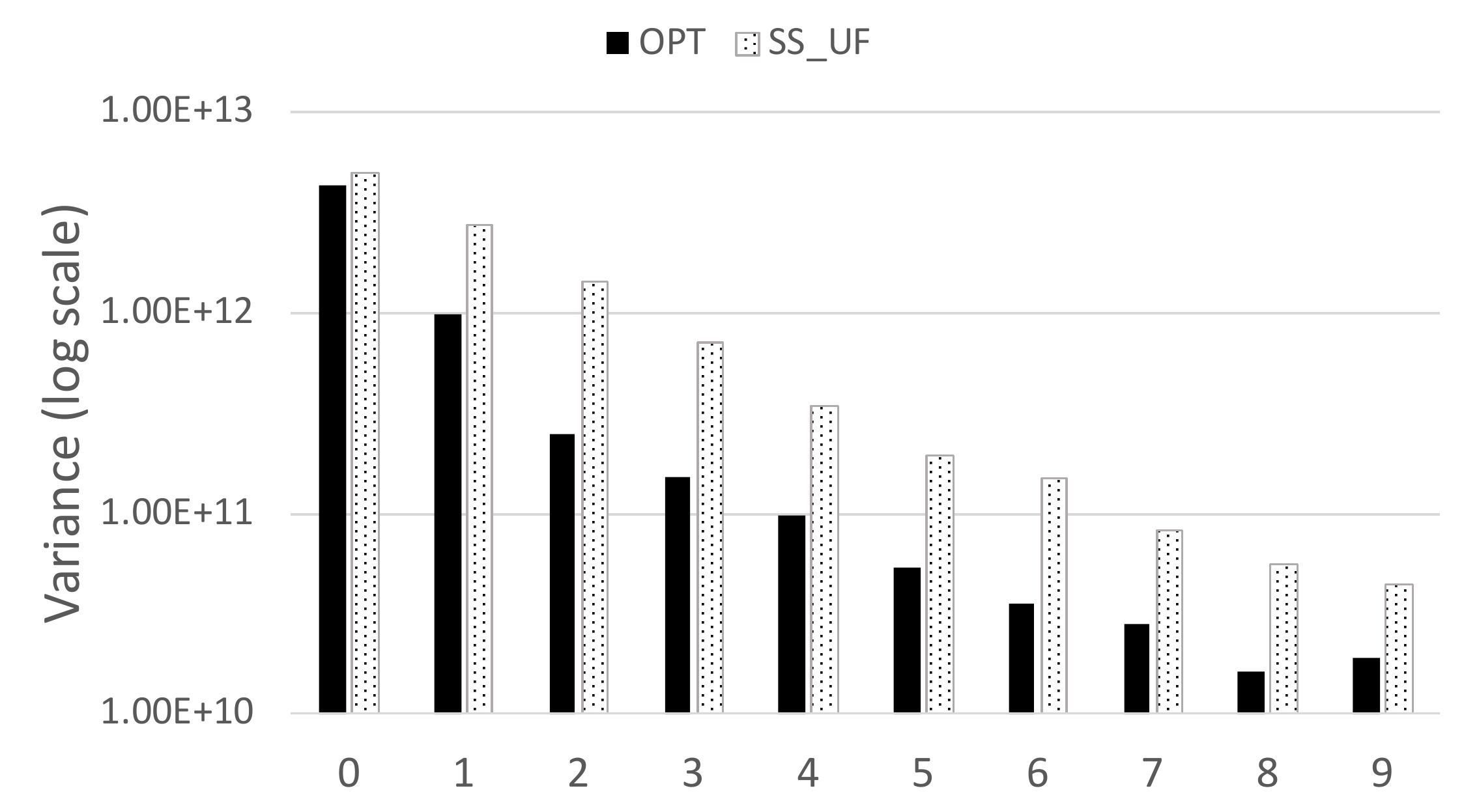}
    \inv\inv\inv
    \caption{\COUNT}
    \label{fig:stratified_count}
  \end{subfigure}
  \begin{subfigure}{0.3\textwidth}
    \includegraphics[width=\textwidth]{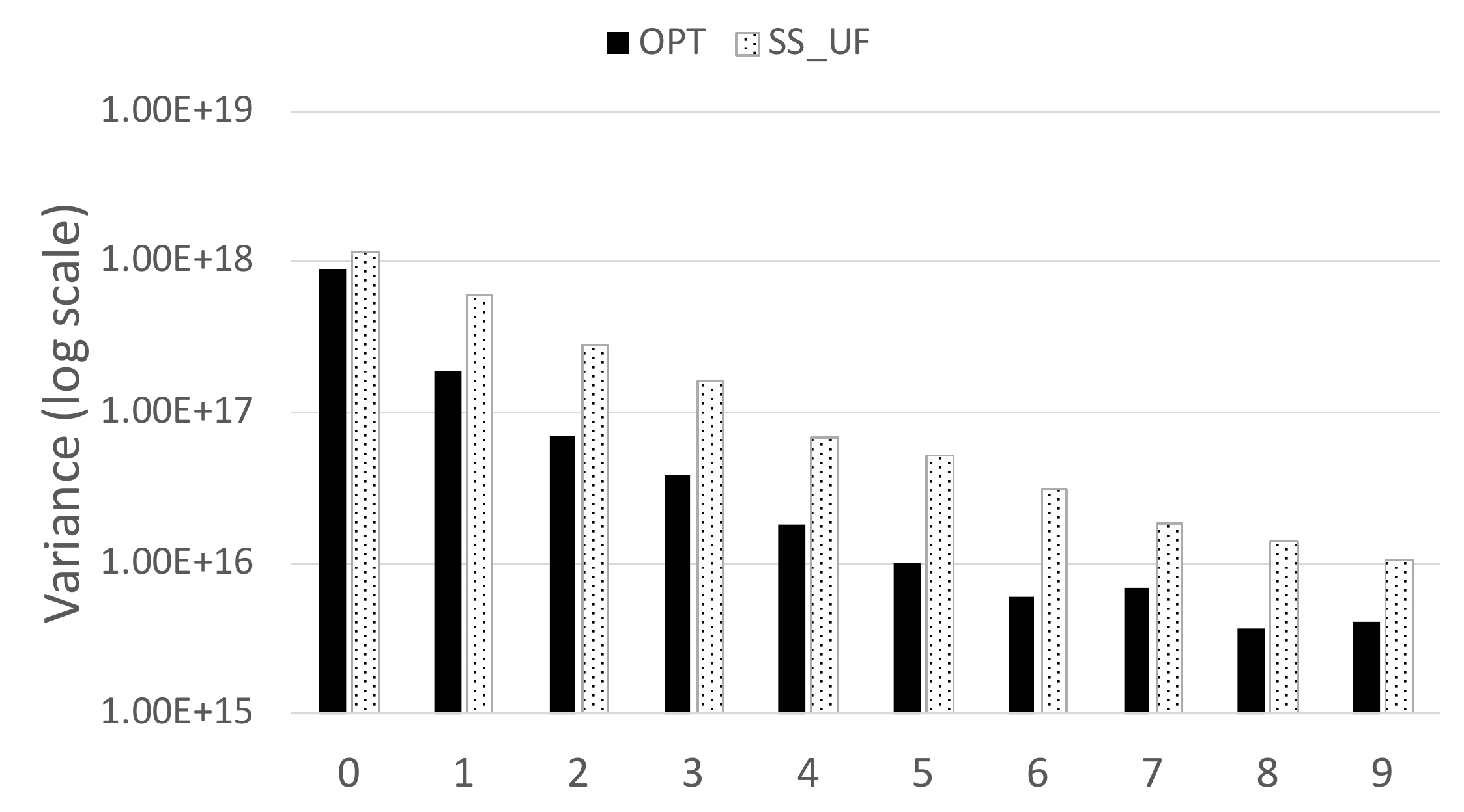}
    \inv\inv\inv
    \caption{\SUM}
    \label{fig:stratified_sum}
  \end{subfigure}
  \begin{subfigure}{0.3\textwidth}
    \includegraphics[width=\textwidth]{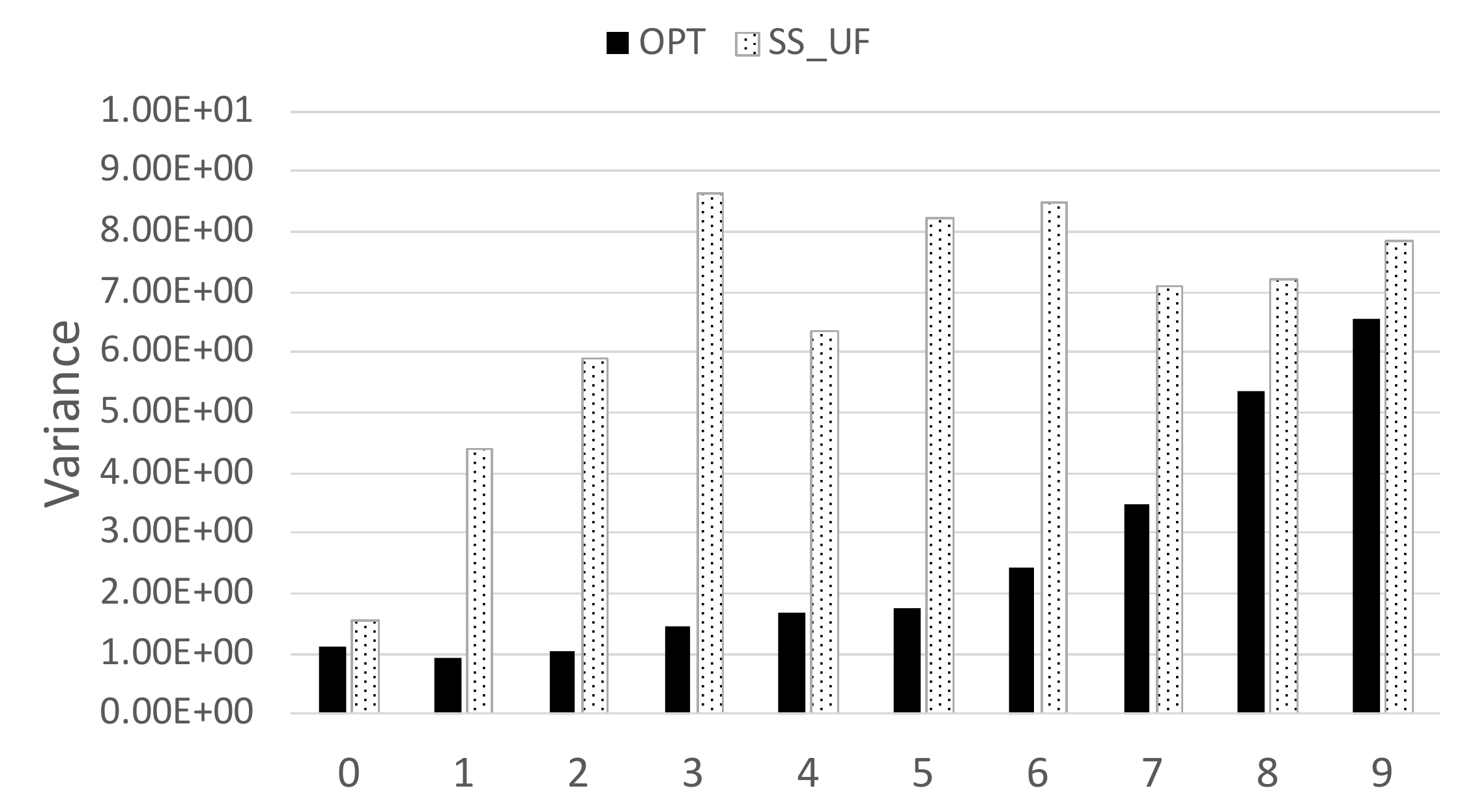}
    \inv\inv\inv
    \caption{\AVG}
    \label{fig:stratified_avg}
  \end{subfigure}
   \inv\sinv
  \caption{\textbf{Query estimator variance per group for 
  for a group-by join aggregate using
  \SUBS versus \SSUF.}}
  \label{fig:strafitied_var}
\end{figure*}

We also evaluated \SUBS for join queries with group-by.
Here, we used the \syn{normal}{normal} dataset, and added an extra group column $G$ to $T_1$ with integers from 0 to 9 drawn from a power law distribution with $\alpha = 1.5$. 
This time we did not randomize the groups, i.e., $G$=$0$ had the most tuples and 
$G$=$9$ had the fewest. 
This was to study SUBS performance with respect to the different group sizes.
As a baseline, we generated stratified samples for $T_1$ on $G$ \rev{with $\kkey = 100,000$}
% \barzan{u have misunderstood stratified sampling: it has a min number of tuples per group not max?!}
 and uniform samples for $T_2$ with a $0.01$ sampling budget.
We denote this baseline as \SSUF. 
% we sampled at most 100,000 tuples per  group from $T_1$ and performed uniform sampling on $T_2$ with 1\% budget.
  For \SUBS, we used  parameters that matched the sample size of \SSUF, \ie, $\kkey = 100, \ktup = 100,000$.
Figure~\ref{fig:strafitied_var} shows the variance of query estimators for each of the 10 groups for different aggregations.
As expected,
\SUBS with \OPT achieved lower variances than \SSUF
across all aggregates and groups with different sizes.
%This experiment illustrates that \SUBS can be very effective for join queries with group-bys, achieving lower variance with its query estimators.

\subsection{Overhead: Centralized vs.~Decentralized}
\label{sec:expr:sample_creation}

\begin{figure}[ht]
		\inv
	\centering
	\begin{subfigure}{0.225\textwidth}
    \includegraphics[width=\textwidth]{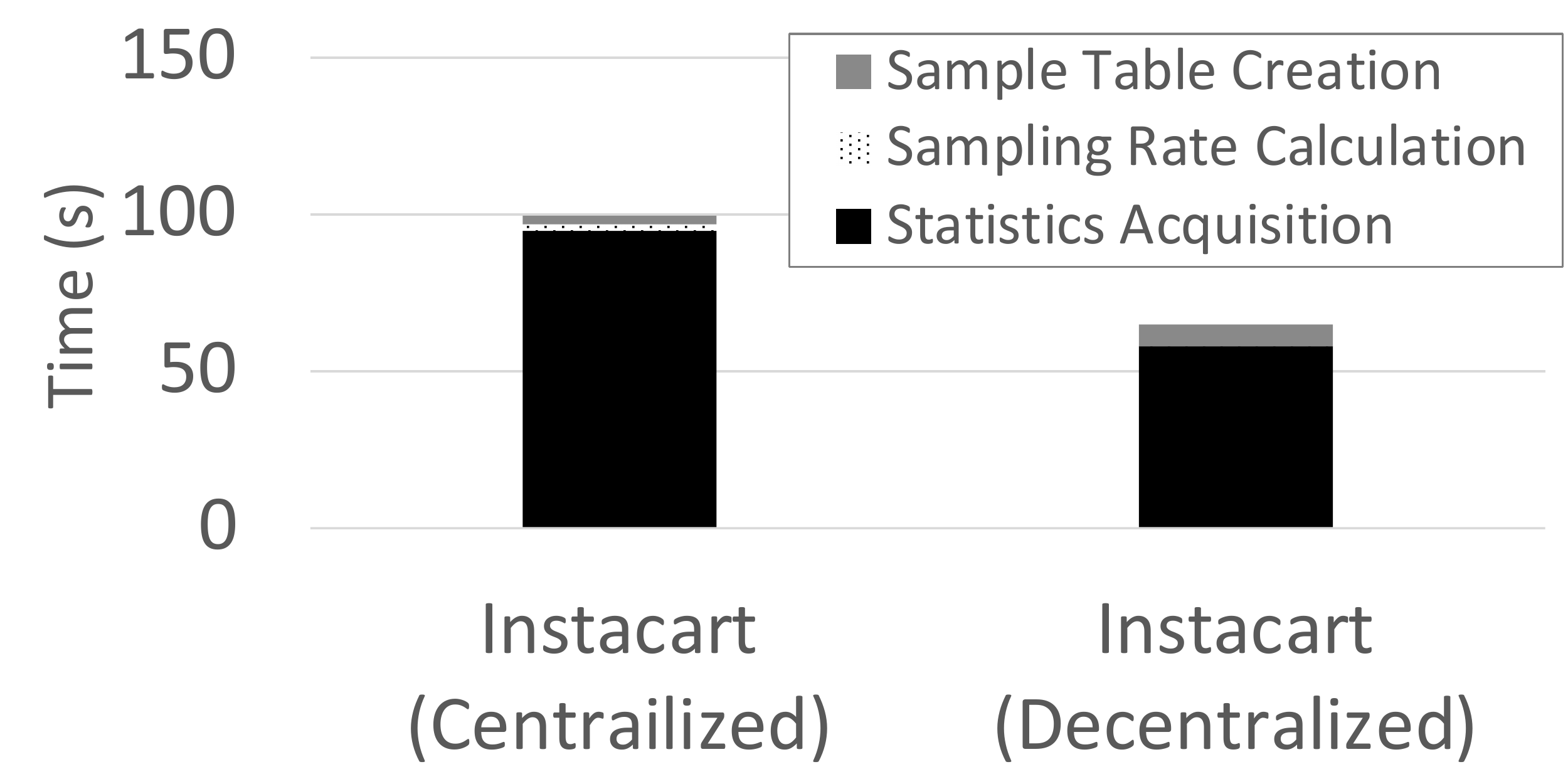}
    \inv\inv
    \caption{\instacart}
    \label{fig:instacart_gen_time}
  \end{subfigure}
  \begin{subfigure}{0.225\textwidth}
    \includegraphics[width=\textwidth]{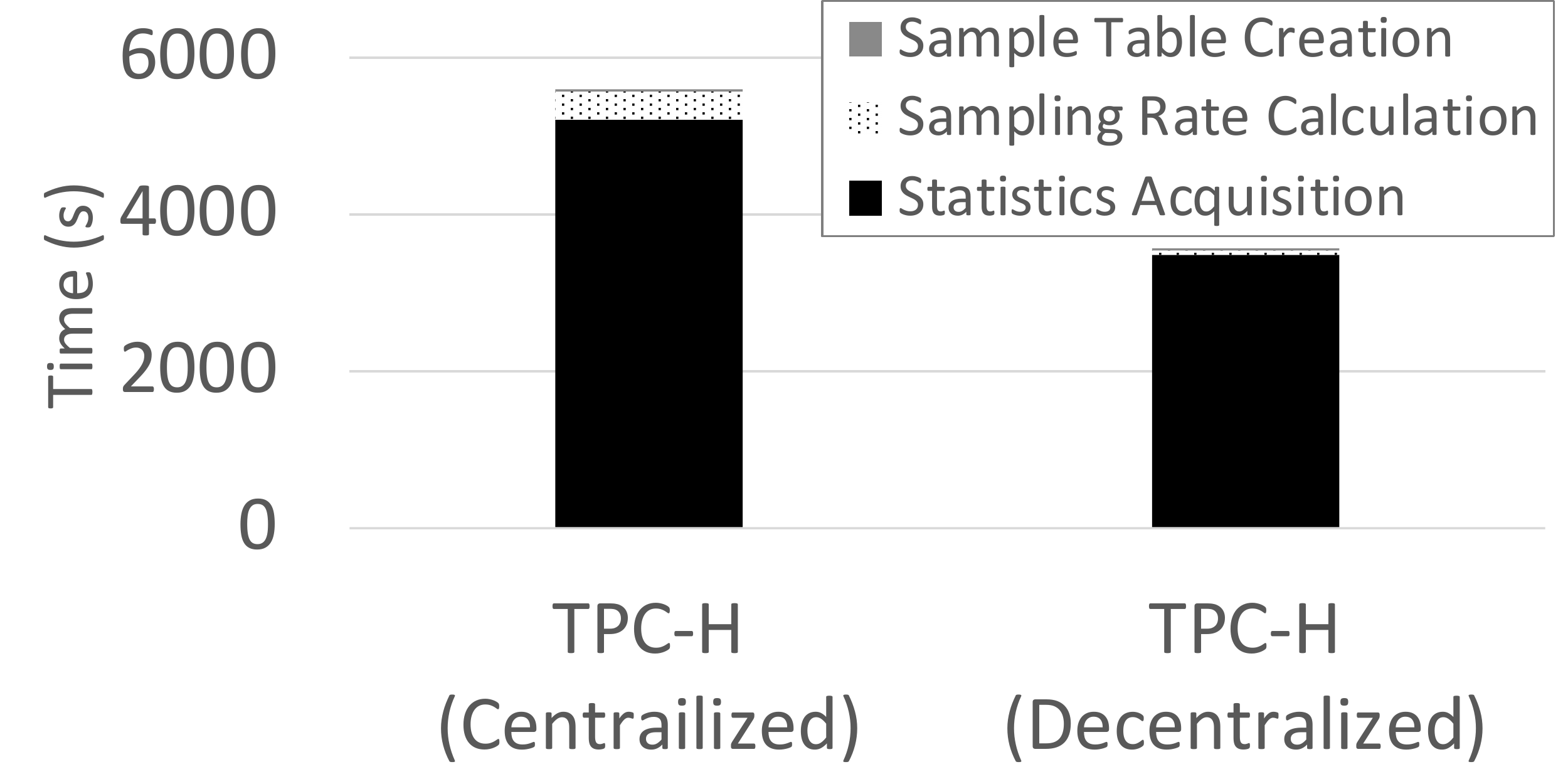}
    \inv\inv
    \caption{\tpch}
    \label{fig:tpch_gen_time}
  \end{subfigure}
	\inv
	\caption{\textbf{Time taken to generate samples for \instacart and \tpch in centralized vs. decentralized setting.}}
	\label{fig:sample_gen_time}
\end{figure}

\begin{figure}[t]
		\inv\inv
	\centering
	\begin{subfigure}{0.225\textwidth}
    \includegraphics[width=\textwidth]{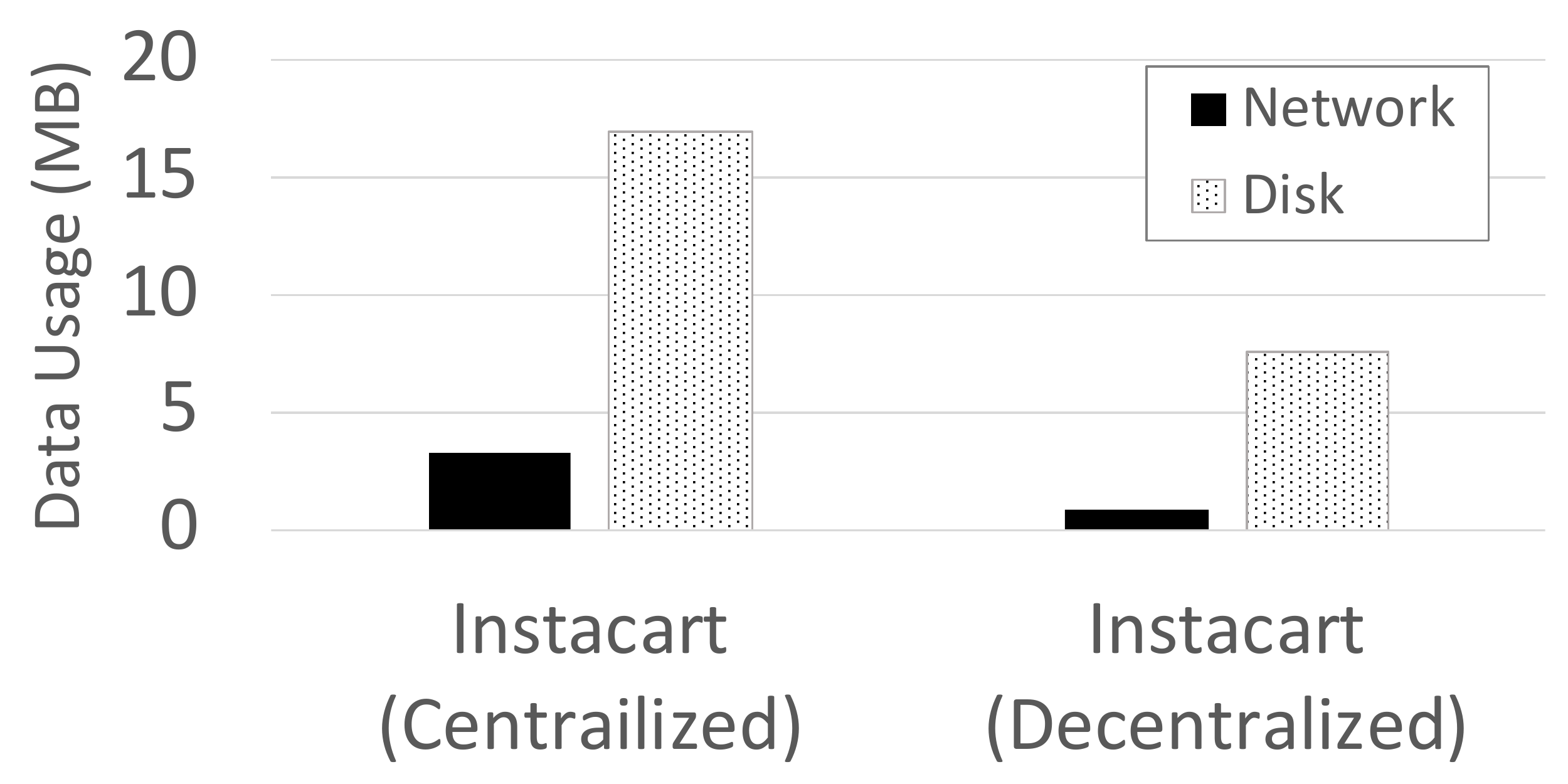}
    \inv\inv
    \caption{\instacart}
    \label{fig:instacart_gen_resource}
  \end{subfigure}
  \begin{subfigure}{0.225\textwidth}
    \includegraphics[width=\textwidth]{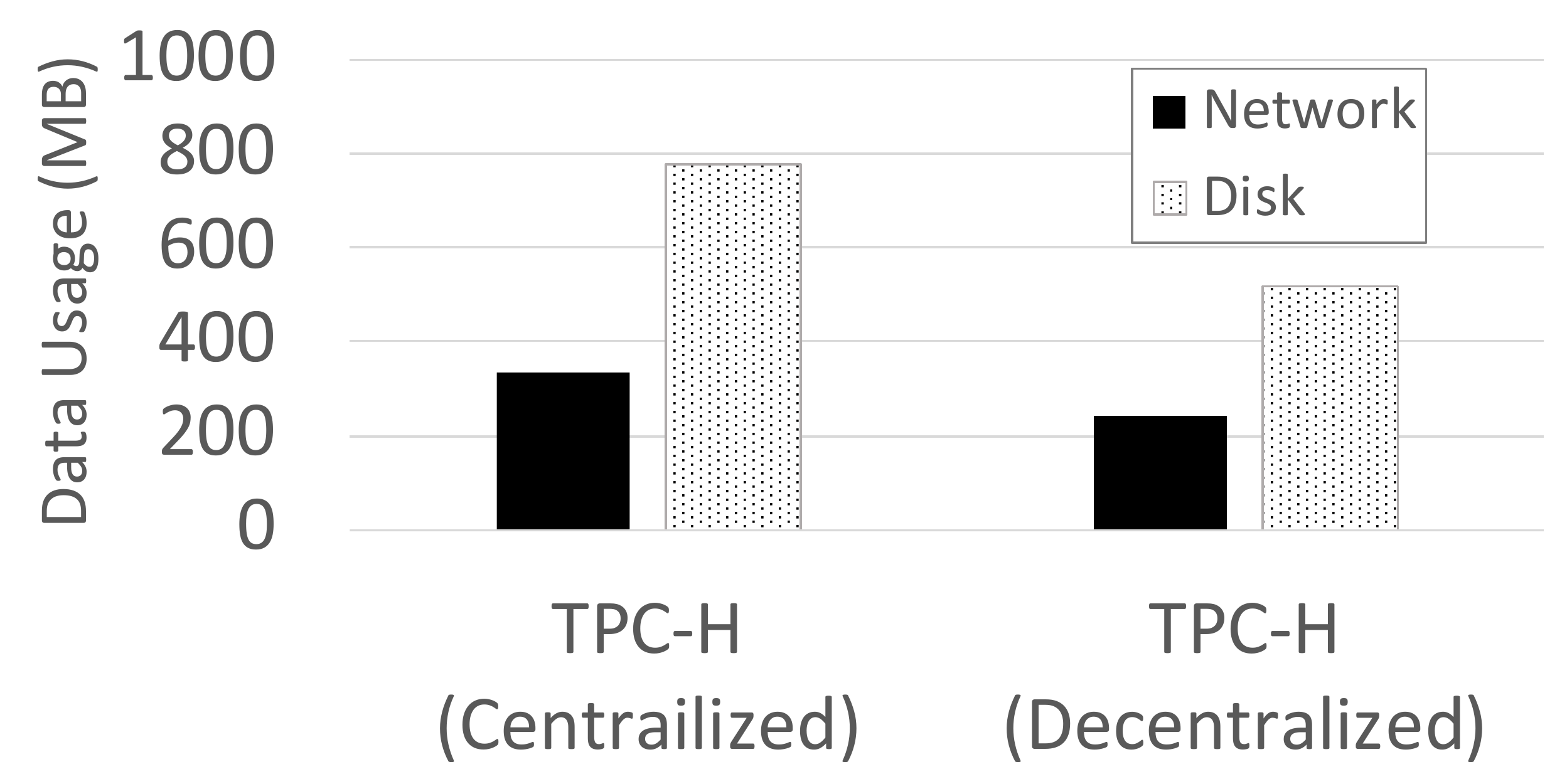}
    \inv\inv
    \caption{\tpch}
    \label{fig:tpch_gen_resource}
  \end{subfigure}
	\inv
	\caption{\textbf{Total network and disk bandwidth used to generate samples for \instacart and \tpch.}}
	\label{fig:sample_gen_resource}
	\inv\inv
\end{figure}

We compared the
    overhead of \OPT in centralized versus decentralized 
    settings, in terms of the sample creation time and resources, such as network and disk. 
\OPT should have  a much higher overhead 
in the centralized setting, 
as  it requires full frequency information of  every join key value in both tables. 
To quantify their overhead difference, 
%$T_2$, while it only needs to know the table size of $T_2$ in the decentralized setting.
we used \instacart and \tpch, and created a pair of samples for \SUM in each case. Here,
    the aggregation type did not matter, as the time spent calculating $p$ and $q$ was negligible compared to the time taken by transmitting the frequency vectors.
%ta inja

As shown in Figure~\ref{fig:sample_gen_time}, we measured the time for statistics acquisition, sampling rate calculation, and sample table creation.
Here, the time taken by collecting the frequencies was the dominant factor.
For \instacart, it took 65.16 secs from start to finish in the decentralized setting, compared to  99.98 secs in the centralized setting, showing 1.53x improvement in time.
% \barzan{not english, not clear what u re comparing to what? cent vs decent? or stats acruision vs generation? u are also contradicting urself if only 65 out of 99 is spent on X as 99-65 is no longer negligible}
For \tpch, it took  59.5 min in the decentralized setting, compared to 
91.7 mins in the centralized, 
showing a speedup of 1.54x.

We also measured the total network and disk I/O usage across the entire cluster, as shown in 
Figure~\ref{fig:sample_gen_resource}.
For \instacart,
compared to the decentralized setting, the centralized one used 3.66x ($0.9 \rightarrow 3.29$ MB) more network and 2.22x ($7.59 \rightarrow 16.9$ MB) more disk bandwidth.
Overall, the overhead was less for \tpch.
The centralized in this case  used 1.38x ($243.39 \rightarrow 337.04$ MB) more network and 1.49x ($519.03 \rightarrow 776.58$ MB) more disk bandwidth  than the decentralized setting.

This experiment shows the graceful tradeoff between the optimality of sampling and its overhead, making the decentralized variant an attractive choice for large datasets and distributed systems.

\iftechreport

\subsection{UBS vs. Two-Level Sampling}
\label{sec:expr:2lv}

Two-level sampling (\TLS)~\cite{two-level-sampling}     is similar to our UBS scheme in that it also applies stratified sampling before Bernoulli sampling. However, unlike UBS which 
    applies the same sampling rate to all tuples,
      \TLS uses a different universe
      sampling rate for each join key, i.e.,
    \TLS is strictly more expressive than UBS.
Thus, by using significantly more parameters (i.e., number of distinct join keys), 
    \TLS should be able to achieve 
    a lower 
        variance than any UBS scheme (which uses only two parameters $p$ and $q$).
To empirically measure this gap,
    we compared the relative error of \TLS
        versus \OPT.
Since \TLS is originally designed for cardinality estimation, 
we compared their \COUNT error on our \synthetic datasets (since \COUNT is symmetric, we only have 6 combinations), 
with 100K tuples, 100 keys, and $\epsilon$$=$$ 1\%$ budget.
Here, we used 
As shown in Figure~\ref{fig:opt-2lv}, 
    \TLS's error was slightly lower than \OPT,
        and both errors increased with the skew  in the join keys, e.g., when a power law was involved. 
This was as expected: having a separate parameter for each join key means more complexity, but also allows \TLS to better adopt to the distribution of the data.

\begin{figure}[t]
    \inv
	\centering
    \includegraphics[width=0.45\textwidth]{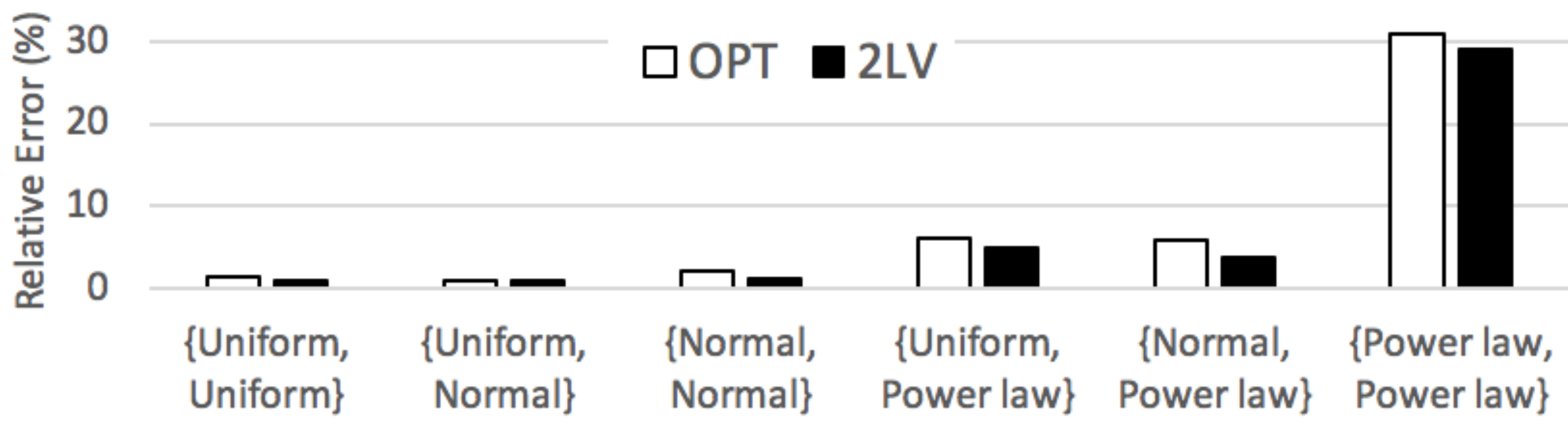}
    \inv
\caption{\textbf{Optimal UBS vs. two-level sampling~\cite{two-level-sampling}.}}
	\label{fig:opt-2lv}
	\inv\inv
\end{figure}

\fi

%!TEX root = main.tex

\section{Related Work}
\label{sec:related}

\ph{Online Sample-based Join Approximation}
 Ripple Join~\cite{ripple_join, ripple_join_scalable} is an online join algorithm that operates under the assumption that the tuples of the original tables are processed in a random order.
Each time, it retrieves a random tuple (or a set of random tuples) from the tables,
and then joins the new tuples with the previously read tuples and with each other.
% Luo et al.~\cite{ripple_join_scalable} propose a parallel version of Ripple Join.
% Jermaine et al.~\cite{disk-based-join} show that the hashed version of Ripple Join
% 	can be quite slow when hash tables exceed memory. 
%which was found to have the fastest convergence to an exact answer in general by Haas and Hellerstein~\cite{ripple_join},
\rev{SMS~\cite{disk-based-join} 
speeds up the hashed version of Ripple Join
 when hash tables exceed memory.}
% SMS join~\cite{disk-based-join} overcomes this problem using a
%  disk-based sort-merge join.
 % and its variants~\cite{dbo, ripple_join_scalable} repeatedly compute samples from tables and perform the join on those samples. \barzan{u re not using proper terminology. are u using sample as a tuple or a set of tuples? in the rest of the paper 
% we are using sample only to refer to the act of sampling or the set of tuples. plz be consistent in this section too}
% \barzan{how are the variants diff than the original?!}
% These algorithms \barzan{all of them?} make an assumption that the tuples in each table are stored in a random order.
 Wander Join~\cite{wander-join} tackles the problem of $k$-way chain join 
and eliminates the random order requirement of Ripple Join. However, it requires an index on every join  column in each of the tables.
%Wander Join models the join as a $k$-partite graph, where vertices are tuples from the original tables and each tuple in the join corresponds to a path of length $k - 1$.
 Using indexes, Wander Join performs a set of random walks %on this graph 
 and obtains a \emph{non-uniform} but independent sample of the join.
Maintaining an approximation of the size of all partial joins can help overcome the
 non-uniformity problem~\cite{wander-join-2018, wander-join-xdb}.
 
\ignore{they address the problem of non-uniformity by maintaining an approximation of the size of all partial joins and using this information to bias their random walks, thereby yielding an independent sample that is close to uniform. 
} 

% The offline approach in join approximation computes the sample of the
% Another folklore approach to sampling in join is to perform uniform sampling in each relation and compute the join of the samples as a sample of join.u
% This is summarized and studied in~\cite{random-sampling-joins}.
% another folklore method is \emph{partial sampling},
% it samples from one of the table and join the sample with the other full table in query time.
% For foreign key joins, this can be generalized to \emph{join synopse}~\cite{join_synopses} where we think of the join query as a rooted tree where we performs sample from the root table and join the sample to the other tables. 
% When frequencies information of join keys are available, wander join~\cite{wander-join} takes samples directly from the join by performing a biased random walk in a graph consists of all the tuple in the join.
% An index on the join attribute is needed to perform wander join.
% Similar works also includes adaptive sampling and bifocal sampling \dy{citations?}.
% However, these previous methods ignore the correlation between the two tables and is inaccurate when dealing with some join such as PK-FK joins.
\ph{Offline Sample-based Join Approximation} 
AQUA~\cite{aqua1} acknowledges the quadratic reduction and the non-uniformity of the output when joining two uniform random samples. 
The same authors propose \emph{Join Synopsis}~\cite{join_synopses}, 
	which computes a sample of one of the tables and joins it with the other tables as a sample of the actual join. 
% 		\barzan{not english, rephrase}
% 	\tofix{which interprets a join query as a tree where it generates a sample from the root table and joins the sample with the other tables.}
% 	\barzan{1) the sentence is not english 2) the word `sample' is used inappropriately to refer to a tuple, see prev comments}
Chaudhuri et al.~\cite{random-sampling-joins} also point out that a join of independent samples from two relations does not yield an
    independent sample of their join, and
  propose using precomputed statistics to overcome this problem.
\rev{However, their solution 
can be quite costly, as it
requires collecting full frequency information of the relation.}
% Specifically, 
% in order to calculate the tightest upper bound on the join size,
% Chaudhuri et al. collect 	full frequency information of
% 	the relation, which can be quite expensive.
Zhao et al.~\cite{wander-join-2018} \rev{provide}
% overcome this limitation by providing
	a better trade-off between sampling efficiency and the join size upper bound.
% Hadjieleftheriou et al.~\cite{hashed-samples} 
Hashed sampling \rev{(a.k.a. universe)} \cite{hashed-samples} is proposed
			in the context of selectivity estimation for set similarity queries.
% \barzan{`)did you mean `extend BLAH to multi-way joins'? or do they do more than that? 
% 2) also, mention the downsides and impracticality of collecting stats}
% Others~\cite{VengerovMZC15, kandula2016quickr} acknowledge another issue with uniform sampling: it ignores the correlation between the two tables, which makes it inaccurate when dealing with certain types of joins, such as PK-FK joins.
% \barzan{1) what does `inaccurate' mean? bias? or high variance or both? 2) what other forms of joins have this issue?
% 3) how do we handle such correlations? is this a limitation of our work? if so, why is it not in the limitation section??}
% They \barzan{who is they?? these 2 papers are 1 year apart how can they BOTH propose it? and are u sure no one before 2015 proposed it? I doubt it} propose \emph{universe sampling} (see Section~\ref{sec:WhereWeIntroducedIt}).
% \barzan{say some of their limitations, e.g., they didn't combine etc. Also, I am quite dubious as QuickR did discuss placing random sampling 
% 	operators in different parts of the query plan which is equivalent to combining different forms of sampling! plz go back to that paper and read it carefully}
% 	\barzan{are u sure VengerovMZC15 offline samplingbi-level-sampling?? i know that QuickR is online sampling but i dont know why you have put it here} 
 Block-level uniform sampling~\cite{block-level-sampling} is less accurate but  more efficient than tuple-level sampling. 
Bi-level sampling~\cite{bi-level-sampling-ola, bi-level-sampling} performs
Bernoulli sampling at both the block- and tuple-level,
	as a trade-off between accuracy and I/O cost of sample generation.
Kamat and Nandi~\cite{Kamat017} use simple stratified sampling on join column but with an objective function measuring the amount of randomness of the sample scheme, which shows improvement over simple correlated sampling.

\ph{AQP Systems on Join}
Most AQP systems rely on sampling and support certain types of joins~\cite{kandula2016quickr, mozafari_sigmod2018_verdict, aqua1, surajit-optimized-stratified, mozafari_eurosys2013, ganti2000icicles, pf-ola, wander-join-xdb}.
% \barzan{this lisit is very incomplete. 1) there are numerous others 2) plz cite some of the systems from Florin as well as the postgres-based AQP system that was proposed by the authors of Wander Join too}
% STRAT~\cite{surajit-optimized-stratified} defines the problem of generating an optimal stratified sample as an optimization problem.
STRAT~\cite{surajit-optimized-stratified} discusses the use of uniform and stratified sampling, and how those can support certain types of join queries. 
More specifically, STRAT only supports PK-FK joins between a fact table and one or more dimension table(s).
BlinkDB~\cite{mozafari_eurosys2013} extends STRAT and considers multiple stratified samples instead of a single one. % but it only supports the same set of join queries as STRAT.
 As previously mentioned, AQUA~\cite{aqua1} supports foreign key joins using join synopses.
% Ganti et al.~\cite{ganti2000icicles} discuss a new class of samples called \textit{icicles}, 
\textit{Icicles}~\cite{ganti2000icicles}   samples   tuples that are more likely to be required by future queries, but, similar to AQUA,  only supports foreign key joins.
PF-OLA~\cite{pf-ola} is a framework for parallel online aggregation. 
It studies  parallel joins with group-bys, when partitions of the two tables fit in memory.
XDB~\cite{wander-join-xdb} integrates Wander Join in PostgreSQL. 
  Quickr~\cite{kandula2016quickr} does not create offline samples.
Instead, it uses universe sampling to support equi-joins, where the group-by columns and the value of aggregates are not correlated with the join keys.
VerdictDB~\cite{mozafari_sigmod2018_verdict} is a universal AQP framework that supports all three types of samples (uniform, universe, and stratified).
%  To solve the problem of inter-tuple dependencies 
% in the output of the join of universe samples, 
\tempcut{VerdictDB utilizes a technique called \textit{variational subsampling},
	which creates subsamples of the sample such that it only requires a single 
	join---instead of repeatedly joining the subsamples multiple times---to produce 				accurate aggregate approximations.}
 % It uses a technique, called \textit{variational subsampling}, on join queries to efficiently provide accurate estimates \tofix{even with small samples.}
% \barzan{is this true? what about variational subsampling improves joins? isn't it just uniform sampling?}
% \tofix{However, both Quickr and VerdictDB only introduce the notion of universe sampling for joins.
% They lack the discussion of how to create \emph{best} samples for joins, which is the problem that our proposed BVS scheme and formulation try to answer in this paper.} \barzan{not english, reword}
ApproxJoin~\cite{QuocABBCFS18} uses Bloom Filters in conjunction with stratified sampling to efficiently produce a sample to the join when relations are distributed across different nodes. 

\ph{Join Cardinality Estimation}
There is extensive work on join cardinality estimation (\ie,\texttt{count(*)}) 
	in the database  community~\cite{VengerovMZC15, pitoura2008self, sampling-re-optimization, EstanN06, alon2002tracking, swami1994estimation, done-right, learned-cardinalities} as an important step 
		of the query optimization process for joins.
 Two-level sampling~\cite{two-level-sampling} first applies universe sampling to the join values, and then, 
for each join value sampled, it performs Bernoulli sampling. However, unlike 
our UBS scheme which applies the same rate to all keys, two-level sampling 
    uses a different rate during its universe 
    sampling for each  join key.
In other words, 
    two-level sampling is a more complex 
        scheme with significantly more parameters than UBS (which requires only two parameters, $p$ and $q$), and is thus less amenable to efficient and decentralized implementation.
Furthermore, 
two-level sampling applies two different sampling methods, whereas bi-level sampling~\cite{bi-level-sampling} 
  uses only 
		Bernoulli sampling but at different granularity levels.
End-biased sampling~\cite{EstanN06} samples each tuple with a probability proportional to 	the frequency of its join key.
Index-based sampling~\cite{done-right} 
 and deep learning~\cite{learned-cardinalities} 
 have also been utilized 
to improve  cardinality estimates.

% \barzan{Dong Young, go over each of the papers disk-based-joincited in the marked area and remove those that don't have ANY PC member in their authors to save space}
% Pitoura and Triantafillou~\cite{pitoura2008self} propose a self-join size estimation algorithm for large distributed databases.
% Vengerov et al.~\cite{VengerovMZC15} show that universe sampling (named  \textit{correlated sampling}) estimates  join cardinalities better than both
%  end-biased and Bernoulli sampling, especially in the presence of filters.
% Wu et al.~\cite{sampling-re-optimization} propose an iterative algorithm that repeatedly refines its
% 	 join cardinality estimate using additional samples, in order to optimize an existing 			query plan. 

\ph{Theoretical Studies}
The question about the limitation of sample-based approximation of joins, to the best of our knowledge, has not been asked in the theory community.
 However, the past work in communication complexity on set intersection and inner product estimation has implications for join approximation. 
 In this problem, the Alice and Bob possess respectively two %\tofix{$d$-dimensional} 
 vectors $x$ and $y$ and they wish to compute their inner product $t = \langle x, y \rangle$ without exchanging the vector $x$ and $y$. 
 \ignore{Consider two $d$-dimensional vectors $x$ and $y$, 
    each   
	owned by a different party, say Alice and Bob.
 In our context, we can think of $x$ and $y$ as frequency vectors of two tables and $\langle x, y \rangle$ as the size of their join.
Alice computes a summary of her vector, $\beta(x)$,
  Alice and Bob each compute a summary  
of their vector, $\beta(x)$ and $\beta(y)$,
 and a third party estimates  the inner product $\langle x, y \rangle$ using the information in $\beta(x)$ and $\beta(y)$.}
In the one-way model, 
Alice computes a summary $\beta(x)$ and sends it to Bob, who will estimate
 $\langle x, y \rangle$ using $y$ and $\beta(x)$.
For this problem, 
%even in the special case where  $x$ and $y$ are $0, 1$ vectors and Bob's vector $y$  contains exactly one $1$ among all  coordinates,  Alice still needs to send $\Omega(d)$ bits to Bob in order for Bob to estimate $\langle x, y\rangle$ with a variance less than $O(t^2)$ where $t = \langle x, y\rangle$.}
  %A careful analysis by
  \cite{min-wise-hashing} shows that any estimator produced by $s$ bits of communication has variance at least $\Omega(dt/s)$. 
%Moreover, this variance is achievable using minwise hashing~\cite{LiK11}.
Estimating inner product for $0,1$ vectors is directly related to estimating \SUM and \COUNT for a PK-FK join. A natural question is whether the join is still hard even if frequencies are all larger than $1$.  Further, the question of whether estimating \AVG is also hard is not answered by prior work.

\section{Conclusion}
\label{sec:conclusion}

Our goal in this paper was to
    improve our understanding of join approximation 
        using offline samples, and formally address  
    some of the key open questions faced 
    by practitioners using and building AQP engines. 
We defined generic sampling schemes that cover 
    the most common sampling strategies, as well as 
    as their combinations. 
Within these schemes, we (1) provided an information-theoretical lower bound on the lowest error achievable   by any offline sampling scheme,
 (2) derived optimal strategies that match this 
    lower bound within a constant factor,
    and (3) offered a decentralized variant 
    that requires minimal communication of statistics 
        across the network.
These results allow practitioners to 
        quickly determine---e.g., based on the distribution of the join columns---if joining offline samples
            will be futile or  will yield a reasonable accuracy. 
We also expect our hybrid samples to improve 
    the accuracy of
     database learning~\cite{mozafari_sigmod2017_dbl} and 
    selectivity estimation~\cite{quicksel_tr} for join queries.    
        
\section{Acknowledgement}

This research is in part supported by the National Science Foundation through grants 1553169 and 1629397.

\balance

\nocite{Kamat017,QuocABBCFS18}

\bibliographystyle{plain}
\bibliography{mozafari,approximate,ldb,dawei}

\begin{thebibliography}{10}

\bibitem{instacart}
The instacart online grocery shopping dataset 2017.
\newblock \url{https://www.instacart.com/datasets/grocery-shopping-2017}.
\newblock Accessed: 2019-07-20.

\bibitem{infobright_acquisition}
Security on-demand announces acquisition of {Infobright} analytics \&
  technology assets.
\newblock \url{https://tinyurl.com/y6ctn4vs}.

\bibitem{tpch}
{TPC-H Benchmark}.
\newblock http://www.tpc.org/tpch/.

\bibitem{cloudlab}
Cloudlab.
\newblock \url{https://www.cloudlab.us}, 2019.

\bibitem{aqua1}
S.~Acharya, P.~B. Gibbons, and V.~Poosala.
\newblock Aqua: A fast decision support system using approximate query answers.
\newblock In {\em VLDB}, 1999.

\bibitem{join_synopses}
S.~Acharya, P.~B. Gibbons, V.~Poosala, and S.~Ramaswamy.
\newblock Join synopses for approximate query answering.
\newblock In {\em SIGMOD}, 1999.

\bibitem{mozafari_sigmod2014_diagnosis}
Sameer Agarwal, Henry Milner, Ariel Kleiner, Ameet Talwalkar, Michael Jordan,
  Samuel Madden, Barzan Mozafari, and Ion Stoica.
\newblock Knowing when you're wrong: Building fast and reliable approximate
  query processing systems.
\newblock In {\em SIGMOD}, 2014.

\bibitem{mozafari_eurosys2013}
Sameer Agarwal, Barzan Mozafari, Aurojit Panda, Henry Milner, Samuel Madden,
  and Ion Stoica.
\newblock {BlinkDB}: queries with bounded errors and bounded response times on
  very large data.
\newblock In {\em EuroSys}, 2013.

\bibitem{mozafari_pvldb2012}
Sameer Agarwal, Aurojit Panda, Barzan Mozafari, Anand~P. Iyer, Samuel Madden,
  and Ion Stoica.
\newblock Blink and it's done: Interactive queries on very large data.
\newblock {\em PVLDB}, 2012.

\bibitem{alon2002tracking}
Noga Alon, Phillip~B Gibbons, Yossi Matias, and Mario Szegedy.
\newblock Tracking join and self-join sizes in limited storage.
\newblock {\em Journal of Computer and System Sciences}, 64, 2002.

\bibitem{AlonMS99}
Noga Alon, Yossi Matias, and Mario Szegedy.
\newblock The space complexity of approximating the frequency moments.
\newblock {\em J. Comput. Syst. Sci.}, 58, 1999.

\bibitem{size-bounds}
Albert Atserias, Martin Grohe, and D{\'{a}}niel Marx.
\newblock Size bounds and query plans for relational joins.
\newblock {\em {SIAM} J. Comput.}, 42(4), 2013.

\bibitem{dynamicp-sample-selection}
Brian Babcock, Surajit Chaudhuri, and Gautam Das.
\newblock Dynamic sample selection for approximate query processing.
\newblock In {\em VLDB}, 2003.

\bibitem{BV2014}
Stephen~P. Boyd and Lieven Vandenberghe.
\newblock {\em Convex Optimization}.
\newblock Cambridge University Press, 2014.

\bibitem{surajit-optimized-stratified}
Surajit Chaudhuri, Gautam Das, and Vivek Narasayya.
\newblock Optimized stratified sampling for approximate query processing.
\newblock {\em TODS}, 2007.

\bibitem{block-level-sampling}
Surajit Chaudhuri, Gautam Das, and Utkarsh Srivastava.
\newblock Effective use of block-level sampling in statistics estimation.
\newblock In {\em SIGMOD}, 2004.

\bibitem{random-sampling-joins}
Surajit Chaudhuri, Rajeev Motwani, and Vivek Narasayya.
\newblock On random sampling over joins.
\newblock In {\em SIGMOD}, 1999.

\bibitem{two-level-sampling}
Yu~Chen and Ke~Yi.
\newblock Two-level sampling for join size estimation.
\newblock In {\em SIGMOD}, 2017.

\bibitem{bi-level-sampling-ola}
Yu~Cheng, Weijie Zhao, and Florin Rusu.
\newblock Bi-level online aggregation on raw data.
\newblock In {\em SSDBM}, 2017.

\bibitem{online-agg-mr2}
Tyson Condie, Neil Conway, Peter Alvaro, Joseph~M. Hellerstein, Khaled
  Elmeleegy, and Russell Sears.
\newblock Mapreduce online.
\newblock In {\em NSDI}, 2010.

\bibitem{CLRS}
Thomas~H Cormen, Charles~E Leiserson, Ronald~L Rivest, and Clifford Stein.
\newblock {\em Introduction to algorithms}.
\newblock MIT press, 2009.

\bibitem{green-book}
Graham Cormode, Minos Garofalakis, Peter~J Haas, and Chris Jermaine.
\newblock Synopses for massive data: Samples, histograms, wavelets, sketches.
\newblock {\em Foundations and Trends in Databases}, 4, 2012.

\bibitem{viz1-brown}
Andrew Crotty, Alex Galakatos, Emanuel Zgraggen, Carsten Binnig, and Tim
  Kraska.
\newblock Vizdom: Interactive analytics through pen and touch.
\newblock {\em {PVLDB}}, 2015.

\bibitem{DobraGGR02}
Alin Dobra, Minos~N. Garofalakis, Johannes Gehrke, and Rajeev Rastogi.
\newblock Processing complex aggregate queries over data streams.
\newblock In {\em {SIGMOD}}, 2002.

\bibitem{dbo}
Alin Dobra, Chris Jermaine, Florin Rusu, and Fei Xu.
\newblock Turbo-charging estimate convergence in dbo.
\newblock {\em PVLDB}, 2009.

\bibitem{EstanN06}
Cristian Estan and Jeffrey~F. Naughton.
\newblock End-biased samples for join cardinality estimation.
\newblock In {\em {ICDE}}, 2006.

\bibitem{ganti2000icicles}
Venkatesh Ganti, Mong-Li Lee, and Raghu Ramakrishnan.
\newblock Icicles: Self-tuning samples for approximate query answering.
\newblock In {\em VLDB}, 2000.

\bibitem{walmart_talk}
Deepak Goyal.
\newblock Approximate query processing at {WalmartLabs}.
\newblock
  {\url{https://fifthelephant.talkfunnel.com/2018/43-approximate-query-processing}}.

\bibitem{ripple_join}
Peter~J. Haas and Joseph~M. Hellerstein.
\newblock {Ripple Joins for Online Aggregation}.
\newblock In {\em SIGMOD}, pages 287--298, 1999.

\bibitem{bi-level-sampling}
Peter~J Haas and Christian K{\"o}nig.
\newblock A bi-level bernoulli scheme for database sampling.
\newblock In {\em SIGMOD}, 2004.

\bibitem{hashed-samples}
Marios Hadjieleftheriou, Xiaohui Yu, Nick Koudas, and Divesh Srivastava.
\newblock Hashed samples: selectivity estimators for set similarity selection
  queries.
\newblock {\em PVLDB}, 2008.

\bibitem{harper2016movielens}
F~Maxwell Harper and Joseph~A Konstan.
\newblock The movielens datasets: History and context.
\newblock {\em TIIS}, 2016.

\bibitem{online-agg}
Joseph~M. Hellerstein, Peter~J. Haas, and Helen~J. Wang.
\newblock Online aggregation.
\newblock In {\em SIGMOD}, 1997.

\bibitem{horvitz1952generalization}
Daniel~G Horvitz and Donovan~J Thompson.
\newblock A generalization of sampling without replacement from a finite
  universe.
\newblock {\em Journal of the American statistical Association}, 47, 1952.

\bibitem{approx-join-techreport}
Dawei Huang, Dong~Young Yoon, Seth Pettie, and Barzan Mozafari.
\newblock Joins on samples: A theoretical guide for practitioners.
\newblock {\small \url{https://arxiv.org/abs/1912.03443}}, 2019.

\bibitem{disk-based-join}
Christopher Jermaine, Alin Dobra, Subramanian Arumugam, Shantanu Joshi, and
  Abhijit Pol.
\newblock A disk-based join with probabilistic guarantees.
\newblock In {\em SIGMOD}, 2005.

\bibitem{Kamat017}
Niranjan Kamat and Arnab Nandi.
\newblock A unified correlation-based approach to sampling over joins.
\newblock In {\em Proceedings of the 29th International Conference on
  Scientific and Statistical Database Management, Chicago, IL, USA, June 27-29,
  2017}, pages 20:1--20:12, 2017.

\bibitem{kandula2016quickr}
Srikanth Kandula, Anil Shanbhag, Aleksandar Vitorovic, Matthaios Olma, Robert
  Grandl, Surajit Chaudhuri, and Bolin Ding.
\newblock Quickr: Lazily approximating complex adhoc queries in bigdata
  clusters.
\newblock In {\em SIGMOD}, 2016.

\bibitem{learned-cardinalities}
Andreas Kipf, Thomas Kipf, Bernhard Radke, Viktor Leis, Peter Boncz, and Alfons
  Kemper.
\newblock Learned cardinalities: Estimating correlated joins with deep
  learning.
\newblock {\em arXiv:1809.00677}, 2018.

\bibitem{done-right}
Viktor Leis, Bernhard Radke, Andrey Gubichev, Alfons Kemper, and Thomas
  Neumann.
\newblock Cardinality estimation done right: Index-based join sampling.
\newblock In {\em CIDR}, 2017.

\bibitem{wander-join}
Feifei Li, Bin Wu, Ke~Yi, and Zhuoyue Zhao.
\newblock Wander join: Online aggregation via random walks.
\newblock In {\em SIGMOD}, 2016.

\bibitem{wander-join-xdb}
Feifei Li, Bin Wu, Ke~Yi, and Zhuoyue Zhao.
\newblock Wander join and {XDB:} online aggregation via random walks.
\newblock {\em TODS}, 2019.

\bibitem{ripple_join_scalable}
Gang Luo, Curt~J Ellmann, Peter~J Haas, and Jeffrey~F Naughton.
\newblock A scalable hash ripple join algorithm.
\newblock In {\em SIGMOD}, 2002.

\bibitem{mozafari_sigmod2017_invited}
Barzan Mozafari.
\newblock Approximate query engines: Commercial challenges and research
  opportunities.
\newblock In {\em SIGMOD Keynote}, 2017.

\bibitem{mozafari_sigmod2015}
Barzan Mozafari, Eugene Zhen~Ye Goh, and Dong~Young Yoon.
\newblock {CliffGuard}: A principled framework for finding robust database
  designs.
\newblock In {\em SIGMOD}, 2015.

\bibitem{approx_chapter}
Barzan Mozafari and Ning Niu.
\newblock A handbook for building an approximate query engine.
\newblock {\em {IEEE} Data Eng. Bull.}, 2015.

\bibitem{mozafari_cidr2017}
Barzan Mozafari, Jags Ramnarayan, Sudhir Menon, Yogesh Mahajan, Soubhik
  Chakraborty, Hemant Bhanawat, and Kishor Bachhav.
\newblock {SnappyData}: A unified cluster for streaming, transactions, and
  interactive analytics.
\newblock In {\em CIDR}, 2017.

\bibitem{min-wise-hashing}
Rasmus Pagh, Morten St{\"{o}}ckel, and David~P. Woodruff.
\newblock Is min-wise hashing optimal for summarizing set intersection?
\newblock In {\em PODS}, 2014.

\bibitem{online-agg-mr1}
Niketan Pansare, Vinayak~R. Borkar, Chris Jermaine, and Tyson Condie.
\newblock Online aggregation for large mapreduce jobs.
\newblock {\em PVLDB}, 4, 2011.

\bibitem{mozafari_icde2016}
Yongjoo Park, Michael Cafarella, and Barzan Mozafari.
\newblock Visualization-aware sampling for very large databases.
\newblock {\em ICDE}, 2016.

\bibitem{mozafari_sigmod2018_verdict}
Yongjoo Park, Barzan Mozafari, Joseph Sorenson, and Junhao Wang.
\newblock {VerdictDB:} universalizing approximate query processing.
\newblock In {\em SIGMOD}, 2018.

\bibitem{mozafari_sigmod2017_dbl}
Yongjoo Park, Ahmad~Shahab Tajik, Michael Cafarella, and Barzan Mozafari.
\newblock {Database Learning}: Towards a database that becomes smarter every
  time.
\newblock In {\em SIGMOD}, 2017.

\bibitem{quicksel_tr}
Yongjoo Park, Shucheng Zhong, and Barzan Mozafari.
\newblock {QuickSel}: Quick selectivity learning with mixture models.
\newblock {\em CoRR}, abs/1812.10568, 2018.

\bibitem{pitoura2008self}
Theoni Pitoura and Peter Triantafillou.
\newblock Self-join size estimation in large-scale distributed data systems.
\newblock In {\em ICDE}, 2008.

\bibitem{pf-ola}
Chengjie Qin and Florin Rusu.
\newblock Pf-ola: a high-performance framework for parallel online aggregation.
\newblock {\em Distributed and Parallel Databases}, 2013.

\bibitem{QuocABBCFS18}
Do~Le Quoc, Istemi~Ekin Akkus, Pramod Bhatotia, Spyros Blanas, Ruichuan Chen,
  Christof Fetzer, and Thorsten Strufe.
\newblock Approxjoin: Approximate distributed joins.
\newblock In {\em Proceedings of the {ACM} Symposium on Cloud Computing, SoCC
  2018, Carlsbad, CA, USA, October 11-13, 2018}, pages 426--438, 2018.

\bibitem{viz1-aditya}
Sajjadur Rahman, Maryam Aliakbarpour, Hidy Kong, Eric Blais, Karrie Karahalios,
  Aditya~G. Parameswaran, and Ronitt Rubinfeld.
\newblock I've seen "enough": Incrementally improving visualizations to support
  rapid decision making.
\newblock {\em {PVLDB}}, 2017.

\bibitem{oracle-aqp}
Hong Su, Mohamed Zait, Vladimir Barri{\`e}re, Joseph Torres, and Andre Menck.
\newblock Approximate aggregates in oracle 12c, 2016.

\bibitem{swami1994estimation}
Arun Swami and K~Bernhard Schiefer.
\newblock On the estimation of join result sizes.
\newblock In {\em EDBT}, 1994.

\bibitem{VengerovMZC15}
David Vengerov, Andre~Cavalheiro Menck, Mohamed Za{\"{\i}}t, and Sunil
  Chakkappen.
\newblock Join size estimation subject to filter conditions.
\newblock {\em {PVLDB}}, 2015.

\bibitem{cosmos}
Sai Wu, Beng~Chin Ooi, and Kian-Lee Tan.
\newblock Continuous sampling for online aggregation over multiple queries.
\newblock In {\em SIGMOD}, 2010.

\bibitem{sampling-re-optimization}
Wentao Wu, Jeffrey~F Naughton, and Harneet Singh.
\newblock Sampling-based query re-optimization.
\newblock In {\em SIGMOD}, 2016.

\bibitem{mozafari_sigmod2018_rdma}
Dong~Young Yoon, Mosharaf Chowdhury, and Barzan Mozafari.
\newblock Distributed lock management with rdma: Decentralization without
  starvation.
\newblock In {\em SIGMOD}, 2018.

\bibitem{mozafari_sigmod2014_abm}
Kai Zeng, Shi Gao, Barzan Mozafari, and Carlo Zaniolo.
\newblock The analytical bootstrap: a new method for fast error estimation in
  approximate query processing.
\newblock In {\em SIGMOD}, 2014.

\bibitem{wander-join-2018}
Zhuoyue Zhao, Robert Christensen, Feifei Li, Xiao Hu, and Ke~Yi.
\newblock Random sampling over joins revisited.
\newblock In {\em SIGMOD}, 2018.

\end{thebibliography}

\end{document}